\documentclass[11pt]{article}
\usepackage{fullpage}
\usepackage[latin2]{inputenc}
\usepackage[english]{babel}
\usepackage[cmex10]{amsmath}
\usepackage{todonotes}
\usepackage{amsthm}
\usepackage{amssymb}
\usepackage{microtype}
\usepackage{xspace}
\usepackage{todonotes}
\usepackage{comment}
\usepackage{enumerate}
\usepackage{tikz}

\usepackage{graphicx}
\usepackage{caption}
\usepackage{subcaption}

\numberwithin{equation}{section}
\newtheorem{theorem}{Theorem}[section]
\newtheorem{lemma}[theorem]{Lemma}
\newtheorem{claim}[theorem]{Claim}
\newtheorem{corollary}[theorem]{Corollary}

\theoremstyle{definition}
\newtheorem{definition}[theorem]{Definition}

\newtheorem{redrule}{Reduction rule}[section]

\newcommand{\Oh}{\mathcal{O}}
\newcommand{\Ohstar}{\Oh^\star}

\def\cqedsymbol{\ifmmode$\lrcorner$\else{\unskip\nobreak\hfil
\penalty50\hskip1em\null\nobreak\hfil$\lrcorner$
\parfillskip=0pt\finalhyphendemerits=0\endgraf}\fi} 

\newcommand{\cqed}{\renewcommand{\qed}{\cqedsymbol}}

\newcommand{\icname}{\textsc{Interval Completion}\xspace}
\newcommand{\picname}{\textsc{Proper Interval Completion}\xspace}
\newcommand{\ic}{\textsc{IC}\xspace}
\newcommand{\pic}{\textsc{PIC}\xspace}

\newcommand{\Ibeg}[1]{\alpha_{#1}}
\newcommand{\Iend}[1]{\omega_{#1}}
\newcommand{\Sbeg}[2]{\alpha_{#2}(#1)}
\newcommand{\Send}[2]{\omega_{#2}(#1)}
\newcommand{\eventssymb}{\mathcal{E}}
\newcommand{\events}[1]{\eventssymb(#1)}
\newcommand{\event}{\varepsilon}
\newcommand{\model}{\sigma}
\newcommand{\Isec}{\Omega}
\newcommand{\sol}{F}
\newcommand{\incsol}[2]{#1(#2)}
\newcommand{\incF}[1]{\incsol{\sol}{#1}}
\newcommand{\Uroot}{\mathfrak{r}}
\newcommand{\Lroot}{\mathfrak{r}_L}
\newcommand{\Rroot}{\mathfrak{r}_R}
\newcommand{\Ccomp}[1]{\mathtt{cc}(#1)}
\newcommand{\Cfam}{\mathcal{C}}
\newcommand{\Dfam}{\mathcal{D}}

\newcommand{\sectionset}{\mathcal{S}}
\newcommand{\pmcset}{\mathcal{K}}

\newcommand{\fiset}{\mathcal{F}}
\newcommand{\Tfam}{\mathcal{T}}

\newcommand{\minnei}[1]{\phi_1(#1)}
\newcommand{\maxnei}[1]{\phi_2(#1)}
\newcommand{\minneiset}[1]{\Phi_1(#1)}
\newcommand{\maxneiset}[1]{\Phi_2(#1)}
\newcommand{\ToutL}[1]{a_1(#1)}
\newcommand{\ToutR}[1]{a_2(#1)}
\newcommand{\TinL}[1]{b_1(#1)}
\newcommand{\TinR}[1]{b_2(#1)}

\newcommand{\W}{\mathbf{W}}
\newcommand{\Win}{\W^{\mathrm{in}}}
\newcommand{\Wout}{\W^{\mathrm{out}}}
\newcommand{\Wv}[1]{v(#1)}
\newcommand{\WsecL}[1]{\Isec_L(#1)}
\newcommand{\WsecR}[1]{\Isec_R(#1)}
\newcommand{\WpL}[1]{p_L(#1)}
\newcommand{\WpR}[1]{p_R(#1)}
\newcommand{\WF}[1]{F_v(#1)}
\newcommand{\WI}[1]{I(#1)}
\newcommand{\worldset}{\mathbb{W}}
\newcommand{\WW}[1]{\Gamma(#1)}

\newcommand{\Smodel}{\pi}

\newcommand{\T}{\mathbf{T}}
\newcommand{\TAI}[1]{I^1(#1)}
\newcommand{\TAsecL}[1]{\Isec_L^1(#1)}
\newcommand{\TAsecR}[1]{\Isec_R^1(#1)}
\newcommand{\TApL}[1]{p_L^1(#1)}
\newcommand{\TApR}[1]{p_R^1(#1)}
\newcommand{\TBI}[1]{I^2(#1)}
\newcommand{\TBsecL}[1]{\Isec_L^2(#1)}
\newcommand{\TBsecR}[1]{\Isec_R^2(#1)}
\newcommand{\TBpL}[1]{p_L^2(#1)}
\newcommand{\TBpR}[1]{p_R^2(#1)}

\newcommand{\stateset}{\mathbb{S}}
\newcommand{\state}{\mathbf{S}}
\newcommand{\SsecL}[1]{\Isec_L(#1)}
\newcommand{\SsecR}[1]{\Isec_R(#1)}
\newcommand{\SpL}[1]{p_L(#1)}
\newcommand{\SpR}[1]{p_R(#1)}
\newcommand{\SI}[1]{I(#1)}
\newcommand{\Sevents}[1]{\events{#1}}
\newcommand{\SW}[1]{\Gamma(#1)}

\newcommand{\terset}{\mathbb{T}}

\newcommand{\proofapp}{\ensuremath{\spadesuit}}

\title{A subexponential parameterized algorithm for\\ \textsc{Interval Completion}%
\thanks{The research leading to these results has received funding from the European Research Council under the European Union's Seventh Framework Programme (FP/2007-2013) / ERC Grant Agreement n. 267959}}
\author{
  Ivan Bliznets%
  \thanks{
    St. Petersburg Academic University of the Russian Academy of Sciences, Russia, \texttt{ivanbliznets@tut.by}.
  }
  \and
  Fedor V. Fomin%
  \thanks{
    Department of Informatics, University of Bergen, Norway, \texttt{fomin@ii.uib.no}.
  }
  \and
  Marcin Pilipczuk%
  \thanks{
    Department of Computer Science, University of Warwick, United Kingdom, \texttt{M.Pilipczuk@dcs.warwick.ac.uk}.
  }
  \and
  Micha\l{} Pilipczuk%
  \thanks{
    Faculty of Mathematics, Computer Science, and Mechanics, University of Warsaw, Poland, \texttt{michal.pilipczuk@mimuw.edu.pl}.
  }
  }

\date{}

\begin{document}

\begin{titlepage}
\def\thepage{}
\thispagestyle{empty}
\maketitle

\begin{abstract}
In the \icname{} problem we are given an $n$-vertex graph $G$ and an integer $k$,
and the task is to transform $G$ by making use of at most $k$ edge additions into an interval graph. This is a fundamental graph modification problem with applications  
in sparse matrix multiplication and molecular biology.
The question about fixed-parameter tractability of  \icname{}  was asked by 
 Kaplan, Shamir and Tarjan~[FOCS 1994; SIAM J. Comput. 1999]
and  was answered affirmatively more than a decade later
by Villanger at el.~[STOC 2007; SIAM J. Comput. 2009],
who presented an algorithm with running time $\Oh(k^{2k} n^3 m)$.
We give the first  subexponential parameterized algorithm solving 
  \icname{}   in 
time $k^{\Oh(\sqrt{k})} n^{\Oh(1)}$. This adds \icname{} to a very small list of parameterized 
graph modification problems solvable in subexponential time.

\end{abstract}
\end{titlepage}

\section{Introduction}\label{sec:intro}

In the \icname{} problem we are asked if a given graph $G$ can be complemented by at most $k$ edges into an interval graph, i.e.,
the intersection graph of intervals of the real line. This is a fundamental NP-complete problem, mentioned as problem GT35 in Garey  and Johnson \cite{GareyJ79}, arising naturally in different areas. In sparse matrix  computations the problem is equivalent to reordering columns and rows of a matrix reducing its profile \cite{Gibbs76}. In molecular biology, the problem models the task of 
building a map describing the relative position of the clones \cite{goldberg1995four,Karp93}. 
  \icname{} 
fits into the broader class of graph modification problems  on which hundreds of papers have been written.
The systematic study of the parameterized complexity of completion problems was initiated by Kaplan, Shamir, and Tarjan in \cite{focs/KaplanST94,KaplanST99}, who showed that 
\textsc{Chordal Completion}, \textsc{Strongly Chordal Completion}, and \picname{} are fixed-parameter tractable (FPT). The parameterized complexity of  \icname{}  remained open   till 2007, when Villanger  et al. \cite{HeggernesPTV07,yngve:ic}  settled this long-standing open problem by showing that the problem is FPT.  
Very recently,  Cao in \cite{yixin:ic,yixin:new} announced a single-exponential time $\Oh(6^{k}(n+m))$ algorithm.

Our main interest to  \icname{} is  due  to the new developments in parameterized complexity. 
It is well known (see e.g.~\cite{FlumGrohebook}) that
  for most of the
natural parameterized problems the existence of subexponential
parameterized algorithms can be refuted, unless the  Exponential Time Hypothesis (ETH)~\cite{ImpagliazzoPZ01} fails.
Until recently, the only notable exceptions of parameterized subexponential problems 
  were problems on special classes of graphs like planar graphs, or more generally, graphs excluding some fixed graph as a
 minor~\cite{demaine2005subexponential}, and on tournaments \cite{alon2009fast}. 
Luckily the structure of the ``parameterized subsexponential world" is  much more interesting and complicated than it was anticipated for a long time.   It appeared very recently that several graph modification problems, mostly problems of complementing to some graph class, like  
\textsc{Chordal  Completion}, \textsc{Threshold Completion},  \picname{},  and \textsc{Trivially Perfect Completion} are
solvable in subexponential time $k^{\Oh(\sqrt{k})}n^{\Oh(1)}$, where $n$ is the input size and $k$ is the number of edges in the completion~\cite{BliznetsFPP14,DrangeFPV13,fomin2011subexponential,FominV13,ghosh2012faster}. 
On the other hand, even for completion problems for a vast majority of graph classes (even very simple ones, like cographs or complements of cluster graphs), it is possible to rule out existence of subexponential parameterized algorithms~\cite{DrangeFPV13,komusiewicz2012cluster} under plausible complexity assumptions. Thus subexponential-time solvability is very unusual and exceptional property of a parameterized problem.

 While the examples of subexponential-time solvability show that some parameterized NP-hard problems are significantly ``easier" than most of the  problems from the same complexity class, we do not know  why this is the case, what the underlying difference is, and how to identify such problems. The usual ``prerequisites" for all parameterized graph modification problems solvable in subexponential time prior to this work  were that establishing membership in  FPT is easy (in most of the cases a simple branching does the job) and, moreover,  the problem is admitting a polynomial kernel.\footnote{Recall that a \emph{polynomial kernel} for a parameterized problem
  is a polynomial-time preprocessing routine that reduces an input instance $(G,k)$
  to one of size bounded polynomially in $k$, without increasing the parameter.} 
 \icname{}  absolutely does not fit into this pattern:  
 All known FPT algorithms solving this problem are quite non-trivial  \cite{yngve:ic,yixin:ic,yixin:new} (it took 13 years to make the first such algorithm) and existence of a polynomial kernel for  \icname{}  is a long time open question. This is why we find 
 the subexponential-time solvability of 
 \icname{} striking.

 Another interesting point about  \icname{} is the following. 
Completion problems 
have deep connections with width measures of graphs.
 For example, the treewidth of a graph, one of the most
fundamental graph parameters, is the minimum over all possible
completions into a chordal graph of the maximum clique size minus one. 
Similarly, the pathwidth of a graph, can be defined as  
  the minimum over all possible
completions into an  interval graph of the maximum clique size minus one.  See the survey of Bodlaender for more information on these parameters \cite{Bodlaender98}.
 Another important graph  parameter is the treedepth, also known as the
vertex ranking number, the ordered chromatic number, and the minimum
elimination tree height. This parameter appears in various settings, in particular
in the theory of sparse
graphs developed by Ne{\v{s}}et{\v{r}}il and Ossona de
Mendez~\cite{NesetrilOdM12}.  Mirroring the connection between
treewidth and chordal graphs, pathwidth and interval graphs, the treedepth of a graph can be defined
as the largest clique size in a completion to a  {trivially perfect graph}. Similarly, we may observe a relation between the class of proper interval graphs and the bandwidth of a graph,
as well as threshold graphs and the vertex cover number of a graph. (We refer for definitions of these graph classes to \cite{brandstadt1999graph}.)
Taking into account relations between these graph classes and parameters, we arrive at the diagram
presented in Fig.~\ref{fig:diagram}. It is interesting to note that all problems related to parameters 
in Fig.~\ref{fig:diagram} were established to be solvable in subexponential parameterized time \cite{BliznetsFPP14,DrangeFPV13,FominV13}. The only and the most difficult piece of the puzzle  in  Fig.~\ref{fig:diagram} remained   \icname.
 \begin{figure}[htb]
\centering
\includegraphics{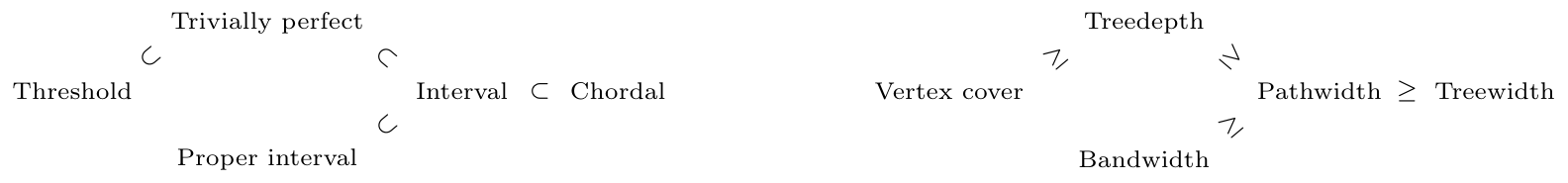}
\caption{Graph classes and corresponding graph parameters. Inequalities on the right side are with $\pm 1$ slackness.}
\label{fig:diagram}
\end{figure}

\paragraph{Our results and techniques.}
Our main result is the following theorem.
\begin{theorem}\label{thm:ic} 
\icname{}  is  solvable in time  $k^{\Oh(\sqrt{k})} \cdot n^{\Oh(1)}$.
\end{theorem}
We now describe briefly our techniques employed to prove Theorem~\ref{thm:ic}, together with the main obstacles making our approach significantly different from the approaches used for previous subexponential algorithms.

First of all,  the subexponential algorithm for \icname{} cannot be obtained by modifying  
previous algorithms  of 
  Villanger et al.  \cite{yngve:ic} and 
  Cao   \cite{yixin:ic,yixin:new}  because 
the crucial step in both
 algorithms 
 is a branching procedure that identifies a subgraph which is a witness of non-membership in the  class of interval graphs, and branches recursively on all possible ways of adding a set of edges destroying the witness. Since such a recursive branching cannot lead to time complexity better than single-exponential, this technique cannot be used in subexponential algorithm, and hence we need something completely different from what was used before.

The natural way to proceed then would be to  follow the approach which worked nicely for other completion problems:  
%
  focus on the structural definition of interval graphs
(as opposed to the definition via forbidden induced subgraphs)
and build an interval model of the output graph via dynamic programming.
The natural ``dividing'' structures in all graph classes on Fig.~\ref{fig:diagram}
are maximal cliques and clique separators, and the core part
of the known subexponential algorithms for \textsc{Chordal Completion}~\cite{FominV13}, \picname~\cite{BliznetsFPP14},
and \textsc{Trivially Perfect Completion}~\cite{DrangeFPV13} is a combinatorial
argument that bounds the number of candidates for such structures by $n^{\Oh(\sqrt{k})}$.
This, in combination with known polynomial kernels 
for these problems, yields a $k^{\Oh(\sqrt{k})}$ bound on the number
of candidates for maximal cliques and clique separators.
A second step is to design a dynamic programming algorithm whose states
are based on these structures. As the number of states is subexponential in $k$,
the entire algorithm would run in subexponential parameterized time.

There are two major problems with this approach in the case of \icname{}.
First, although we are able to provide a combinatorial bound of
$n^{\Oh(\sqrt{k})}$ reasonable candidates for maximal cliques and clique separators
in the output interval graph (see Lemma~\ref{lem:over:char}),
the existence of the second ingredient---a polynomial kernel for  \icname ---
remains a notorious open problem.
Observe that a $n^{\Oh(\sqrt{k})}$ term is unacceptable in any fixed-parameter algorithm,
not to mention a subexponential one.
To cope with this obstacle, we employ a much more insightful analysis of maximal cliques in the output interval graph,
and arrive at a (finally useful) improved $k^{\Oh(\sqrt{k})} n^8$ bound on the number of candidates.

The lack of known polynomial kernel for the problem raises also one more difficulty.
One of the more popular ``atomic operations'' in the known subexponential algorithms
 is to choose one vertex $v$ and guess \emph{all} edges from the solution incident with  it,
provided that there are at most $\sqrt{k}$ of them.
In the presence of a polynomial kernel, such a step leads to $k^{\Oh(\sqrt{k})}$ subcases---perfectly fine if we perform only a constant number of such steps.
However, in the case of  \icname{} such a step yields an (again) unacceptable $n^{\Oh(\sqrt{k})}$
term in the running time. Luckily, 
  a deep analysis of the structure of YES-instances to  \icname{} shows that
there are actually only $k^{\Oh(\sqrt{k})} n^{\Oh(1)}$ reasonable ways
to choose solution edges incident with such a ``cheap'' vertex,
making the aforementioned ``atomic operation'' possible also in our case.
Despite its triviality in the case of previous works,
it turns out that the proof of the $k^{\Oh(\sqrt{k})} n^{\Oh(1)}$ bound
is the most technical and involved part of our paper.
%

The  second major obstacle in our quest for a subexponential parameterized algorithm for \icname{}
appears when we try to develop a dynamic programming algorithm based on the knowledge
of candidates for maximal cliques and clique separators in the output interval graph.
Contrary to the case of \textsc{Chordal Completion} and \textsc{Trivially Perfect Completion},
it turns our that these structures   are far from being sufficient to design a
dynamic programming algorithm  constructing a model of the output interval graph
in a natural ``left-to-right'' manner.
The reason is that the knowledge of a clique separator $\Omega$ in the output interval graph
does not tell us much which of the components of $G \setminus \Omega$ are to the left,
and which are to the right of the separator $\Omega$ in an interval model of the output interval
graph.
(Recall that in an interval graph, each clique separator corresponds to a vertical line
that pierces intervals belonging to the separator.)
However, the knowledge which vertices of $G$ were already processed is
crucial for constructing an interval model in a ``left-to-right'' manner.

\begin{figure}[tbh]
\centering
\includegraphics{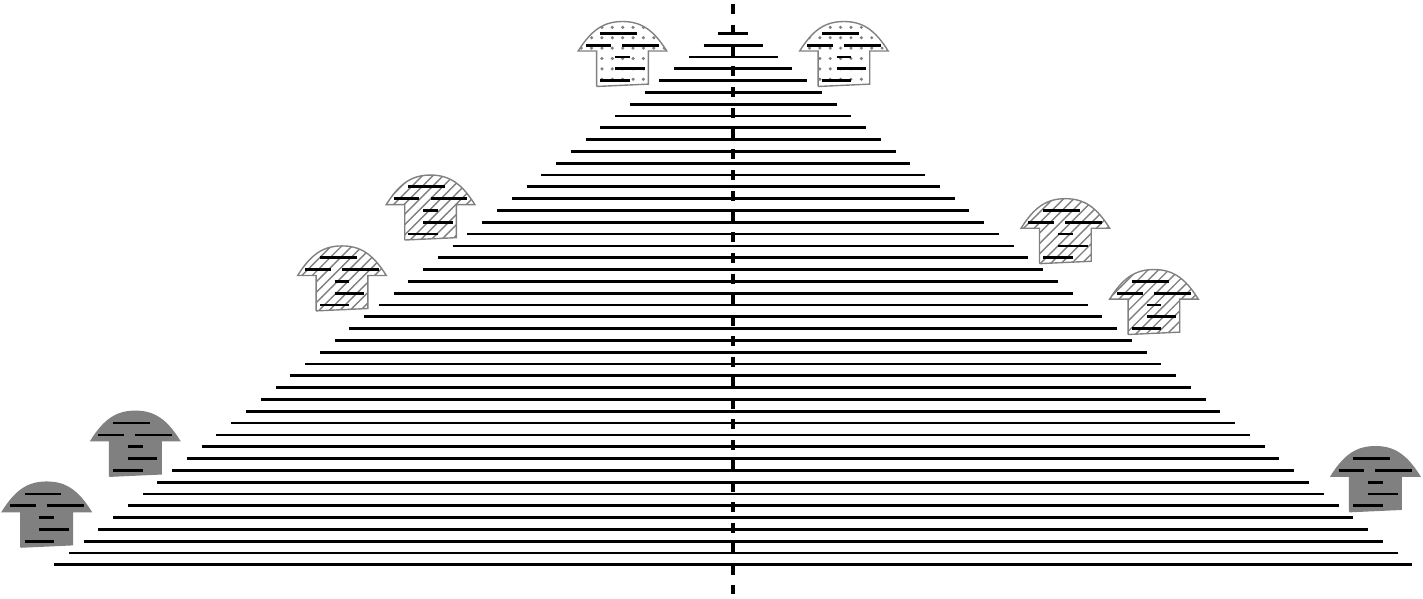}
\caption{An example of an interval graph with a large number of choices for left/right
  alignment. Within each pattern group (dotted, lined, solid), the small ``mushroom'' components
  can swap sides. A state of a dynamic programming algorithm at the middle clique marked with a dashed line
  would need to remember an alignment choice for each pattern group.}
\label{fig:piramid}
\end{figure}

An example illustrating
why it is hard to deduce the alignment of the components of $G \setminus \Omega$
for a maximal clique or clique separator $\Omega$ is depicted on Fig.~\ref{fig:piramid}.
Here, a maximal clique $\Omega$ is marked with a vertical dashed line.
The small ``mushrooms'' are components of $G \setminus \Omega$. Observe that
one can swap (take mirror image) the set of dotted mushrooms, stripped mushrooms and
solid mushrooms independently of each other. Hence, a state of a dynamic programming algorithm
needs to remember, apart from the maximal clique $\Omega$, the alignment
choice of each ``pattern'' group
of mushrooms (dotted, stripped, solid) --- and there can be many of them.

Looking at the example on Fig.~\ref{fig:piramid}, it is tempting to develop
a different dynamic programming algorithm that processes the graph in a ``top-to-bottom'' manner,
subsequently taking alignment decisions on each mushroom group, but not remembering the decision
in the state between the groups.
However, observe that if the graph locally looks as a proper interval graph (as opposed to the example
on Fig.~\ref{fig:piramid}), the ``left-to-right'' approach seems much more feasible.
Hence, to make the dynamic programming approach work in the case of \icname{},
we need to merge the ``left-to-right'' and ``top-to-bottom'' approaches, arriving at a quite technical 
definition of an actual state of dynamic programming.

A short comparison with the algorithm for seemingly similar \picname{} (\pic{} for short) is in place.
Although both algorithms follow the same general approach paved by Fomin and Villanger~\cite{FominV13},
the actual difficulties, and methods to avoid them, are completely different.
First, in the \pic{} case a polynomial kernel is known~\cite{pic-kernel}, and a subexponential bound
on both the number of candidates for maximal cliques $\Omega$, and on the number of left/right choices for $G-\Omega$,
are not trivial, but relatively simple. The main difficulty in the \pic{} case lies in the fact that this information is not sufficient
to perform a natural left-to-right dynamic programming, as one needs to ensure that no interval contains another
in the output model; an issue non-existent in the interval case.
To cope with this obstacle, in~\cite{BliznetsFPP14} the dynamic programming structure is also reengineered,
but not only for a completely different reason than here, and also in a completely different manner --- loosely speaking,
apart from maximal cliques, the algorithm of~\cite{BliznetsFPP14} uses a type of separation similar to the classic
$\Ohstar(10^n)$ exact algorithm for bandwidth of Feige~\cite{Feige00}.

\paragraph{Organisation of the paper.}
We first introduce notation and preliminary results in Section~\ref{sec:prelims},
and give a more detailed, yet still informal overview of the proof of Theorem~\ref{thm:ic}
in Section~\ref{sec:overview}.

Then, in Sections~\ref{sec:neighbors}--\ref{sec:dp}, we provide a full proof of Theorem~\ref{thm:ic}.
Section~\ref{sec:neighbors} describes a module-based reduction rule
and introduces some auxiliary results on neighborhood classes in a (near) interval graph.
In Section~\ref{sec:pmc} we prove the subexponential bound on the number of candidates
for \emph{sections}, a technical notion close to a clique separator.
In Section~\ref{sec:fill-in} we provide a bound of $k^{\Oh(\sqrt{k})} n^{\Oh(1)}$ reasonable ways to add solution
edges incident to one vertex, provided that there are at most $\sqrt{k}$ of them.
After one additional combinatorial lemma in Section~\ref{sec:left-right}, we describe
the final dynamic programming algorithm in Section~\ref{sec:dp}.

Section~\ref{sec:conc} concludes the paper and suggests directions of future research.

\section{Preliminaries}\label{sec:prelims}
\paragraph{Graph notation.} In most cases, we follow standard graph notation.
For a graph $G$, by $\Ccomp{G}$ we denote the family of vertex sets of connected components of $G$.
For a path $P$ and two vertices $x,y \in V(P)$, by $P[x,y]$ we denote the subpath of $P$ between $x$ and $y$, inclusive.  For   vertex $v$, we use  
$N_G(v)$  and   $N_G[v]$ to denote the open and the closed neighborhood of $v$.
For a vertex set $S\subseteq V$ we denote by $N_G(S)$ the set $\bigcup_{v \in S} N_G(v)\setminus S$.

For any graph $G$ we shall speak about, we implicitly
fix some arbitrary total ordering $\prec$ on $V(G)$.
We shall use this ordering to break ties and canonize some objects
(interval models, completion sets, solutions, etc.). Such a canonization
will turn out to be helpful when handling greedy arguments
in the final dynamic programming routine.

\paragraph{Interval graphs.}
A graph $G$ is an \emph{interval graph} if it admits an intersection model of the following form:
each vertex is assigned a closed interval on a line, and two vertices are adjacent if and only if
their intervals intersect.

We formalize the notion of a model in the following combinatorial
way. For each $v \in V(G)$ we create two symbols $\Ibeg{v}$ and $\Iend{v}$, called henceforth
\emph{events}, and denote $\events{X} = \bigcup_{v \in X} \{\Ibeg{v}, \Iend{v}\}$ for
any $X \subseteq V(G)$. 
An \emph{interval model} is a permutation (bijection) $\model:\events{V(G)} \to \{1,2,\ldots,2n\}$
such that:
\begin{enumerate}
\item for each $v \in V(G)$ we have $\model(\Ibeg{v}) < \model(\Iend{v})$ (an interval starts
before it ends), and
\item for each $u, v \in V(G)$ we have $uv \notin E(G)$ if and only if
$\model(\Iend{v}) < \model(\Ibeg{u})$ or $\model(\Iend{u}) < \model(\Ibeg{v})$
(vertices are nonadjacent if and only if their intervals are disjoint).
\end{enumerate}
The numbers $1,2,\ldots,2n$ in the codomain of a model $\model$ are called \emph{positions}.

Informally speaking, the aforementioned combinatorial notion of an interval model corresponds
to a ``real'' model, where no two endpoints of intervals coincide (which we can assume
without loss of generality).
The permutation $\model$ corresponds to the order of endpoints of intervals: $\Ibeg{v}$ represents
the starting (left) endpoint of the interval associated with $v$, and $\Iend{v}$ represents
the ending (right) endpoint. See Figure~\ref{fig:model} for an example.

\begin{figure}
\centering
\includegraphics{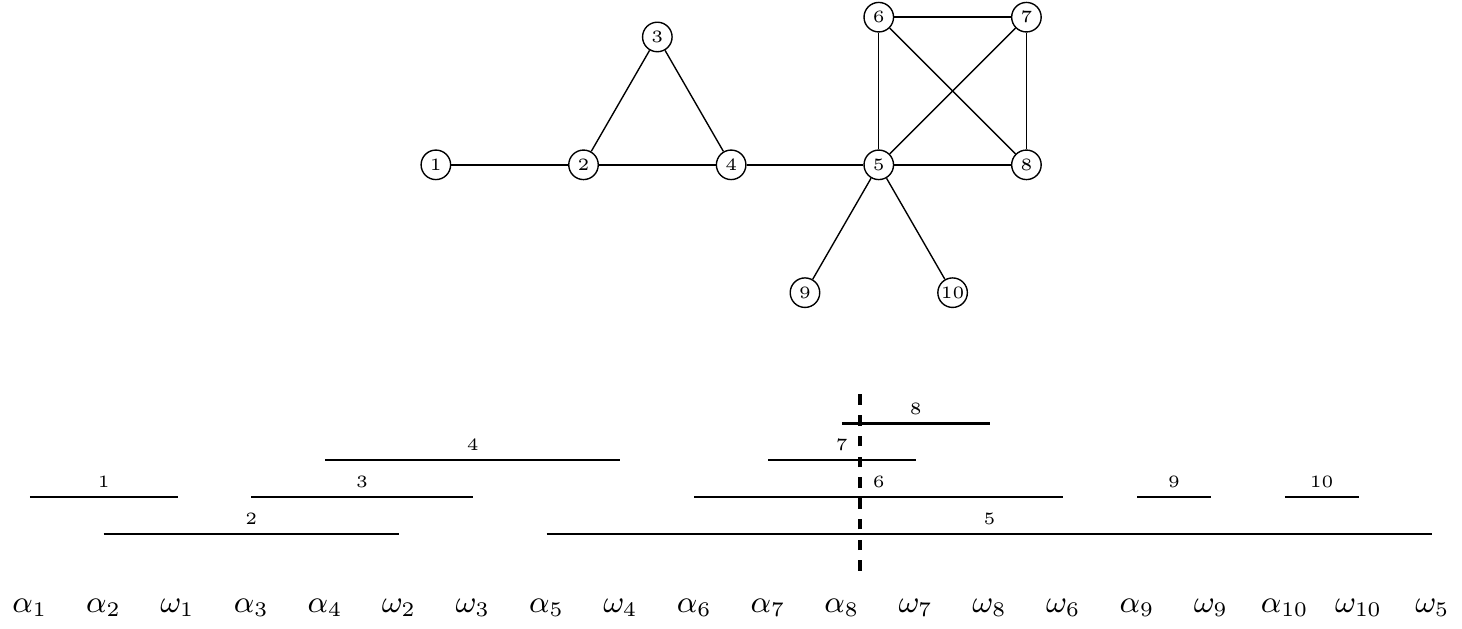}
\caption{An example of a graph with an interval model and its combinatorial representation.
  The vertical dashed line represents one of the maximal cliques of the graph, being section $\Omega_\sigma(12)$.
  We remark that this is not the canonical model of the represented graph (assuming the natural order on the vertex labels):
    for the canonical model, one should swap events $\omega_2$ with $\omega_3$ and $\omega_7$ with $\omega_8$.
}
\label{fig:model}
\end{figure}

Given an interval model $\model$ of a graph $G$, we say that an event $\event_1$ is \emph{before}
or \emph{to the left} of an event $\event_2$ iff $\model(\event_1) < \model(\event_2)$.
In this situation we also say that $\event_2$ is \emph{later} or \emph{to the right}
of $\event_1$.

For an interval model $\model$ of a graph $G$ and a set $X \subseteq V(G)$, we denote
by $\Sbeg{X}{\model}$ and $\Send{X}{\model}$, respectively, the first and last positions where
events of $\events{X}$ appear in $\model$.

For an interval model $\model$ of a graph $G$ and an integer $p$, the set
$$\Isec_\model(p) = \{v \in V(G): \model(\Ibeg{v}) \leq p < \model(\Iend{v})\}$$
is called a \emph{section at position $p$}. By somehow abusing the notation, for an event $\event$
we write $\Isec_\model(\event)$ for $\Isec_\model(\model(\event))$, and call it
a \emph{section at event $\event$}. We omit the subscript if it is clear from the context.
Note that every  section is  a clique in $G$.

Intuitively speaking, a section is a set of vertices whose intervals become ``pinned down" by a vertical line drawn {\emph{just after}} event $\model^{-1}(p)$, see Figure~\ref{fig:model}. Thus, all these intervals share a common point, so they are pairwise adjacent in the graph.

We refer to an inclusion-wise maximal clique of a graph $G$ as to  a maximal clique. 
It is well-known~\cite{Golumbic80}
that $\Isec \subseteq V(G)$ is a maximal clique in an interval graph $G$ with model $\model$
if and only if it is a section drawn between a starting and ending event: there exists $v_1,v_2 \in V(G)$ (possibly $v_1 = v_2$) such that
$\Isec = \Isec_\model(\Ibeg{v_2})$ and $\model(\Ibeg{v_2})+1 = \model(\Iend{v_1})$.

We also use the following notions of maximality and minimality in interval models.
Let $X \subseteq V(G)$,
where $G$ is an interval graph with a fixed model $\model$.
We say that $v \in X$ is \emph{interval-maximal} in $X$ (w.r.t. $\model$)
if for no other $w \in X$ it holds that $\model(\Ibeg{w}) < \model(\Ibeg{v}) < \model(\Iend{v}) < \model(\Iend{w})$.
Analogously, $v \in X$ is \emph{interval-minimal} in $X$ (w.r.t. $\model$)
if for no other $w \in X$ it holds that $\model(\Ibeg{v}) < \model(\Ibeg{w}) < \model(\Iend{w}) < \model(\Iend{v})$.
Clearly, each non-empty set of vertices has an interval-maximal and interval-minimal vertex,
but these vertices may not be defined uniquely.

We recall that in linear time 
we can check if a given graph $G$ is an interval graph, and if this is the case,
find an interval model of $G$~\cite{Golumbic80}.
In our work we will need a slightly stronger statement.\footnote{Proofs marked with \proofapp{} are straightforward, and have been moved to the appendix in order not to disturb the flow
  of the arguments.}
\begin{lemma}[\proofapp{}]\label{lem:ic-cliques-drawing}
Given an interval  graph $G$ and two cliques  $\Omega_1,\Omega_2 \subseteq V(G)$,
one can in polynomial time check whether there exists an interval model of $G$ that starts with
all starting events of $\events{\Omega_1}$
and ends with all ending events of $\events{\Omega_2}$.
\end{lemma}

For the final dynamic programming routine, we need to  ``canonize" a model of an interval
graph $G$. Recall that we have fixed a total order $\prec$ on $V(G)$;
assume $V(G) = \{v_1,v_2,\ldots,v_n\}$ where $v_1 \prec v_2 \prec \ldots \prec v_n$.
For a model $\model$ of $G$, we consider a tuple
$$(\model(\Ibeg{v_1}),\model(\Ibeg{v_2}),\ldots,\model(\Ibeg{v_n}),\model(\Iend{v_n}),\model(\Iend{v_{n-1}}),\ldots,\model(\Iend{v_1}))$$
and define a \emph{canonical model} of $G$ to be the model with the aforementioned
tuple being lexicographically minimum among all models of $G$.

We note two properties of a canonical model $\model$ that are of our interest.
The first one is straightforward.
\begin{lemma}\label{lem:can:endpoints}
Assume $\model$ is the canonical model of an interval graph $G$.
Then, for each $u,v \in V(G)$, if $\model(\Ibeg{u}) + 1 = \model(\Ibeg{v})$
then $u \prec v$ and if $\model(\Iend{u}) + 1 = \model(\Iend{v})$ then $u \succ v$.
That is,
the canonical model orders consecutive starting/ending points of the intervals according to $\prec$.
\end{lemma}

The second one says that canonizing a model fixes an order in which modules with the same
neighborhood appear in the model.
\begin{lemma}[\proofapp]\label{lem:can:modules}
Let $\model$ be  the canonical model of an interval graph $G$.
Let $X \subseteq V(G)$ be   a clique,
and let $C_1,C_2,\ldots,C_s$ be components of $G \setminus X$ (not necessarily all of them)
such that $N_G(v) \setminus C_i = X$ for every $1 \leq i \leq s$ and every $v \in C_i$.
Since the components $C_i$ are pairwise nonadjacent, $\Send{C_i}{\model} < \Sbeg{C_j}{\model}$
or $\Send{C_j}{\model} < \Sbeg{C_i}{\model}$ for any $i \neq j$.
Without loss of generality, assume that
$$\Sbeg{C_1}{\model} < \Send{C_1}{\model} < \Sbeg{C_2}{\model} < \Send{C_2}{\model} < \cdots < \Sbeg{C_s}{\model} < \Send{C_s}{\model}.$$
For each $1 \leq i \leq s$, let $x_i \in C_i$ be the first vertex of $C_i$ in the order $\prec$.
Then
$$x_1 \prec x_2 \prec \cdots \prec x_s.$$
That is, $\model$ sorts the components $C_i$ according to the order of their $\prec$-minimum vertices.
\end{lemma}

\paragraph{Interval completion.}
For a graph $G$, a \emph{completion} of $G$ is a set $\sol \subseteq \binom{V(G)}{2} \setminus E(G)$ such that $G+\sol := (V(G),E(G) \cup \sol)$ is an interval graph.
A completion is \emph{minimal} if it is inclusion-wise minimal, and \emph{minimum}
if it has minimum possible cardinality.
In the \textsc{Interval Completion} problem the input consists of a graph $G$ and an integer $k$,
and we ask for a completion of $G$ of size at most $k$.
For an instance $(G,k)$ of \textsc{Interval Completion}, a completion of cardinality at most $k$
is called a \emph{solution}. The notions of \emph{minimal } and \emph{minimum solutions}
are defined naturally.

For a completion $F$ in a graph $G$, we say that $v$ is \emph{touched} by $F$
if there is an edge in $F$ incident with $v$; otherwise $v$ is \emph{untouched}.
A set of vertices $X$ is \emph{touched} if it contains a touched vertex, and \emph{untouched}
otherwise.
We also say that a vertex $v \in V(G)$ is \emph{cheap} (with respect to the completion $F$)
if at most $\sqrt{k}$ edges of $F$ are incident with $v$; a vertex is \emph{expensive}
if it is not cheap. Note that there are at most $2k$ touched vertices and at most
$2\sqrt{k}$ expensive ones.
For a completion $F$ and a vertex $v \in V(G)$, by $\incF{v}$ we denote the set of edges
$e \in F$ that are incident with $v$.

We now canonize solutions $F$ to an \icname{} instance $(G,k)$.
Given a partial order $\prec$ on a finite set $U$, we  define a partial order
on the family of subsets of $U$ as follows: if $A,B \subseteq U$, then
we first sort the elements of $A$ and $B$ according to $\prec$, and then compare the
obtained sequences lexicographically.
By somehow abusing the notation, we denote by $\prec$ the imposed order on the subsets of $U$
as well.

This definition automatically extends the partial order $\prec$ on $V(G)$
first onto $\binom{V(G)}{2}$, and then onto the family of completions of $G$.
We define the canonical solution to $(G,k)$ to be the minimum solution in the order $\prec$
among all minimum solutions to $(G,k)$.

Given an instance $(G,k)$ of \textsc{Interval Completion},
      we start with augmenting it in the following
way. We add a universal vertex $\Uroot$ adjacent to all vertices of $V(G)$, and two vertices
$\Lroot$ and $\Rroot$, adjacent only to $\Uroot$, obtaining a graph $G'$.
 We assume $\Uroot \prec \Lroot \prec v \prec \Rroot$ for any $v \in V(G)$. 
Note that for any completion $F$ of $G$, $F$ is also a completion of $G'$: given a model
of $G+F$, we may construct a model of $G'+F$ by preceding the events of $\events{V(G)}$
with $\Ibeg{\Uroot}, \Ibeg{\Lroot}, \Iend{\Lroot}$ and succeeding them
with $\Ibeg{\Rroot}, \Iend{\Rroot}, \Iend{\Uroot}$.
Consequently, in every minimal completion of $G'$, the vertices $\Uroot$, $\Lroot$ and $\Rroot$
are untouched.
Thus, henceforth we assume that, whenever we consider an instance $(G,k)$ to \textsc{Interval Completion},
$G$ already contains vertices $\Uroot$, $\Lroot$ and $\Rroot$.
By Lemmata~\ref{lem:can:endpoints} and~\ref{lem:can:modules} (applied to $X = \{\Uroot\}$),
the canonical model of any completion of $G$ starts
with $\Ibeg{\Uroot}, \Ibeg{\Lroot}, \Iend{\Lroot}$ and ends
with $\Ibeg{\Rroot}, \Iend{\Rroot}, \Iend{\Uroot}$.

A short informal rationale for this augmentation is that in some places of the algorithm we would
like to pick the ``first/last untouched vertex whose interval ends/starts after/before position $p$''
or ``an untouched vertex whose interval contains the interval of $v$''; note that $\Lroot$/$\Rroot$
is always a good candidate for the first choice, and $\Uroot$ for the second one.

\section{Overview of the algorithm}\label{sec:overview}
 
In this section we provide an informal overview on the proof of Theorem~\ref{thm:ic}.

\subsection{Module Reduction Rule}

We start with a simple module-based reduction rule. 
Recall that $M \subseteq V(G)$ is a \emph{module} in a graph $G$ if
$N(v_1) \setminus M = N(v_2) \setminus M$ for any $v_1,v_2 \in M$.
(Equivalently, for any $v \notin M$ we have either $M \subseteq N(v)$ or $M \cap N(v) = \emptyset$.)

Assume that in a YES-instance $(G,k)$ of \icname{} we have recognized a set $X \subseteq V(G)$
such that many (significantly more than $2k$) connected components $M_1,M_2,\ldots,M_r$
of $G \setminus X$ are 
modules, fully adjacent to $X$.
Then it is easy to observe that any solution $\sol$ to $(G,k)$ needs to yield
an ordering $\model$ of $G+\sol$ similar to the one
depicted on Figure~\ref{fig:over:module}:
$X$ becomes a clique, and most of the components $M_i$ are drawn one after another
on the ``plateau'' formed by all the intervals of the vertices of $X$. 
Moreover, note that all but at most $k$ components $M_i$ need to induce interval graphs,
and all but at most $2k$ components $M_i$ are left untouched by the solution $\sol$.

\begin{figure}
\centering
\includegraphics{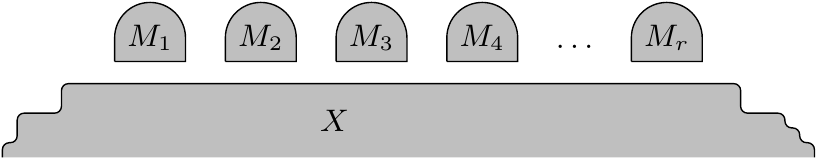}
\caption{The alignment of most of the components $M_i$ in 
  the model $\model$ of the interval graph $G+\sol$.}
\label{fig:over:module}
\end{figure}

However, if $r \geq 2k+2$, then there are at least two such untouched $M_i$s --- say $M_1$ and $M_2$ ---
and, in the interval graph $G+\sol$ they force $X$ to be a clique, 
reserving space between $M_1$ and $M_2$ for any other $M_i$ with $G[M_i]$ being an interval graph.
Thus, we may reduce the number of such $M_i$s to $2k+2$, without changing the answer
to the instance $(G,k)$.

\begin{redrule}[Module Reduction Rule]
Let $(G,k)$ be an instance of \textsc{Interval Completion}.
Assume there exists $X \subseteq V(G)$ and connected components
$M_1,M_2,\ldots,M_{2k+3}$ of $G \setminus X$ that are modules in $G$ and, moreover,
$N(M_i) = N(M_1)$ for each $1 \leq i \leq 2k+3$.
Then proceed as follows. If for more than $k$ indices $i$ the subgraph $G[M_i]$
is not an interval graph, return that $(G,k)$ is a NO-instance.
Otherwise, pick arbitrary $j$ such that $G[M_j]$ is an interval graph and remove $M_j$
from $G$.
\end{redrule}

We remark here that the Module Reduction Rule can be applied exhaustively in polynomial
time, using the module decomposition of the graph $G$:
It is easy to observe that, if the rule is applicable, then all components $M_i$
are children of a single \emph{union} node in the module decomposition tree.

\begin{figure}
\centering
\includegraphics{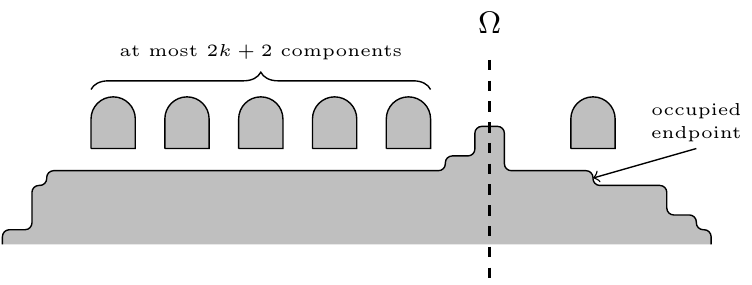}
\caption{Motivation for the Module Reduction Rule.}
\label{fig:over:module-apply}
\end{figure}

Let us now explain our motivation for introducing the Module Reduction Rule
(see also Figure~\ref{fig:over:module-apply}).
In many steps of the algorithm, we analyse some clique $\Isec$ of the interval
graph $G+\sol$, and we would like to control the number of  connected components of
$G \setminus \Isec$.
There are two types of such components: the ones that are modules, and the ones that are not modules.
If Module Reduction Rule has been applied exhaustively, then we have a bound
on the number of components of the first type
for a fixed neighborhood
$X \subseteq \Isec$; observe that there are only $2(|\Isec|+1)$ choices for such neighborhood.
For a component $C$ that is not a module, with vertices $v_1,v_2 \in C$ s.t.
$N(v_1) \setminus C \neq N(v_2) \setminus C$,
observe that either $C$ is touched by the solution $\sol$
or $C$ ``occupies'', in the interval model of $G+\sol$,
an endpoint event of every vertex of $(N(v_1) \triangle N(v_2)) \setminus C$.
Consequently, there are at most $2k+ 2|\Isec|$ components of the second type.

\subsection{Candidates for sections and maximal cliques}

Our first milestone combinatorial result is the following:

\begin{theorem}\label{thm:over:sections}
Given an \textsc{Interval Completion} instance $(G,k)$, where the Module Reduction Rule is not applicable, one
can in $k^{\Oh(\sqrt{k})} n^{\Oh(1)}$ time enumerate
a family $\sectionset$ of $k^{\Oh(\sqrt{k})} n^{17}$ subsets
of $V(G)$, such that for any minimal solution $F$ to $(G,k)$,
in the canonical model $\model$ of $G+F$
all sections of $\model$ belong to $\sectionset$.
\end{theorem}

As an intermediate step, we provide an enumeration algorithm for potential maximal cliques
in the \textsc{Interval Completion} problem, showing the following.

\begin{theorem}\label{thm:over:pmcs}
Given an \textsc{Interval Completion} instance $(G,k)$, where the Module Reduction Rule is not applicable, one
can in $k^{\Oh(\sqrt{k})} n^{\Oh(1)}$ time enumerate
a family $\pmcset$ of $k^{\Oh(\sqrt{k})} n^{8}$ subsets
of $V(G)$, such that for any minimal solution $F$ to $(G,k)$,
all maximal cliques of $G+F$ belong to $\pmcset$.
\end{theorem}

It is not hard to see that Theorem~\ref{thm:over:pmcs} implies Theorem~\ref{thm:over:sections}.
\begin{proof}[Proof of Theorem~\ref{thm:over:sections}]
Let $(G,k)$ be an \textsc{Interval Completion} instance, $F$ be a minimal solution to $(G,k)$ with $\model$ being the canonical model of $G+F$.
Clearly, $\emptyset$, $\{\Uroot\}$, $\{\Uroot,\Lroot\}$ and $\{\Uroot,\Rroot\}$ are sections of $\model$;
we include them into $\sectionset$ at the beginning.

Let $\Isec_\model(p)$ be a section of $\model$. Without loss of generality, assume that
$\Isec_\model(p)$ 
is not one  of the four aforementioned ``obvious'' sections.
Let $p_1 \leq p$ be the largest integer such that $\Isec_\model(p_1)$ is a maximal clique
of $G+F$; such $p_1$ always exists as $p_1=2$ with $\Isec_\model(2) = \{\Uroot,\Lroot\}$
is a candidate value. Symmetrically, we define $p_2$ to be the smallest integer with
$p_2 \geq p$ such that $\Isec_\model(p_2)$ is a maximal clique of $G+F$.

Let $r = |\Isec_\model(p_1) \setminus \Isec_\model(p_2)|$.
We infer that 
$\model$ places events of $\{\Iend{v}: v \in \Isec_\model(p_1) \setminus \Isec_\model(p_2)\}$ on positions $p_1+1, p_1+2,\ldots,p_1+r$, and then it places events of $\{\Ibeg{v}: v \in \Isec_\model(p_2) \setminus \Isec_\model(p_1)\}$ on positions $p_1+r+1,p_1+r+2,\ldots,p_2$; otherwise there would be a section between sections $\Isec_\model(p_1)$ and $\Isec_\model(p_2)$ that would yield a maximal clique, contradicting the choice of $p_1$ or of $p_2$.
Moreover, by Lemma~\ref{lem:can:endpoints} the events of $\{\Iend{v}: v \in \Isec_\model(p_1) \setminus \Isec_\model(p_2)\}$ are sorted according to the reversed total order $\prec$, while the events of $\{\Ibeg{v}: v \in \Isec_\model(p_2) \setminus \Isec_\model(p_1)\}$ are sorted according to the total order $\prec$.
Consequently, the set $\Isec_\model(p)$ can be deduced from
the maximal cliques $\Isec_\model(p_1)$ and $\Isec_\model(p_2)$ (both belonging
to the set $\pmcset$ given by Theorem~\ref{thm:over:pmcs})
and the value of $p-p_1$, for which we have $n+1$ choices. Theorem~\ref{thm:over:sections} follows.
\end{proof}

Hence, we now sketch the proof of Theorem~\ref{thm:over:pmcs}.
We first start with an $n^{\Oh(\sqrt{k})}$ bound, and then argue how to obtain
the actual FPT bound of Theorem~\ref{thm:over:pmcs}.

Let us fix an \textsc{Interval Completion} instance $(G,k)$, its minimal solution $F$,
a model $\model$ of $G+F$ and a maximal clique $\Isec = \Isec_\model(p)$. Recall
that $\model(\Ibeg{v_2}) = p$ and $\model(\Iend{v_1}) = p+1$ for some vertices $v_1$ and $v_2$.
Without loss of generality, assume that $\Isec$ is different than two ``obvious'' maximal
cliques $\{\Uroot,\Lroot\}$ and $\{\Uroot,\Rroot\}$ and, consequently,
$3 < p < 2n-3$ and $v_1,v_2 \notin \{\Uroot,\Lroot,\Rroot\}$.

\begin{figure}
\centering
\includegraphics{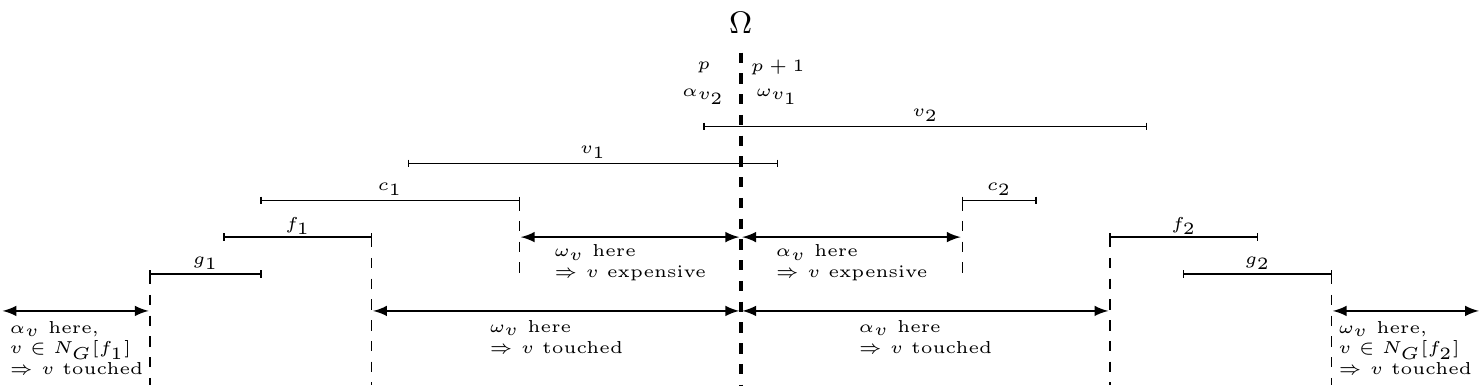}
\caption{The anatomy of a maximal clique $\Omega$, with eight important vertices
  guessed by the algorithm.}
\label{fig:over:eight}
\end{figure}

We define the following vertices (see also Figure~\ref{fig:over:eight}):
\begin{enumerate}
\item $c_1$ is the cheap vertex with the rightmost position of $\Iend{c_1}$, among
all cheap vertices $c$ satisfying $\model(\Iend{c}) \leq \model(\Iend{v_1}) = p+1$;
\item $c_2$ is the cheap vertex with the  leftmost position of $\Ibeg{c_2}$, among
all cheap vertices $c$ satisfying $\model(\Ibeg{c}) \geq \model(\Ibeg{v_2}) = p$;
\item $f_1$ is the untouched vertex with the rightmost position of $\Iend{f_1}$, among
all untouched vertices $f$ satisfying $\model(\Iend{f}) \leq \model(\Iend{v_1}) = p+1$;
\item $f_2$ is the untouched vertex with the leftmost position of $\Ibeg{f_2}$, among
all  untouched vertices $f$ satisfying $\model(\Ibeg{f}) \geq \model(\Ibeg{v_2}) = p$;
\item $g_1$ is the untouched vertex with the leftmost position of $\Ibeg{g_1}$, among
all untouched vertices of $N_G[f_1] \setminus \{\Isec \setminus \{v_1\}\}$;
\item $g_2$ is the untouched vertex with the rightmost position of $\Iend{g_2}$, among
all untouched vertices of $N_G[f_2] \setminus \{\Isec \setminus \{v_2\}\}$.
\end{enumerate}
Let us remark that some of these vertices can be in fact equal. We also remark that all quantifications in the aforementioned definitions are done on nonempty sets:
$\Lroot$ is a good candidate for both $c_1$ and $f_1$, $\Rroot$ is a good candidate for both
$c_2$ and $f_2$, $f_1$ is a good candidate for $g_1$ and $f_2$ is a good candidate for $g_2$.
Hence, all these vertices are well-defined.

Also, observe that $v_1\in N_G[v_2]$ and $v_2\in N_G[v_1]$, as otherwise
$v_1v_2 \in \sol$ and, by swapping the events $\Iend{v_1}$ and $\Ibeg{v_2}$
in the model $\model$, we obtain a model for $G+(\sol \setminus \{v_1v_2\})$,
contradicting the minimality of $\sol$.

We say that a vertex $v$ \emph{lies to the left} of the clique $\Isec$ 
if $\model(\Iend{v}) \leq p+1$, and \emph{lies to the right} if $\model(\Ibeg{v}) \geq p$.
Clearly, $v_1,c_1,f_1,g_1$ lie to the left of $\Isec$ and
$v_2,c_2,f_2,g_2$ lie to the right of $\Isec$.
Note that, perhaps a bit counter-intuitively, if $v=v_1=v_2$, then $v$ lies both to the left and
to the right of $\Isec$.

Let $w$ be any vertex of the graph. Observe that
if some vertex of $N_{G+F}[w]$ lies to the left of $\Isec$, then $\model(\Ibeg{w}) \leq p$.
Similarly,
  if some vertex of $N_{G+F}[w]$ lies to the right of $\Isec$, then $\model(\Iend{w}) \geq p+1$.
In particular, if both these events happen, then $w$ belongs to $\Isec$.

Define now the following sets.
\begin{align*}
F_i^\circ &= \{v \in V(G): vc_i \in F\}\ \mathrm{for}\ i = 1,2; \\
X_1^\circ &= \{v \in V(G): \model(\Iend{c_1}) < \model(\Iend{v}) \leq p+1\}; \\
X_2^\circ &= \{v \in V(G): p \leq \model(\Ibeg{v}) < \model(\Ibeg{c_2})\}.
\end{align*}
As $c_1$ and $c_2$ are cheap, $|F_1^\circ|, |F_2^\circ| \leq \sqrt{k}$.
By the definition of $c_1$ and $c_2$, all vertices of $X_1^\circ \cup X_2^\circ$ are expensive. 
Note that $|X_1^\circ \cap X_2^\circ| \leq 1$ and $X_1^\circ \cap X_2^\circ$ is nonempty only if it consists
of $v_1=v_2$. Therefore $|X_1^\circ| + |X_2^\circ| \leq 2\sqrt{k}+1$.

We now show the main combinatorial observation: the knowledge of vertices $v_1,v_2,c_1,c_2$
and sets $F_i^\circ$ and $X_i^\circ$ for $i=1,2$ already uniquely defines the 
clique $\Isec$.
\begin{lemma}\label{lem:over:char}
$$\Isec = (N_G[\{v_1,c_1\} \cup X_1^\circ] \cup F_1^\circ)\cap (N_G[\{v_2,c_2\} \cup X_2^\circ] \cup F_2^\circ).$$
\end{lemma}
\begin{proof}
The inclusion ``$\supseteq$'' is immediate from the previous discussion: 
every vertex  $v\in N_G[\{v_1,c_1\} \cup X_1^\circ] \cup F_1^\circ$ is either 
to the left of $\Isec$ in $G+F$, or at least one  neighbor of $v$ is  to the left of $\Isec$. Similarly, for every 
$u\in N_G[\{v_2,c_2\} \cup X_2^\circ] \cup F_2^\circ$, at least one vertex from $N_{G+F}[u]$ is to the right of $\Isec$ in $G+F$. Hence, we now focus
 on the other inclusion.


Without loss of generality, assume there exists a vertex $v \in \Isec$ that does not belong
to $F_2^\circ$ nor to $N_G[\{v_2,c_2\} \cup X_2^\circ]$. In particular $v \notin \{v_1,v_2,c_2\}$,
   and hence $\Ibeg{v}<p$.
As $v \notin F_2^\circ$ and $vc_2 \notin E(G)$, we have $\model(\Iend{v}) < \model(\Ibeg{c_2})$.
Moreover, by the definition of $X_2^\circ$, $v$ is not adjacent in $G$ to any vertex whose starting event
lies between positions $p$ and $\model(\Ibeg{c_2}) - 1$.
Hence, $v$ is not adjacent in $G$ to any vertex whose starting event lies on or after position $p$.

Consider an ordering $\model'$ that is created from the model $\model$ by moving the event
$\Iend{v}$ to the position just before the event $\Ibeg{v_2}$ (that is, we move $\Iend{v}$
to the position $p$ and shift all events on positions $p$ and later by one to the right).
By our previous arguments, $\model'$ is a valid interval model of some completion
$F'$ of $G$. As $v \in \Isec$, the event $\Iend{v}$ has been moved to the left during this operation,
and $F' \subseteq F$. Moreover $vv_2 \in F \setminus F'$, which contradicts the minimality of $F$.
\end{proof}
As the sets $F_i^\circ$ and $X_i^\circ$ are of size $\Oh(\sqrt{k})$,
Lemma~\ref{lem:over:char} already gives us an $n^{\Oh(\sqrt{k})}$ bound on the number
of candidates for maximal cliques in $G+\sol$.
However, in the absence of polynomial kernel for \icname{}, we need to work further
to obtain the bound promised in Theorem~\ref{thm:over:pmcs}. In this quest we will make
use of the vertices $f_i$ and $g_i$.

The choice of vertices $v_i,c_i,f_i$ and $g_i$ for $i=1,2$ contributes with factor $n^8$ to the
bound of Theorem~\ref{thm:over:pmcs}; our goal is to produce $k^{\Oh(\sqrt{k})}$ candidates
for a fixed choice of these eight vertices. To this end, we develop a branching algorithm
that maintains a choice of candidate sets $X_1,X_2,F_1,F_2$ for $X_1^\circ$, $X_2^\circ$,
$F_1^\circ$ and $F_2^\circ$, respectively, and a guess $K$ on the clique $\Isec$.
At each step of the recursion, the algorithm
outputs the current set $K$ as a possible choice, and branches into $k^{\Oh(1)}$ number of subcases,
choosing one additional vertex to include into one of the sets $X_i$ or $F_i$, updating
$K$ accordingly%
\footnote{This statement is not completely true, in some cases we are able only to guess
a \emph{neighborhood} of a vertex in $X_i^\circ$, without indicating the vertex itself.
However, this is sufficient for the purpose of the reasoning of Lemma~\ref{lem:over:char}.}.
As the depth can be bounded by $\Oh(\sqrt{k})$, we obtain the promised bound
of $k^{\Oh(\sqrt{k})}$ candidates for the clique $\Isec$.

Obviously, the main technical difficulty lies in the argumentation that there are only
$k^{\Oh(1)}$ reasonable choices in each step of the recursion.
Here the guess on the vertices $f_i$ and $g_i$ help: we carefully analyse the structure
of connected components of $G \setminus (X_1 \cup X_2 \cup K \cup \{v_1,v_2,c_1,c_2,f_1,f_2,g_1,g_2\})$ and argue that only a limited number of vertices may possibly live between $f_1$ and $f_2$
in the model $\model$ of $G+\sol$. Moreover, in this argument we heavily rely on the fact that the Module Reduction Rule is not applicable, which in various places enables us to bound the number of components that are considered.
For all the details of the reasoning, we refer to Section~\ref{sec:pmc}.

\subsection{Guessing fill-in edges with fixed endpoint}

Armed with the bound on the number of possible sections (Theorem~\ref{thm:over:sections}),
we move to the most technical result of our work.
\begin{theorem}\label{thm:over:cheap-fill-in}
Given an \icname{} instance $(G,k)$, where the Module Reduction Rule is not applicable,
and a designated vertex $v \in V(G)$, 
one can in $k^{\Oh(\sqrt{k})}n^{\Oh(1)}$ time enumerate a family $\fiset$
of at most $k^{\Oh(\sqrt{k})} n^{70}$ subsets of $V(G)$, such that
for any minimal solution $F$ to $(G,k)$ for which $v$ is cheap w.r.t. $F$,
the set $\{w \in V(G): vw \in F\}$ belongs to $\fiset$.
\end{theorem}

\begin{figure}
\centering
\includegraphics{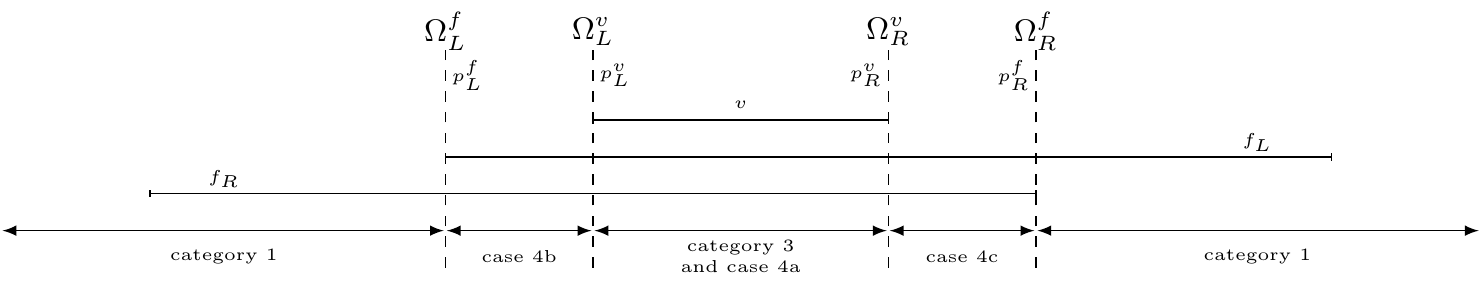}
\caption{Situation around the vertex $v$ in the proof of Theorem~\ref{thm:over:cheap-fill-in},
  together with categories and cases of Lemma~\ref{lem:over:fi:structure}.}
\label{fig:over:fill-in}
\end{figure}

We now sketch the proof of Theorem~\ref{thm:over:cheap-fill-in}; let $(G,k)$ and $v \in V(G)$
be as in the statement.
Fix a minimal completion $F$ of the \icname{} instance $(G,k)$, and fix a model $\sigma$ of $G+F$.
We define the following (see also Figure~\ref{fig:over:fill-in}).
\begin{enumerate}
\item Denote $p_L^v = \model(\Ibeg{v})$ and $p_R^v = \model(\Iend{v})$.
\item Let $f_L$ be the untouched vertex with the rightmost starting endpoint
among untouched vertices $f$ satisfying
$\model(\Ibeg{f}) \leq p_L^v < p_R^v \leq \model(\Iend{f})$.
\item Let $f_R$ be the untouched vertex with the leftmost ending endpoint
among untouched vertices $f$ satisfying
$\model(\Ibeg{f}) \leq p_L^v < p_R^v \leq \model(\Iend{f})$.
\item Denote $p_L^f = \model(\Ibeg{f_L})$ and $p_R^f = \model(\Iend{f_R})$.
\item Denote $\Isec_L^f = \Isec_\model(p_L^f)$, $\Isec_L^v = \Isec_\model(p_L^v)$,
  $\Isec_R^v = \Isec_\model(p_R^v-1)$ and $\Isec_R^f = \Isec_\model(p_R^f-1)$.
\end{enumerate}
Note that $\Uroot$ is a good candidate for both $f_L$ and $f_R$, thus these vertices exist.
We remark also that it may happen that $v=f_L$, $v=f_R$ or $f_L = f_R$. However, we may say
the following about the order of these vertices.
$$\model(\Ibeg{f_R}) \leq p_L^f \leq p_L^v < p_R^v \leq p_R^f \leq \model(\Iend{f_L}).$$

We start by enumerating all possible choices of vertices $f_L,f_R$ and sections
$\Isec_L^f$, $\Isec_L^v$, $\Isec_R^v$, $\Isec_R^f$, using the family $\sectionset$
of Theorem~\ref{thm:over:sections}.
By the bound of Theorem~\ref{thm:over:sections}, there are at most $k^{\Oh(\sqrt{k})} n^{70}$
subcases (henceforth called \emph{branches}) to consider.
In the rest of the proof we aim to compute a single set $B$ of size $\Oh(k^5)$
for a single choice of the aforementioned two vertices and four sections, such that
$B$ contains $\{w: vw \in F\}$ for any minimal solution $F$ to $(G,k)$
for which the choice of $f_L, f_R$ and $\Isec_L^f$, $\Isec_L^v$, $\Isec_R^v$, $\Isec_R^f$
is correct.
When the set $B$ is computed, we insert all its subsets of size at most $\sqrt{k}$
into the family $\fiset$.

Thus, henceforth we fix a choice of
$f_L,f_R$ and $\Isec_L^f$, $\Isec_L^v$, $\Isec_R^v$, $\Isec_R^f$
and we assume that the guess of these vertices and sets 
is correct for a minimal solution $F$ with model $\model$ of $G+F$.
Observe that we should expect the following:
\begin{align*}
v &\in \Isec_L^v \cap \Isec_R^v, \\
f_L,f_R &\in \Isec_L^f \cap \Isec_R^f, \\
\Isec_L^f \cap \Isec_R^f &\subseteq \Isec_L^f \cap \Isec_R^v \subseteq \Isec_L^v \cap \Isec_R^v, \\
\Isec_L^f \cap \Isec_R^f &\subseteq \Isec_L^v \cap \Isec_R^f \subseteq \Isec_L^v \cap \Isec_R^v.
\end{align*}

We maintain also a set $B^\mathrm{sure}$ of vertices $w$ for which we deduce
that $vw \in F$ is implied by the choice of 
$f_L,f_R$ and $\Isec_L^f$, $\Isec_L^v$, $\Isec_R^v$, $\Isec_R^f$.
We start with $B^\mathrm{sure} = (\Isec_L^v \cup \Isec_R^v) \setminus N_G(v)$.
If at any point the size of $B^\mathrm{sure}$ exceeds $k$, we discard the current branch.

We start with the following observation, directly implied by the assumption that $f_L$ and $f_R$
are untouched and $|F| \leq k$.
\begin{lemma}\label{lem:over:fi:structure}
For any connected component $C$ of $G \setminus (\Isec_L^f \cup \Isec_L^v \cup \Isec_R^v \cup \Isec_R^f)$ the following holds:
\begin{enumerate}
\item\label{case:over:fi:outside}
If $C \cap N_G(f_L) \cap N_G(f_R) = \emptyset$, then $\Send{C}{\model} < p_L^f$
or $\Sbeg{C}{\model} > p_R^f$. In particular, $vw \notin E(G) \cup F$ for every $w \in C$.
\item If $C$ contains a vertex of  $N_G(f_L) \cap N_G(f_R)$, then
$p_L^f < \Sbeg{C}{\model} < \Send{C}{\model} < p_R^f$ and
$C \subseteq N_G(f_L) \cap N_G(f_R)$.
\item\label{case:over:fi:top-forced}
If, moreover, $C$ contains a neighbor of $v$ in $G$, then
$p_L^v < \Sbeg{C}{\model} < \Send{C}{\model} < p_R^v$
and $vw \in E(G) \cup F$ for every $w \in C$.
\item\label{case:over:fi:dont-know}
In the last case, if $C \subseteq (N_G(f_L) \cap N_G(f_R)) \setminus N_G(v)$, then
one of the following cases hold:
\begin{enumerate}
\item\label{case:over:fi:top}
$p_L^v < \Sbeg{C}{\model} < \Send{C}{\model} < p_R^v$
and $vw \in F$ for every $w \in C$.
Moreover, in this case $N_G(C) \subseteq \Isec_L^v \cup \Isec_R^v$.
\item \label{case:over:fi:left}
$p_L^f < \Sbeg{C}{\model} < \Send{C}{\model} < p_L^v$
and $vw \notin F$ for every $w \in C$.
Moreover, in this case $N_G(C) \subseteq \Isec_L^f \cup \Isec_L^v$.
\item \label{case:over:fi:right}
$p_R^v < \Sbeg{C}{\model} < \Send{C}{\model} < p_R^f$
and $vw \notin F$ for every $w \in C$.
Moreover, in this case $N_G(C) \subseteq \Isec_R^f \cup \Isec_R^v$.
\end{enumerate}
Moreover, if $|C| > k$, then the first option does not happen.
\end{enumerate}
\end{lemma}
By Lemma~\ref{lem:over:fi:structure}, we can sort the connected components of
$G \setminus (\Isec_L^f \cup \Isec_L^v \cup \Isec_R^v \cup \Isec_R^f)$
into three \emph{categories}, depending on whether they fall into
point~\ref{case:over:fi:outside}, \ref{case:over:fi:top-forced}
or~\ref{case:over:fi:dont-know}. Obviously, the last category is the most interesting, as we are not able to directly decide whether the vertices of the component should be inserted into $B$ or not. 
The subpoints of this category (i.e,~\ref{case:over:fi:top},~\ref{case:over:fi:left} and~\ref{case:over:fi:right})
are henceforth called \emph{cases}. Note that for each connected component $C$
we know its category, but we do not know its case if it falls into category~\ref{case:over:fi:dont-know}.

We now perform some cleaning. 
If there exists a component
$C \in \Ccomp{G \setminus (\Isec_L^f \cup \Isec_L^v \cup \Isec_R^v \cup \Isec_R^f)}$
that does not fall into any category
(e.g., we have 
$C \not\subseteq N_G(f_L) \cap N_G(f_R)$, but $C$ contains a common neighbor
of $f_L$ and $f_R$), we discard the current branch.
Moreover,
we may include into $B^\mathrm{sure}$ all non-neighbors of $v$ that lie
in a connected component $C$ that falls into category~\ref{case:over:fi:top-forced}
of Lemma~\ref{lem:over:fi:structure}, that is, that contains a neighbor of $v$.

Clearly, only at most $k$ components fall into case~\ref{case:over:fi:top} of Lemma~\ref{lem:over:fi:structure}, since each such component induces at least one fill edge incident to $v$. However, we do not know which of the components falling into category~\ref{case:over:fi:dont-know} are in fact those interesting ones. Hence, our main task now is to pinpoint a set of $\Oh(k^4)$ potential components
falling into category~\ref{case:over:fi:dont-know}
for which case~\ref{case:over:fi:top} may possibly happen. As each such component
is of size at most $k$, this would conclude the proof of Theorem~\ref{thm:over:cheap-fill-in}.

Let $\Cfam$ be the family of all connected component $C$ of 
$G \setminus (\Isec_L^f \cup \Isec_L^v \cup \Isec_R^v \cup \Isec_R^f)$
that fall into category~\ref{case:over:fi:dont-know} of Lemma~\ref{lem:over:fi:structure},
that is, $C \subseteq (N_G(f_L) \cap N_G(f_R)) \setminus N_G(v)$.
We distinguish the following subfamilies that correspond to the subcases of category~\ref{case:over:fi:dont-know}.
\begin{align*}
\Cfam_v &= \{C \in \Cfam: N_G(C) \subseteq \Isec_L^v \cup \Isec_R^v\} \\
\Cfam_L &= \{C \in \Cfam: N_G(C) \subseteq \Isec_L^f \cup \Isec_L^v\} \\
\Cfam_R &= \{C \in \Cfam: N_G(C) \subseteq \Isec_R^f \cup \Isec_R^v\}
\end{align*}
If $\Cfam_v \cup \Cfam_L \cup \Cfam_R \neq \Cfam$, we discard the current branch.
Moreover, for any $C \in \Cfam_v \setminus (\Cfam_L \cup \Cfam_R)$ we include
all vertices of $C$ into $B^\mathrm{sure}$, as such a component will surely fall into case~\ref{case:over:fi:top}.

Our goal now is to focus on $\Cfam_L$ and pinpoint a small set
of components of $\Cfam_L \cap \Cfam_v$ that may possibly fall into case~\ref{case:over:fi:top}
of Lemma~\ref{lem:over:fi:structure}. The arguments for $\Cfam_R$ will be symmetrical.

To this end, we will construct a family $\Tfam \subseteq \Cfam_L$
of \emph{troublesome} components. Informally speaking, a component is troublesome
if it is highly unclear where or how it should live in the model $\model$.
We will argue that there is a bounded number of troublesome components (strictly speaking, $\Oh(k^2)$ of them)
and any component that falls into case~\ref{case:over:fi:top} of Lemma~\ref{lem:over:fi:structure}
is in some sense ``close'' to a troublesome component.

We first focus on components $C \in \Cfam_L \cap \Cfam_R$. Observe that
for such a component we have $N_G(C) \subseteq  \Isec_L^v \cap \Isec_R^v$.
Denote $P = \Isec_L^f \cap \Isec_R^f$ and $K = (\Isec_L^v \cap \Isec_R^v) \setminus P$.
By the choice of $f_L$ and $f_R$, each vertex in $K$ is touched by the solution $\sol$
and, consequently, $|K| \leq 2k$.
If there exists a vertex $v \in C$ with $P \not\subseteq N_G(v)$, then necessarily
$C$ is touched by the solution. Otherwise, $P \subseteq N_G(v) \subseteq P \cup K$ for any
$v \in C$ and, since the Module Reduction Rule is not applicable, we infer
that there are only $\Oh(k^2)$ components of $\Cfam_L \cap \Cfam_R$.
We treat all of them as troublesome ones, and put them into $\Tfam$.

Furthermore, we put into $\Tfam$ all connected components $C \in \Cfam_L$
that cannot be drawn in the model of a completion of $G$ between sections $\Isec_L^f$ and $\Isec_L^v$
without adding a fill-in edge. More formally,
we denote 
$F_L = \binom{\Isec_L^v}{2} \setminus E(G) \subseteq F$
and define the following:
\begin{definition}
A component $C \in \Cfam_L \cap \Cfam_v$ is \emph{freely drawable}
if there exists an interval model $\model_C$ of $(G+F_L)[C \cup \Isec_L^v]$
that starts with all starting events of $\events{\Isec_L^v \cap \Isec_L^f}$\
and ends with all ending events of $\events{\Isec_L^v}$.
\end{definition}
Observe that one can recognize freely drawable components in polynomial time using Lemma~\ref{lem:ic-cliques-drawing}.

It is easy to see that each component that is not freely drawable either
is touched by the solution $\sol$, or
falls into case~\ref{case:over:fi:right}.
However, in the latter case we have $C \in \Cfam_L \cap \Cfam_R$, and all such components
have already been considered troublesome.
Hence, we expect at most $2k$ not freely drawable components of $(\Cfam_L \cap \Cfam_v) \setminus \Cfam_R$, and we put all of them into $\Tfam$.

We now inspect the possible order of the starting endpoints
of the vertices of $\Isec_L^v \setminus \Isec_L^f$; all these endpoints
appear between positions $p_L^f$ and $p_L^v$. We denote
$$X = \bigcup_{C \in \Cfam_L \setminus \Cfam_v} N_G(C) \cap \Isec_L^v.$$
It turns out that any component
$C \in (\Cfam_L \cap \Cfam_v) \setminus \Cfam_R$ that contains a vertex
$w \in C$ with $X \not\subseteq N_G(w)$ 
is necessarily touched by $F$: the solution $F$ needs to make $w$
adjacent either to the entire $X$, or to some vertices of the connected component
of $\Cfam_L \setminus \Cfam_v$ that neighbors a vertex of $X \setminus N_G(w)$.
Thus, we may treat all such components $C$ as troublesome, and assume henceforth
that each remaining component $C \in (\Cfam_L \cap \Cfam_v) \setminus \Cfam_R$
is both freely drawable and fully adjacent to $X$.
We refer to Figure~\ref{fig:over:fill-in2} for an illustration.

\begin{figure}
\centering
\includegraphics{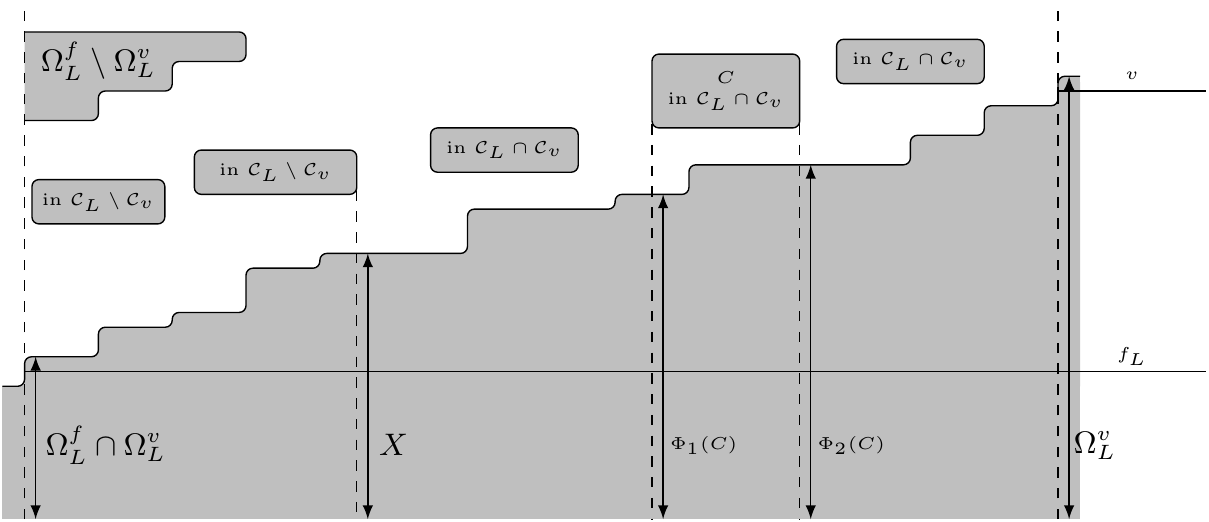}
\caption{The situation between $\Isec_L^f$ and $\Isec_L^v$.}
\label{fig:over:fill-in2}
\end{figure}

Now observe that if a component $C \in (\Cfam_L \cap \Cfam_v) \setminus \Cfam_R$ is freely
drawable, then there exist vertices $v_1,v_2 \in C$ with
\begin{align*}
N_G(v_1) \cap \Isec_L^v &= \minneiset{C} := \bigcap_{w \in C} N_G(w) \cap \Isec_L^v, \\
N_G(v_2) \cap \Isec_L^v &= \maxneiset{C} := \bigcup_{w \in C} N_G(w) \cap \Isec_L^v.
\end{align*}
Consider now two components $C_1,C_2 \in (\Cfam_L \cap \Cfam_v) \setminus \Cfam_R$.
If neither of them is touched by $\sol$
(in particular, neither of them falls into case~\ref{case:over:fi:top}), then
we should expect $\maxneiset{C_1} \subseteq \minneiset{C_2}$ or $\maxneiset{C_2} \subseteq \minneiset{C_1}$, depending on the relative order of $C_1$ and $C_2$ in the model $\model$.
Hence, if this is not the case, we have a \emph{conflict} between $C_1$ and $C_2$: one 
of these two components needs to be touched by $F$. 

We construct an auxiliary conflict graph, where each vertex corresponds to a not-yet-troublesome
component of $(\Cfam_L \cap \Cfam_v) \setminus \Cfam_R$,
and each edge corresponds to a conflict; by the previous argumentation, the components
touched by the solution need to form a vertex cover of this auxiliary conflict graph.
Hence, we may compute a $2$-approximate vertex cover of the conflict graph, and consider
all components of this vertex cover as troublesome.

This step concludes the recognition of troublesome components $\Tfam$.

We now observe that
$$(G+F_L)\left[\Isec_L^v \cup \bigcup \left((\Cfam_L \cap \Cfam_v) \setminus (\Cfam_R \cup \Tfam)\right)\right]$$
is an interval graph and, moreover, it admits an interval model that starts with the starting
events of $X$ and ends with the ending events of $\Isec_L^v$.
The crucial observation now is the following: if for some $C \in (\Cfam_L \cap \Cfam_v) \setminus (\Cfam_R \cup \Tfam)$, the sets $\minneiset{C}$ and $\maxneiset{C}$ differ significantly
from sets $\minneiset{D}$ and $\maxneiset{D}$ for all $D \in \Tfam$, then no troublesome
component will interfere with the representation of $C$ between positions $p_L^f$ and $p_L^v$
and, consequently, $C$ is untouched by the solution and falls into case~\ref{case:over:fi:left}.
The exhaustive application of Module Reduction Rule ensures that only a bounded number of 
components $C$ may have sets $\minneiset{C}$ and $\maxneiset{C}$ similar to some troublesome
component.
As there are only $\Oh(k^2)$ troublesome components, we are left only with a bounded number
of candidates for case~\ref{case:over:fi:top}.
This concludes the sketch of the proof of Theorem~\ref{thm:over:cheap-fill-in}.

\subsection{Dynamic programming}

Using the structural results of Theorems~\ref{thm:over:sections} and~\ref{thm:over:cheap-fill-in},
we now design a dynamic programming for \icname{}.

A straightforward approach, basing on the subexponential algorithm for the \textsc{Chordal Completion} problem,
would be to enumerate all possible sections via Theorem~\ref{thm:over:sections} and, for each section $\Isec$, try
to deduce (or guess) which components of $G \setminus \Isec$ lie to the left and which lie to the right
to the section $\Isec$. However, if $\Isec$ is large, there may be many such components with many different
neighborhoods in $\Isec$ and, consequently, such a guessing step seems expensive (see Figure~\ref{fig:piramid} in the introduction).
Thus, we need to employ a more involved definition of a ``separation'' to define a subproblem for the dynamic programming.

Inspired by the example on Figure~\ref{fig:piramid}, we start with the following approach.
For each vertex $v$ that is cheap in the canonical solution $\sol$,
we take all possible candidate values for $p_L = \model(\Ibeg{v})$, $p_R = \model(\Iend{v})-1$,
$\Isec_L^v = \Isec_\model(p_L)$, $\Isec_R^v = \Isec_\model(p_R)$
and $\incF{v}$; we call such a tuple a \emph{world} $\W$.
Observe that, by Theorems~\ref{thm:over:sections} and~\ref{thm:over:cheap-fill-in},
there are only $k^{\Oh(\sqrt{k})} n^{\Oh(1)}$ reasonable worlds. 
For each world, we would like to know the optimum way to arrange the events
between positions $p_L$ and $p_R$, i.e., among vertices of $N_{G+\incF{v}}(v)$.
Observe that, in particular, a world does not distinguish
which vertices of $G \setminus N_{G+\incF{v}}(v)$ are before or after $v$
in the model $\model$.

\begin{figure}
\centering
\includegraphics{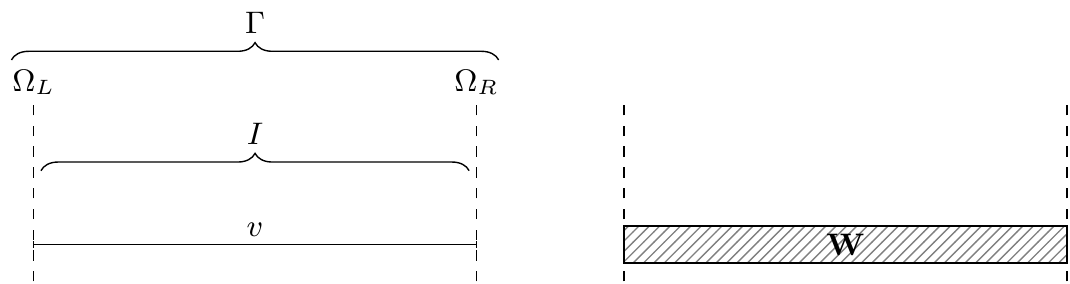}
\caption{A world with its most important elements (to the left) and its symbolic notation used in subsequent figures (to the right).}
\label{fig:over:world}
\end{figure}

\begin{figure}
\centering
\includegraphics{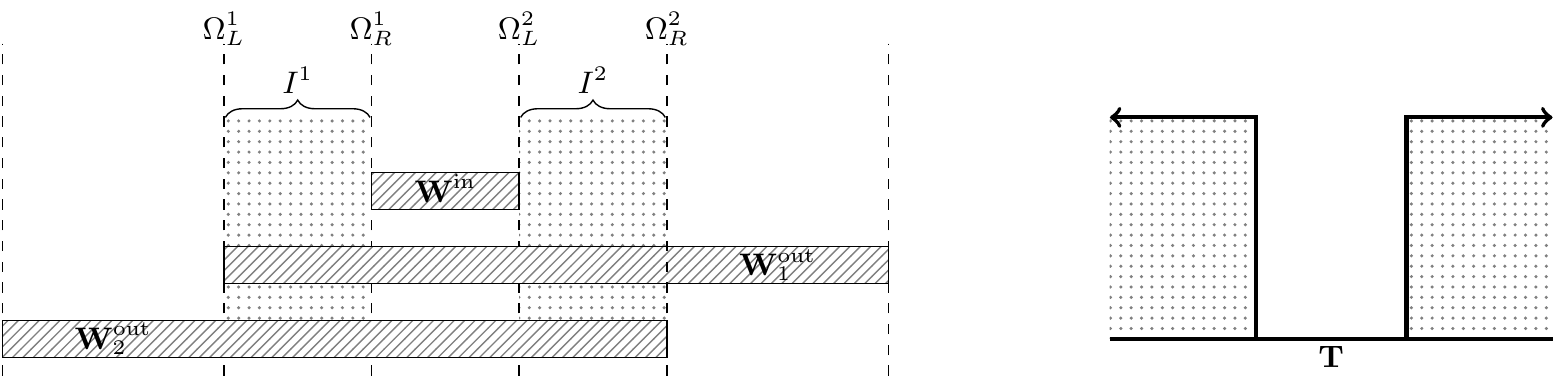}
\caption{A terrace with its most important notation (to the left) and its symbolic notation used in subsequent figures (to the right).
The dotted areas are the `important' areas for a terrace: the left one has borders $\Omega_L^1$, $\Omega_R^1$ and interior $I^1$,
    and the right one has borders $\Omega_L^2$, $\Omega_R^2$ and interior $I^2$.}
\label{fig:over:terrace}
\end{figure}

However, the family of worlds is not rich enough to allow a dynamic programming algorithm. To understand it, consider two worlds that are nested, i.e., one contains the other. To compute the value for the outer world basing on the inner one, we need to control {\em{two}} areas in their difference, which seems difficult given only other worlds as states. Therefore, we introduce the notion of a \emph{terrace}, depicted on Figure~\ref{fig:over:terrace}.
Here, we consider three worlds $\Win$, $\Wout_1$ and $\Wout_2$ with their respective
cheap vertices $v$, $v_1$ and $v_2$ where:
\begin{itemize}
\item $\model(\Ibeg{v_i}) < \model(\Ibeg{v}) < \model(\Iend{v}) < \model(\Iend{v_i})$ for $i=1,2$, and
\item $v_1$ has the rightmost starting event in the model $\model$, among cheap vertices satisfying the previous condition,
 and $v_2$ has the leftmost ending event. 
\end{itemize}
In a terrace, we are interested the optimum way to arrange events
in one of the dotted areas on Figure~\ref{fig:over:terrace}.
Observe that each vertex whose interval is fully contained in one of these areas
belongs to $I := (N_{G+\sol}(v_1) \cap N_{G+\sol}(v_2)) \setminus (N_{G+\sol}(v) \cup \Omega_L^1 \cup \Omega_R^2)$.

We would like to reason how the vertices of $I$ are split between areas $I^1$ and $I^2$.
The crucial observation is that, by the choice of $v_1$ and $v_2$,
each vertex of $\Omega_R^1 \cap \Omega_L^2$ that
has an endpoint in the dotted areas (i.e., does not belong to $\Omega_L^1 \cap \Omega_R^2$)
needs to be expensive and, consequently, there are at most $2\sqrt{k}$ such vertices.
Denote the set of these vertices as $K$, that is, $K = (\Omega_R^1 \cap \Omega_L^2) \setminus (\Omega_L^1 \cap \Omega_R^2)$.

Consider now a connected component $C$ of $G[I]$. We distinguish two cases for the alignment
of $C$ in the interval graph $G+\sol$: either there exist two vertices $v_1,v_2 \in C$
with $N_{G+\sol}(v_1) \cap K \neq N_{G+\sol}(v_2) \cap K$, or all of the vertices
of $C$ have the same neighborhood in $K$ in the graph $G+\sol$.
In the latter case, we argue that the component $C$ chooses its place in the model $\model$
in a greedy manner, and there are only $k^{\Oh(\sqrt{k})}$ ways to arrange such components.
In the first case, observe that such a component $C$ ``occupies'' an endpoint event of a vertex of $K$
and, if two components $C_1$ and $C_2$ occupy the same endpoint, they need to be connected
by an edge of $\sol$. Since $|K|\leq 2\sqrt{k}$, then we have at most $4\sqrt{k}$ endpoints of vertices of $K$. If endpoint $\varepsilon\in \events{K}$ is occupied by $a_\varepsilon$ components, then this means that we need to add at least $\binom{a_\epsilon}{2}$ fill-in edges between these components. Then we have that $|\events{K}|\leq 4\sqrt{k}$ and $\sum_{\varepsilon\in \events{K}}\binom{a_\varepsilon}{2}\leq k$, and a simple application of the Cauchy-Schwarz inequality shows that $\sum_{\varepsilon\in \events{K}} a_\varepsilon = \Oh(k^{3/4})$, i.e., there are only $\Oh(k^{3/4})$ components that fall into the first case. Moreover, the exhaustive application of Module Reduction Rule ensures us
that there are only $k^{\Oh(1)}$ components of $G[I]$ in total.

Hence, we have $k^{\Oh(k^{3/4})}$ guesses which components fall into the first case, $2^{\Oh(k^{3/4})}$ guesses about their alignment to $I^1$ or $I^2$,
and then the remaining components can be processed greedily.
In Section~\ref{sec:left-right} we develop a more careful argument that bounds
the number of components that fall into the first case by $\Oh(\sqrt{k})$, instead of $\Oh(k^{3/4})$ as presented in the argument above.

To sum up, we have $k^{\Oh(\sqrt{k})} n^{\Oh(1)}$ reasonable choices
for a terrace, together with the partition of the set $I$ into dotted areas
$I^1$ and $I^2$.

It turns out that 
the family of all terraces and worlds is almost sufficient to perform a dynamic programming
algorithm.
More precisely, we consider \emph{pairs} of terraces or worlds, together with their
``important areas'', and ask for the best way to arrange events in the \emph{intersection}
of the important areas (see Figure~\ref{fig:over:state}).
As the number of such dynamic programming states is bounded by $k^{\Oh(\sqrt{k})} n^{\Oh(1)}$,
   we obtain a dynamic programming algorithm running within the promised time bound,
   concluding the proof of Theorem~\ref{thm:ic}.

\begin{figure}
\centering
\includegraphics{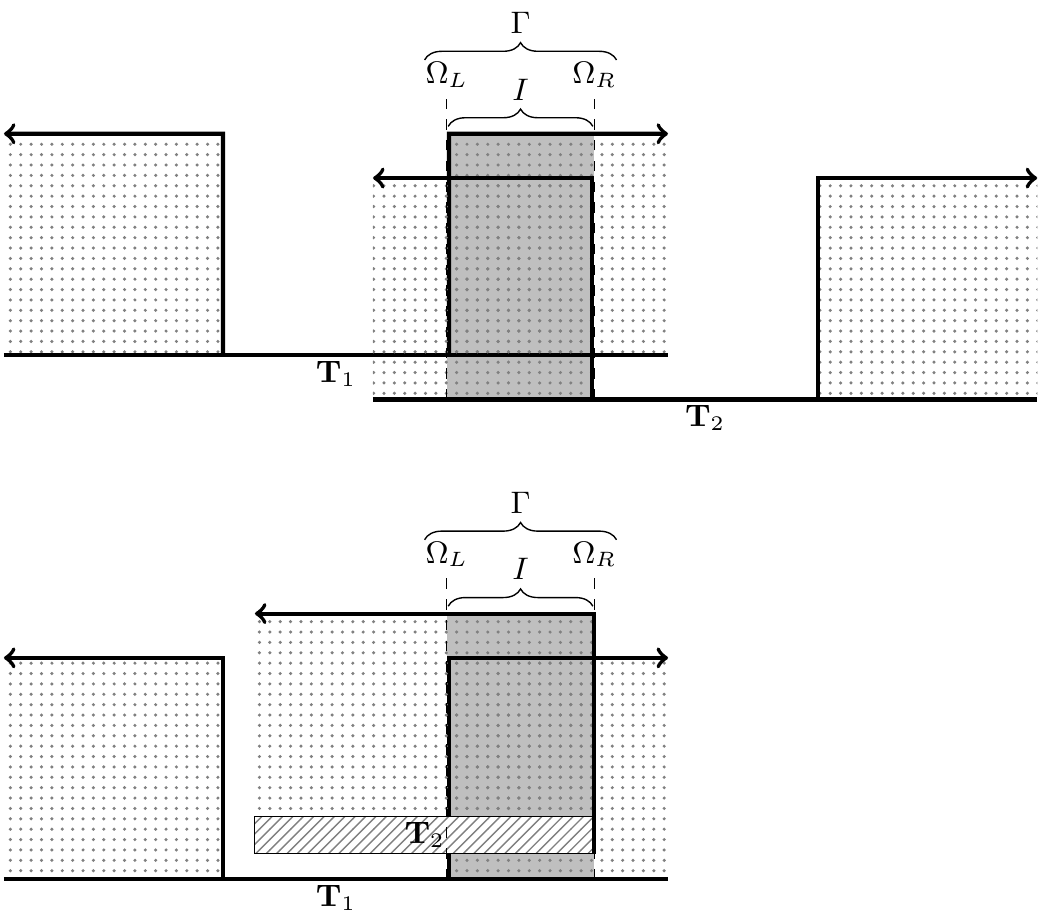}
\caption{A DP state defined by two terraces (above) and a terrace and a world (below).
 The DP state asks for the optimal way to arrange events in the gray area.}
\label{fig:over:state}
\end{figure}

\section{Modules and neighborhood classes}\label{sec:neighbors}
Sections~\ref{sec:neighbors}--\ref{sec:dp} contain a full proof of Theorem~\ref{thm:ic}.
We start with a study of possible neighborhood classes in
a (almost) interval graph $G$, and provide the aforementioned module-based reduction rule
in full detail.

\subsection{Modules and module-based reduction rule}

Recall that $M \subseteq V(G)$ is a \emph{module} in a graph $G$ if
$N(v_1) \setminus M = N(v_2) \setminus M$ for any $v_1,v_2 \in M$.
(Equivalently, for any $v \notin M$ we have either $M \subseteq N(v)$ or $M \cap N(v) = \emptyset$.)
A module $M$ is \emph{connected} if $G[M]$ is connected.
Cao proved the following:
\begin{lemma}[Theorem 4.2 of~\cite{yixin:ic}]\label{lem:module-stays}
If $M$ is a connected module in $G$, and $F$ is a minimum completion of $G$, then $M$ is a module in $G+F$ as well.
\end{lemma}
Motivated by Lemma~\ref{lem:module-stays}, we formulate the following reduction rule.

\begin{redrule}[Module Reduction Rule]
Let $(G,k)$ be an instance of \textsc{Interval Completion}.
Assume there exists $X \subseteq V(G)$ and connected components
$M_1,M_2,\ldots,M_{2k+3}$ of $G \setminus X$ that are modules in $G$ and, moreover,
$N(M_i) = N(M_1)$ for each $1 \leq i \leq 2k+3$.
Then proceed as follows. If for more than $k$ indices $i$ the subgraph $G[M_i]$
is not an interval graph, return that $(G,k)$ is a NO-instance.
Otherwise, pick arbitrary $j$ such that $G[M_j]$ is an interval graph and remove $M_j$
from $G$.
\end{redrule}

Clearly, if $G[M_i]$ is not an interval graph, any completion of $G$ needs to contain
an edge with both endpoints in $M_i$. Hence, the size of a minimum completion of $G$
is lower bounded by the number of $M_i$s such that $G[M_i]$ is not an interval graph. Consequently,
if the Module Reduction Rule concludes that $(G,k)$ is a NO-instance, then the
conclusion is correct.

Moreover, observe that any solution to \textsc{Interval Completion} in $G$
naturally projects to a solution in $G \setminus M_j$ of at most the same size:
if $G+F$ is an interval graph, so is $(G+F) \setminus M_j$.
The following lemma shows that the deletion of $M_j$ in the
Module Reduction Rule actually does not change our task at all.

\begin{lemma}\label{lem:module-rule-safe}
Assume that Module Reduction Rule is applicable to graph $G$, and its application
deletes a module $M_j$.
Then any solution to $(G \setminus M_j,k)$ is a solution to $(G,k)$ as well.
\end{lemma}
\begin{proof}
Without loss of generality assume that $j = 2k+3$.
Let $G' = G \setminus M_j$, let $F$ be a solution to $(G', k)$ and let $\model$
be an interval model of $G'+F$.
As $|F| \leq k$, there are at least two modules $M_i$ ($1 \leq i \leq 2k+2$) untouched by $F$;
w.l.o.g. assume $M_1$ and $M_2$ are untouched by the solution.
In the following we show that $M_1$ and $M_2$ ``reserve'' a space in the model $\model$
where we can insert $M_j$ without any further cost.

As $M_1$ and $M_2$ are two connected component of $G \setminus X$ and both are untouched by $F$,
all events of $\events{M_1}$ lie before all events of $\events{M_2}$, or
all events of $\events{M_1}$ lie after all events of $\events{M_2}$ in the model $\model$;
w.l.o.g. assume the first case.
Denote $p_1 = \Send{M_1}{\model}$ and $p_2 = \Sbeg{M_2}{\model}$; note that $p_1 < p_2$. 
Let $Y = N(M_1) = N(M_2) \subseteq X$.
As both $M_1$ and $M_2$ are untouched by $F$,
we infer that $\Isec_\model(p_1) = \Isec_\model(p_2-1) = Y$,
and $Y$ is a clique in $G'+F$.

Let $\hat{\model}$ be an interval model of $G[M_j]$.
Consider a model $\model'$ created from $\model$
by inserting all events of $\events{M_j}$ after position $p_1$ in $\model$, in the order
according to model $\hat{\model}$.
As $\Isec_\model(p_1) = N_G(M_j) = Y$, this is an interval model of $G+F$, and the lemma is proven.
\end{proof}

We now describe how to apply the Module Reduction Rule efficiently.
To this end, we recall the module decomposition theorem, introduced by Gallai~\cite{modular-decomp}.

A module decomposition of a graph $G$ is a rooted tree $T$,
where each node $t$ is labeled by a module $M^t \subseteq V(G)$,
and is one of four types:
\begin{description}
\item[leaf] $t$ is a leaf of $T$, and $M^t$ is a singleton;
\item[union] $G[M^t]$ is disconnected, and the children of $t$ are labeled with different connected
components of $G[M^t]$;
\item[join] the complement of $G[M^t]$ is disconnected, and the children of $t$ are labeled
with different connected components of the complement of $G[M^t]$;
\item[prime] neither of the above holds, and the children of $t$ are labeled with different
modules of $G$ that are proper subsets of $M^t$, and are inclusion-wise maximal with this property.
\end{description}
Moreover, we require that the root of $T$ is labeled with the module $V(G)$. We need the following properties of the module decomposition.
\begin{theorem}[see \cite{compute-modular-decomp}]\label{thm:module-decomp}
For a graph $G$, the following holds.
\begin{enumerate}
\item A module decomposition $(T,(M^t)_{t \in V(T)})$ of $G$ exists, is unique, and computable in linear time.
\item At any prime node $t$ of $T$, the labels of the children form a partition of $M^t$. In particular, for each vertex $v$ of $G$ 
there exists exactly one leaf node with label $\{v\}$.
\item Each module $M$ of $G$ is either a label of some node of $T$, or there exists
a \textbf{union} or \textbf{join} node $t$ such that $M$ is a union of labels of some children
of $G$.
\end{enumerate}
\end{theorem}
We now show that the Module Reduction Rule can be applied efficiently
using the module decomposition of a graph.
\begin{lemma}\label{lem:module-rule-apply}
There is a polynomial-time algorithm that, given an instance $(G,k)$
finds sets $X,M_1,\ldots,M_{2k+3} \subseteq V(G)$ on which Module Reduction Rule
is applicable, or correctly concludes that no such sets exists.
\end{lemma}
\begin{proof}
We claim that, if the Module Reduction Rule is applicable to sets $X,M_1,\ldots,M_{2k+3}$
then there exists a \textbf{union} node $t$ such that each set $M_i$ is a label of some
child of $t$.

From the last property of Theorem~\ref{thm:module-decomp} we infer that, 
for any two modules $M$, $M'$ of $G$, we have $M \subseteq M'$, $M' \subseteq M$
or $M \cap M' = \emptyset$ unless
there exists a \textbf{union} or \textbf{join} node $t$ in the module decomposition of $G$
such that both $M$ and $M'$ are unions of labels of some children of $t$.

Notice now that a union of arbitrary number of sets $M_i$ is a module in $G$ as well. By applying the conclusion of the last paragraph to the modules $\bigcup_{i=1}^{2k+2} M_i$ and $\bigcup_{i=2}^{2k+3} M_i$, and using the fact that all $M_i$s are connected and pairwise non-adjacent, we infer that $M_i$s must be in fact children of the same \textbf{union} node $t$.

Therefore, to look for an application of the Module Reduction Rule it suffices
to inspect all \textbf{union} nodes of the module decomposition of $G$, and
for each such node $t$, classify the labels of the children of $t$ according to their
neighborhood.
The Module Reduction Rule is applicable if and only if for some \textbf{union} node $t$
at least $2k+3$ children of $t$ have labels with equal neighborhood.
\end{proof}

By Lemma~\ref{lem:module-rule-safe}, an application of the Module Reduction Rule
does not change the answer to the input instance $(G,k)$.
Lemma~\ref{lem:module-rule-apply} shows that the rule can be applied in polynomial time.
Thus, we may apply Module Reduction Rule exhaustively and henceforth we assume, sometimes implicitly, that
it is no longer applicable.

\subsection{Neighborhood classes}\label{ssec:nei-classes}

We now provide some auxiliary structural lemmas about neighborhood classes
in the input graph $G$.

For a graph $G$ and a set $A \subseteq V(G)$, we say
that two vertices $v_1,v_2 \notin A$ have \emph{the same neighborhood with respect to $A$}
if $N_G(v_1) \cap A = N_G(v_2) \cap A$.
Clearly, this is an equivalence relation on $V(G) \setminus A$; each equivalence class
of this relation is called \emph{a neighborhood class w.r.t. $A$}.

The motivation for the results in this section is the following.
In many places the algorithm makes
some branching, choosing some vertex or a connected subgraph.
In a straightforward analysis, each such branching will have around $n$ options.
With a branching of depth $\sqrt{k}$, and without a polynomial kernel for \textsc{Interval Completion},
this would lead to undesirable $n^{\sqrt{k}}$ factor in the running time.
The structural results developed here limit the number of options in such branchings to polynomial in $k$;
in some sense they are ``local'' kernelization results.

\begin{lemma}\label{lem:A-nei}
Assume $G$ is a graph with completion set $F$, and let $A \subseteq V(G)$.
Then in $G$ there are at most $(2|A|+1)^2 + |F|$ neighborhood classes w.r.t. $A$.
In particular, if $(G,k)$ is a YES-instance of \textsc{Interval Completion},
then there are at most $(2|A|+1)^2 + k$ neighborhood classes w.r.t. $A$.
\end{lemma}
\begin{proof}
Let $X \subseteq V(G) \setminus A$ be the set of vertices such that there
exists some fill-in edge $xa \in F$ with $x \in X$ and $a \in A$. Clearly $|X| \leq |F|$.
To prove the lemma it suffices to show that there are at most $(2|A|+1)^2$ neighborhood classes
w.r.t. $A$ in the graph $G \setminus X$.

Let $\sigma$ be an interval model of the graph $G+F$.
Pick any $v \in V(G) \setminus (A \cup X)$. 
As $v \notin X$, the edges between $v$ and $A$ in $G$ are defined by the interval model $\sigma$,
that is, $va \notin E(G)$ for $a \in A$ iff $\sigma(\Iend{a}) < \sigma(\Ibeg{v})$ or $\sigma(\Ibeg{a}) > \sigma(\Iend{v})$.
Consider the model $\sigma$ restricted to $\events{A}$, and note that
there are $|\events{A}|+1 = 2|A|+1$ ways to insert the event $\Ibeg{v}$ into this model, and at most this number of ways to insert $\Iend{v}$.
Consequently, there at most $(2|A|+1)^2$ possible neighborhood classes w.r.t. $A$ for vertices $v \in V(G) \setminus (A \cup X)$
and the lemma follows.
\end{proof}

\begin{lemma}\label{lem:A-r-nei}
Assume $(G,k)$ is a YES-instance of \textsc{Interval Completion}, and the Module Reduction Rule
is not applicable to $(G,k)$. Let $r$ be a positive integer and let $A \subseteq V(G)$.
Then the number of connected components $C$ of $G \setminus A$ for which there exists $v_C \in C$ with
$|A \setminus N_G(v_C)| \leq r$ is at most $12kr + 4k + 18r + 4$.
\end{lemma}

\begin{figure}
\centering
\includegraphics{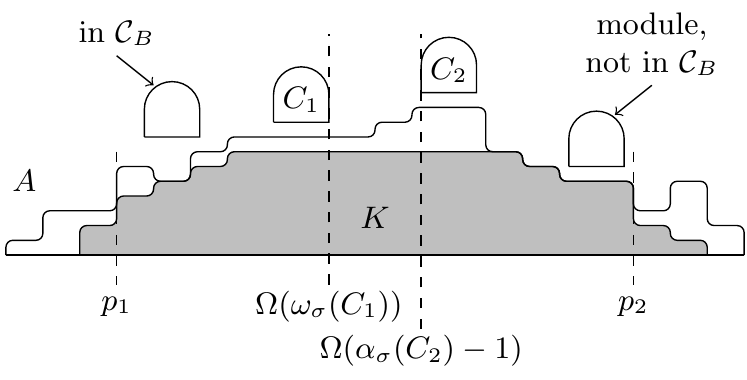}
\caption{Notation used in the proof of Lemma~\ref{lem:A-r-nei}.}
\label{fig:A-r-nei}
\end{figure}

\begin{proof}
Let $F$ be a solution to $(G,k)$, and let $\model$ be a model of $G+F$.
Let $\mathcal{C}$ be the set of all connected components $C$ of $G \setminus A$ that are untouched by $F$
and for which there exists $v_C \in C$ with $|A \setminus N_G(v_C)| \leq r$. We aim to show that $|\mathcal{C}| \leq (6r+1)(2k+2) + 6r+2$, which will settle the claim since at most $2k$ components of $G\setminus A$ are touched by $F$.

If $|\mathcal{C}| \leq 1$ then there is nothing to show, so assume otherwise.
Let $C_1,C_2 \in \mathcal{C}$. As both $C_1$ and $C_2$ are untouched, and there are no edges between the vertices of $C_1$ and the vertices of $C_2$,
in the model $\model$ all events of $\events{C_1}$ lie before or after all events of $\events{C_2}$; without loss of generality assume that $\Send{C_1}{\model} < \Sbeg{C_2}{\model}$.
Denote $K = A \cap N_G(v_{C_1}) \cap N_G(v_{C_2})$.
Note that $|K| \geq |A| - 2r$ and $K \subseteq \Isec(\Send{C_1}{\model})$, $K \subseteq \Isec(\Sbeg{C_2}{\model}-1)$. Consequently, $K$ is a clique in $G+F$.
We refer to Figure~\ref{fig:A-r-nei} for an illustration of the notation used in this proof.

Denote $B = A \setminus K$, we have $|B| \leq 2r$.
Let $\eventssymb \subseteq \events{K}$ be the set of  the last
$r+1$ starting events of $\events{K}$ and the first $r+1$ ending events of $\events{K}$ in the model
$\model$ (or $\eventssymb=\events{K}$ in case $|K|\leq r+1$).
Recall that $K$ is a clique in $G+F$ and $K \subseteq \Isec(\Send{C_1}{\model})$,
so all starting events of $\events{K}$ appear before position $\Send{C_1}{\model}$,
and all ending events of $\events{K}$ appear after this position.

Let $\mathcal{C}_B$ be the set of these connected components $C\in \mathcal{C}$ for which there exists $\event \in \eventssymb \cup \events{B}$ with
\begin{equation}\label{eq:eats-endpoint}
\Sbeg{C}{\model} < \model(\event) < \Send{C}{\model}.
\end{equation}
As the components of $\mathcal{C}$
are untouched by $F$ and pairwise non-adjacent in $G$, no two components of $\mathcal{C}$ can satisfy \eqref{eq:eats-endpoint}
with the same event $\event$. Consequently,
$$|\mathcal{C}_B| \leq |\eventssymb \cup \events{B}| \leq 6r+2.$$

Denote by $p_1$ and $p_2$ the positions of the first and last event of $\eventssymb$, respectively.
By the definition of $\eventssymb$, all events of $\events{A}$ that lie between $p_1$ and $p_2$
belong to $\eventssymb \cup \events{B}$.

Let $C \in \mathcal{C} \setminus \mathcal{C}_B$.
As $|A \setminus N_G(v_C)| \leq r$, 
in the model $\model$ all events of $\events{C}$ lie between the first and the last
event of $\eventssymb$. Consequently, by the definition of $\mathcal{C}_B$, $C$
is a module in $G+F$; as $C$ is untouched by $F$, $C$ is a module in $G$ as well.
Moreover, if for two components $C,C' \in \mathcal{C} \setminus \mathcal{C}_B$ 
the events of $\events{C}$ and $\events{C'}$ lie between the same two events of
$\eventssymb \cup \events{B}$, then $N_G(C) = N_G(C')$.
Therefore, if more than $2k+2$ such components lie between two consecutive events
of $\eventssymb \cup \events{B}$, the Module Reduction Rule would be applicable.
Consequently $|\mathcal{C} \setminus \mathcal{C}_B| \leq (6r+1)(2k+2)$,
and the lemma is proven.
\end{proof}

\section{Listing potential maximal cliques and sections}\label{sec:pmc}
In this section we prove the following result.

\begin{theorem}\label{thm:sections}
Given an \textsc{Interval Completion} instance $(G,k)$, where the Module Reduction Rule is not applicable, one
can in $\Ohstar(k^{\Oh(\sqrt{k})})$ time enumerate
a family $\sectionset$ of $k^{\Oh(\sqrt{k})} n^{17}$ subsets
of $V(G)$, such that for any minimal solution $F$ to $(G,k)$,
in the canonical model $\model$ of $G+F$
all sections of $\model$ belong to $\sectionset$.
\end{theorem}

As an intermediate step, we provide an enumeration algorithm for potential maximal cliques
in the \textsc{Interval Completion} problem, showing the following.

\begin{theorem}\label{thm:pmcs}
Given an \textsc{Interval Completion} instance $(G,k)$, where the Module Reduction Rule is not applicable, one
can in $\Ohstar(k^{\Oh(\sqrt{k})})$ time enumerate
a family $\pmcset$ of $k^{\Oh(\sqrt{k})} n^{8}$ subsets
of $V(G)$, such that for any minimal solution $F$ to $(G,k)$,
all maximal cliques of $G+F$ belong to $\pmcset$.
\end{theorem}

It is not hard to see that Theorem~\ref{thm:pmcs} implies Theorem~\ref{thm:sections}.
\begin{proof}[Proof of Theorem~\ref{thm:sections}]
Let $(G,k)$ be an \textsc{Interval Completion} instance, $F$ be a minimal solution to $(G,k)$ with $\model$ being the canonical model of $G+F$.
Clearly, $\emptyset$, $\{\Uroot\}$, $\{\Uroot,\Lroot\}$ and $\{\Uroot,\Rroot\}$ are sections of $\model$;
we include them into $\sectionset$ at the beginning.

Let $\Isec_\model(p)$ be a section of $\model$. Without loss of generality, assume that
$\Isec_\model(p)$ does not equal any of the four aforementioned ``obvious'' sections. 
Let $p_1 \leq p$ be the largest integer such that $\Isec_\model(p_1)$ is a maximal clique
of $G+F$; such $p_1$ always exists as $p_1=2$ with $\Isec_\model(2) = \{\Uroot,\Lroot\}$
is a candidate value. Symmetrically, we define $p_2$ to be the smallest integer with
$p_2 \geq p$ such that $\Isec_\model(p_2)$ is a maximal clique of $G+F$.

Let $r = |\Isec_\model(p_1) \setminus \Isec_\model(p_2)|$.
We infer that 
$\model$ places events of $\{\Iend{v}: v \in \Isec_\model(p_1) \setminus \Isec_\model(p_2)\}$ on positions $p_1+1, p_1+2,\ldots,p_1+r$, and then it places events of $\{\Ibeg{v}: v \in \Isec_\model(p_2) \setminus \Isec_\model(p_1)\}$ on positions $p_1+r+1,p_1+r+2,\ldots,p_2$; otherwise there would be a section between sections $\Isec_\model(p_1)$ and $\Isec_\model(p_2)$ that would yield a maximal clique, contradicting the choice of $p_1$ or of $p_2$.
Moreover, by Lemma~\ref{lem:can:endpoints} the events of $\{\Iend{v}: v \in \Isec_\model(p_1) \setminus \Isec_\model(p_2)\}$ are sorted according to the reversed total order $\prec$, while the events of $\{\Ibeg{v}: v \in \Isec_\model(p_2) \setminus \Isec_\model(p_1)\}$ are sorted according to the total order $\prec$.
Consequently, the set $\Isec_\model(p)$ can be deduced from
the maximal cliques $\Isec_\model(p_1)$ and $\Isec_\model(p_2)$ (both belonging
to the set $\pmcset$ given by Theorem~\ref{thm:pmcs})
and the value of $p-p_1$, for which we have $n+1$ choices. Theorem~\ref{thm:sections} follows.
\end{proof}

Thus, the rest of this section is devoted to the proof of Theorem~\ref{thm:pmcs}.

\subsection{Eight important vertices and the structure of the clique}

Let us fix an \textsc{Interval Completion} instance $(G,k)$, its minimal solution $F$,
a model $\model$ of $G+F$ and a maximal clique $\Isec = \Isec_\model(p)$. Recall
that $\model(\Ibeg{v_2}) = p$ and $\model(\Iend{v_1}) = p+1$ for some vertices $v_1$ and $v_2$.
Without loss of generality, assume that $\Isec$ is different than two ``obvious'' maximal
cliques $\{\Uroot,\Lroot\}$ and $\{\Uroot,\Rroot\}$ and, consequently,
$3 < p < 2n-3$ and $v_1,v_2 \notin \{\Uroot,\Lroot,\Rroot\}$.

\begin{figure}
\centering
\includegraphics{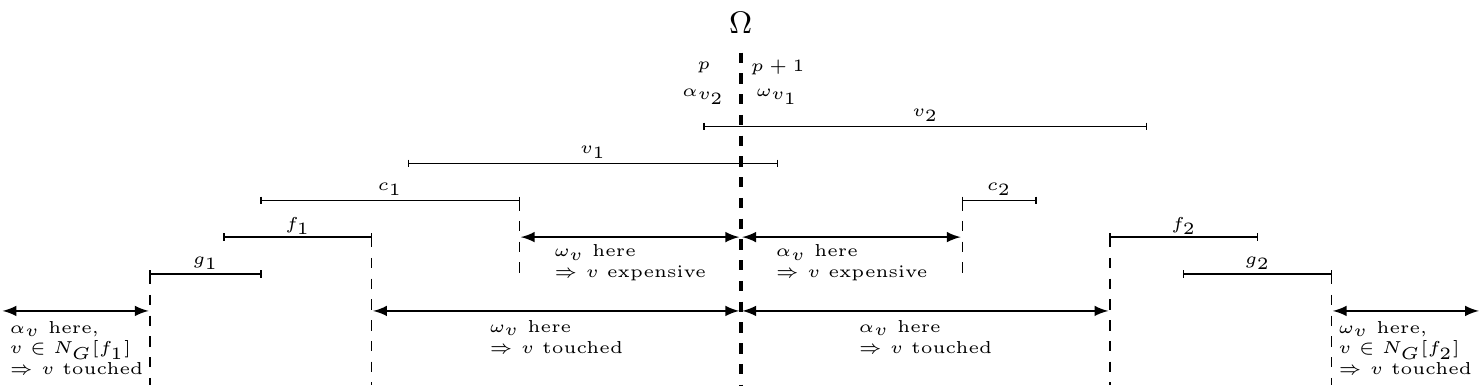}
\caption{The anatomy of a maximal clique $\Omega$, with eight vertices guessed by the algorithm.}
\label{fig:eight}
\end{figure}

Define the following vertices (see also Figure~\ref{fig:eight}):
\begin{enumerate}
\item $c_1$ is the cheap vertex with rightmost position of $\Iend{c_1}$, among
the cheap vertices $c$ satisfying $\model(\Iend{c}) \leq \model(\Iend{v_1}) = p+1$;
\item $c_2$ is the cheap vertex with leftmost position of $\Ibeg{c_2}$, among
the cheap vertices $c$ satisfying $\model(\Ibeg{c}) \geq \model(\Ibeg{v_2}) = p$;
\item $f_1$ is the untouched vertex with rightmost position of $\Iend{f_1}$, among
the untouched vertices $f$ satisfying $\model(\Iend{f}) \leq \model(\Iend{v_1}) = p+1$;
\item $f_2$ is the untouched vertex with leftmost position of $\Ibeg{f_2}$, among
the untouched vertices $f$ satisfying $\model(\Ibeg{f}) \geq \model(\Ibeg{v_2}) = p$;
\item $g_1$ is the untouched vertex with leftmost position of $\Ibeg{g_1}$, among
all untouched vertices of $N_G[f_1] \setminus \{\Isec \setminus \{v_1\}\}$;
\item $g_2$ is the untouched vertex with rightmost position of $\Iend{g_2}$, among
all untouched vertices of $N_G[f_2] \setminus \{\Isec \setminus \{v_2\}\}$.
\end{enumerate}
Let us remark that some of these vertices can be in fact equal. We also remark that all quantifications in the aforementioned definitions are done on nonempty sets:
$\Lroot$ is a good candidate for both $c_1$ and $f_1$, $\Rroot$ is a good candidate for both
$c_2$ and $f_2$, $f_1$ is a good candidate for $g_1$ and $f_2$ is a good candidate for $g_2$.
Hence, all these vertices are well-defined.

We observe the following relations between the positions of endpoints of the previously defined vertices.
\begin{lemma}\label{lem:pmc:vtx-pos}
The following inequalities hold:
\begin{align*}
& \model(\Iend{g_1}) \leq \model(\Iend{f_1}) \leq \model(\Iend{c_1}) \leq \model(\Iend{v_1}) = p+1 \\
& \model(\Ibeg{g_2}) \geq \model(\Ibeg{f_2}) \geq \model(\Ibeg{c_2}) \geq \model(\Ibeg{v_2}) = p
\end{align*}
\end{lemma}
\begin{proof}
The first inequality in each line follows from the definition of $f_1$ and $f_2$, as otherwise
$g_1$ or $g_2$ would be a better candidate for $f_1$ or $f_2$, respectively.
The remaining inequalities follow directly from the
definitions of the corresponding vertices.
\end{proof}

We also need the following observation.
\begin{lemma}\label{lem:pmc:v1v2}
$v_1\in N_G[v_2]$ and $v_2\in N_G[v_1]$.
\end{lemma}
\begin{proof}
If $v_1=v_2$ then the claim is obvious, so assume otherwise. For the sake of contradiction suppose $v_1v_2\notin E(G)$, so $v_1v_2 \in F$ since $v_1v_2\in E(G+F)$.
Note that by swapping the events $\Iend{v_1}$ and $\Ibeg{v_2}$ in the model $\model$
we obtain a model for $G+(F \setminus \{v_1v_2\})$, contradicting the minimality of $F$.
\end{proof}

We say that a vertex $v$ \emph{lies to the left} of the clique $\Isec$ 
if $\model(\Iend{v}) \leq p+1$, and \emph{lies to the right} if $\model(\Ibeg{v}) \geq p$.
Clearly, $v_1,c_1,f_1,g_1$ lie to the left of $\Isec$ and
$v_2,c_2,f_2,g_2$ lie to the right of $\Isec$.
Note that, perhaps a bit counter-intuitively, if $v=v_1=v_2$, then $v$ lies both to the left and
to the right of $\Isec$.

We note the following straightforward observation.
\begin{lemma}\label{lem:pmc:K}
If some vertex of $N_{G+F}[w]$ lies to the left of $\Isec$, then $\model(\Ibeg{w}) \leq p$.
If some vertex of $N_{G+F}[w]$ lies to the right of $\Isec$, then $\model(\Iend{w}) \geq p+1$.
In particular, if both these events happen, $w$ belongs to $\Isec$.
\end{lemma}

Define now the following sets.
\begin{align*}
F_i^\circ &= \{v \in V(G): vc_i \in F\}\ \mathrm{for}\ i = 1,2; \\
X_1^\circ &= \{v \in V(G): \model(\Iend{c_1}) < \model(\Iend{v}) \leq p+1\}; \\
X_2^\circ &= \{v \in V(G): p \leq \model(\Ibeg{v}) < \model(\Ibeg{c_2})\}.
\end{align*}
As $c_1$ and $c_2$ are cheap, $|F_1^\circ|, |F_2^\circ| \leq \sqrt{k}$.
By the definition of $c_1$ and $c_2$, all vertices of $X_1^\circ \cup X_2^\circ$ are expensive. 
Note that $|X_1^\circ \cap X_2^\circ| \leq 1$ and $X_1^\circ \cap X_2^\circ$ is nonempty only if it consists
of $v_1=v_2$. Therefore $|X_1^\circ| + |X_2^\circ| \leq 2\sqrt{k}+1$.

The following lemma characterizes $\Isec$ in terms of previously defined
vertices and sets, and is a starting point of our algorithm.
\begin{lemma}\label{lem:pmc:char}
$$\Isec = (N_G[\{v_1,c_1,f_1\} \cup X_1^\circ] \cup F_1^\circ)\cap (N_G[\{v_2,c_2,f_2\} \cup X_2^\circ] \cup F_2^\circ).$$
\end{lemma}
\begin{proof}
The inclusion ``$\supseteq$'' follows directly from Lemma~\ref{lem:pmc:K}: vertices of $N_G[\{v_1,c_1,f_1\} \cup X_1^\circ] \cup F_1^\circ$ either are or have at least one neighbor on the left of $\Isec$ in $G+F$, while vertices $N_G[\{v_2,c_2,f_2\} \cup X_2^\circ] \cup F_2^\circ$ either are or have at least one neighbor on the right of $\Isec$ in $G+F$. Hence, we now focus
on the other inclusion.

Without loss of generality, assume there exists a vertex $v \in \Isec$ that does not belong
to $F_2^\circ$ nor to $N_G[\{v_2,c_2,f_2\} \cup X_2^\circ]$. In particular $v \notin \{v_1,v_2,c_2\}$ by Lemma~\ref{lem:pmc:v1v2}, and hence $\Ibeg{v}<p$.
As $v \notin F_2^\circ$ and $vc_2 \notin E(G)$, we have $\model(\Iend{v}) < \model(\Ibeg{c_2})$.
Moreover, by the definition of $X_2^\circ$, $v$ is not adjacent in $G$ to any vertex whose starting event
lies between positions $p$ and $\model(\Ibeg{c_2}) - 1$.
Hence, $v$ is not adjacent in $G$ to any vertex whose starting event lies on or after position $p$.

Consider an ordering $\model'$ that is created from the model $\model$ by moving the event
$\Iend{v}$ to the position just before the event $\Ibeg{v_2}$ (that is, we move $\Iend{v}$
to the position $p$ and shift all events on positions $p$ and later by one to the right).
By our previous arguments, $\model'$ is a valid interval model of some completion
$F'$ of $G$. As $v \in \Isec$, the event $\Iend{v}$ has been moved to the left during this operation,
and $F' \subseteq F$. Moreover $vv_2 \in F \setminus F'$, which contradicts the minimality of $F$.
\end{proof}

We note that, if a polynomial kernel for \textsc{Interval Completion} had been known,
Lemma~\ref{lem:pmc:char} would have finished the proof of Theorem~\ref{thm:pmcs},
as it provides us with a way to enumerate $n^{\Oh(\sqrt{k})}$ candidates for maximal cliques
in $G+F$, by guessing the vertices $v_i,c_i,f_i$ and sets $F_i^\circ$, $X_i^\circ$ for $i=1,2$.%
\footnote{Actually, one may observe that the vertices $f_1$ and $f_2$ are not needed for the argumentation of Lemma~\ref{lem:pmc:char}.
  We include them for convenience, as they will be needed in further arguments.}
However, the question of existence of such a kernel is widely open.
Hence, we need to employ a careful and involved analysis of the structure of the clique
$\Isec$ and the sets defined above to show the following: we may replace brute-force guessing of sets $F_i^\circ$, $X_i^\circ$ with a branching procedure that selects each vertex of $F_i^\circ$, $X_i^\circ$ among $\mathrm{poly}(k)$ potential candidates, instead of $n$.

\subsection{Structure of the recursion}

We now proceed to the description of the algorithm of Theorem~\ref{thm:pmcs}.
The algorithm first iterates through all possible choices of the vertices $v_i,c_i,f_i,g_i$ for $i=1,2$;
for each choice, we seek for maximal cliques where the chosen vertices correspond to their definitions in the previous section.
This step yields the promised $n^8$ factor in the bound on the size of the family $\pmcset$.

Hence, for fixed choice of vertices $v_i,c_i,f_i,g_i$, we aim to output $k^{\Oh(\sqrt{k})}$
sets in the family $\pmcset$. The algorithm now becomes a  branching algorithm:
at each recursive call, in polynomial time we will insert at most one set into the family $\pmcset$, invoke
at most $\mathrm{poly}(k)$ recursive calls, and the depth of the recursion will be bounded
by $\Oh(\sqrt{k})$. Intuitively, we aim to guess the sets $F_i^\circ$ and $X_i^\circ$,
and at each step we want to identify a set of only $\mathrm{poly}(k)$ candidate vertices, such that one of the candidates certainly belongs to one of the sets $F_i^\circ$, $X_i^\circ$.
Thus, we describe the algorithm in the language of ``guessing'' the
maximal clique $\Isec$.

More formally,
during the course of the recursive branching algorithm we keep five sets
$X_1,X_2,F_1,F_2,K \subseteq V(G)$,
and we are looking for maximal cliques $\Isec$ satisfying the following:
\begin{enumerate}
\item $\{v_1,c_1,f_1\} \subseteq X_1 \subseteq X_1^\circ \cup \{v_1,c_1,f_1\}$
and
$\{v_2,c_2,f_2\} \subseteq X_2 \subseteq X_2^\circ \cup \{v_2,c_2,f_2\}$.
\item $F_1 \subseteq F_1^\circ$ and $F_2 \subseteq F_2^\circ$.
\item $(N_G[X_1] \cup F_1) \cap (N_G[X_2] \cup F_2) \subseteq K \subseteq \Isec$.
\end{enumerate}
The set $X_i$ is our ``current guess'' on the set $X_i^\circ \cup \{v_i,c_i,f_i\}$
and the set $F_i$ is our ``current guess'' on the set $F_i^\circ$.
By Lemma~\ref{lem:pmc:char}, already properties 1 and 2 imply
$(N_G[X_1] \cup F_1) \cap (N_G[X_2] \cup F_2) \subseteq \Isec$; the set $K$
is our ``current guess'' for the clique $\Isec$.

However, in some cases we will not be able to guess a vertex of $X_1$ or $X_2$, but
instead we will be guessing its \emph{neighborhood class} with respect to $\Isec$.
The results of Section~\ref{ssec:nei-classes} help us to limit the number of choices
in such a step. 
For this reason, we allow the set $K$ to be a proper superset of 
$(N_G[X_1] \cup F_1) \cap (N_G[X_2] \cup F_2)$, that is, to contain more than the vertices definitely
included in $\Isec$ by Lemma~\ref{lem:pmc:char}.

We initially define $X_1 = \{v_1,c_1,f_1\}$, $X_2 = \{v_2,c_2,f_2\}$, $F_1 = F_2 = \emptyset$
and $K = N_G[X_1] \cap N_G[X_2]$.
It is straightforward to verify that these sets satisfy all aforementioned properties.
We note the following:
\begin{lemma}\label{lem:pmc:missing-bound}
$$|\Isec \setminus (N_G[v_1] \cap N_G[v_2])| \leq k.$$
\end{lemma}
\begin{proof}
Note that for any $v \in \Isec \setminus (N_G[v_1] \cap N_G[v_2])$, either $vv_1$
or $vv_2$ belongs to $F$.
\end{proof}

Let us now focus on one recursive call, where the sets $X_1,X_2,F_1,F_2,K$ are given.
We consider connected components of $G \setminus (X_1 \cup X_2 \cup K)$ and classify them
into four classes, depending on whether they contain a vertex of $N_G(X_1) \cup F_1$
and whether they contain a vertex of $N_G(X_2) \cup F_2$.
That is, we partition the set $\Ccomp{G \setminus (X_1 \cup X_2 \cup K)}$ into
four classes $\Cfam_{ab}$ for $a,b \in \{0,1\}$:
$C \in \Cfam_{10} \cup \Cfam_{11}$ iff $C \cap (N_G(X_1) \cup F_1) \neq \emptyset$ and
$C \in \Cfam_{01} \cup \Cfam_{11}$ iff $C \cap (N_G(X_2) \cup F_2) \neq \emptyset$.

\subsection{Case one: components knowing both sides of the clique}

Assume there exists $C \in \Cfam_{11}$.
Note that $v_1,v_2 \notin C$, since $v_1\in X_1$ and $v_2\in X_2$. Hence, by Lemma~\ref{lem:pmc:K}, $C$ contains a vertex whose interval
starts before position $p$ in the model $\model$, and a vertex whose interval ends after position $p+1$.
As $G[C]$ is connected, $C \cap (\Isec \setminus K) \neq \emptyset$.

Let $P$ be a shortest path between $N_G(X_1) \cup F_1$ and $N_G(X_2) \cup F_2$ in the subgraph $G[C]$.
Note that $P$ contains at least two vertices, as otherwise the single vertex of $P$ should be included in $K$.
We note the following.
\begin{lemma}\label{lem:pmc:P-cases}
Either $V(P) \subseteq \Isec$ or $V(P)$ contains a vertex of
$(F_1^\circ \setminus F_1) \cup (X_1^\circ \setminus X_1) \cup (F_2^\circ \setminus F_2) \cup (X_2^\circ \setminus X_2)$.
\end{lemma}
\begin{proof}
Assume there exists $v \in V(P) \setminus \Isec$. Without loss of generality, assume
that $v$ is to the right of $\Isec$, that is, $\model(\Ibeg{v}) > p+1$ (as $v \notin \{v_1,v_2\}$).
Moreover, assume that $v$ is the first vertex on the path $P$ (when traversed from $N_G(X_1) \cup F_1$ to $N_G(X_2) \cup F_2$) that lies to the right of $\Isec$.

As the first vertex of $P$ belongs to $N_G(X_1) \cup F_1$, $v$ is not the first vertex of $P$.
Let $w$ be the predecessor of $v$ on the path $P$. 
Since $w$ does not lie to the right of $\Isec$ (by the choice of $v$), and $vw \in E(G)$, we infer
that $w \in \Isec$.
As $P$ is a shortest path between $N_G(X_1) \cup F_1$ and $N_G(X_2) \cup F_2$, we have $w \notin F_2$ and
$wc_2 \notin E(G)$.

If $\model(\Iend{w}) \geq \model(\Ibeg{c_2})$ then $wc_2 \in F$, but $w \notin F_2$. Hence, $w \in F_2^\circ \setminus F_2$.
Otherwise, if $\model(\Iend{w}) < \model(\Ibeg{c_2})$, then we have $p+1 < \model(\Ibeg{v}) < \model(\Iend{w}) < \model(\Ibeg{c_2})$. By the choice
of $c_2$, we infer that $v \in X_2^\circ$. Clearly $v \notin X_2$, so $v\in X_2^\circ\setminus X_2$ and the lemma is proven.
\end{proof}

Lemma~\ref{lem:pmc:P-cases} enables us to do a good branching providing that $P$ is short. Luckily, this is always the case.
\begin{lemma}\label{lem:pmc:P-short}
$|V(P)| \leq 3k$.
\end{lemma}
\begin{proof}
Denote $H = G+F$.
Let $R$ be a shortest path between the first and the last vertex of $P$ in the graph $H[V(P)]$.
We first claim that each vertex on $R$ is touched by the solution $F$ and, consequently, $|V(R)| \leq 2k$.

Clearly, each vertex $v \in V(R) \cap \Isec$ is touched by $F$, as $vv_1$ or $vv_2$ needs to belong to $F$.
Consider then $v \in V(R) \setminus \Isec$ and, without loss of generality, assume that $v$ lies to the left of $\Isec$,
that is, $\model(\Iend{v}) < p$.
We now show that $\model(\Iend{v}) > \model(\Iend{f_1})$; this would prove the claim as then $v$ is touched by the definition of $f_1$.
Assume otherwise. Clearly, $v$ is not the last vertex of $P$ (and $R$), and the vertex $w$ succeeding $v$ on $R$ needs to
satisfy $\model(\Ibeg{w}) \leq \model(\Iend{f_1})$.
Consequently, there exists a vertex $w'$ on $R$ that lies later than $v$ on $R$, and which neighbors $f_1$ in $H$.
As $f_1$ is untouched, we have that $w'f_1 \in E(G)$, which means that $w'\in N_G(X_1)$. Since $w'$ is not the first vertex of $P$, this contradicts the choice of $P$.

To finish the proof we now show that $|V(P)| - |V(R)| \leq |F| \leq k$. 
Let $s = |V(P)|$ and $x_1,x_2,\ldots,x_s$ be the vertices
of $P$ in the order of their appearance.
The essence of the proof lies in the fact that whenever $R$ uses some edge $x_ax_b \in F$, $a < b$,
then $F$ needs to contain a triangulation of the cycle $x_a-x_{a+1}-\ldots-x_b-x_a$, consisting of $(b-a-2)$ edges.
Thus, we need to ``pay'' with $(b-a-1)$ edges of $F$ (including $x_ax_b$) to shorten the length of $P$ by, again, $(b-a-1)$.
The formal argumentation follows.

Define the sequence $a_1,a_2,\ldots,a_r$ as follows.
Let $a_1 = 1$ and, given $1 \leq a_i < s$, define $a_{i+1}$ to be such an index, such that
$x_{a_{i+1}}$ is the vertex from the set $\{x_{a_i + 1}, x_{a_i + 2}, \ldots, x_s\}$ that appears earliest on the path $R$.
Clearly, by the definition, $x_{a_{i+1}}$ appears on $R$ later than $x_{a_i}$ and $a_i < a_{i+1}$.
This definition ends when $a_r = s$ for some index $r$.

Consider now an edge $e_{i+1} := x_{b_{i+1}}x_{a_{i+1}}$ on the path $R$, that is, $x_{b_{i+1}}$ is the predecessor of $x_{a_{i+1}}$ on $R$. Clearly $b_{i+1}\leq a_i$, since otherwise $b_{i+1}$ would be a better candidate for $a_{i+1}$.
If $e_{i+1} \in E(G)$, then we have $b_{i+1} = a_i = a_{i+1}-1$ since $P$ is an induced path in $G$.
Otherwise, $e_{i+1} \in F$.
By the definition of $a_{i+1}$, all internal vertices $x_b$ of $R[x_{a_i}, x_{a_{i+1}}]$ satisfy $b < a_i$, as otherwise they would be better candidates for $a_{i+1}$.
Hence, as $P$ is an induced path in $G$ and $R$ is an induced path in $H=G+F$, $F$ needs to contain a triangulation of the cycle consisting of the subpath $R[x_{a_i}, x_{a_{i+1}}]$
and the subpath $P[x_{a_i},x_{a_{i+1}}]$. This triangulation consists of at least $(a_{i+1} - a_i - 2)$ edges.
Moreover, since $R$ is an induced path in $H = G+F$, all the edges of the triangulation needs to have at least one endpoint
in the set $\{x_{a_i + 1}, x_{a_i + 2}, \ldots, x_{a_{i+1}-1}\}$; note that the second endpoint always lies in the set
$\{x_1,x_2,\ldots,x_{a_{i+1}}\}$.
Together with the edge $e_{i+1}$, we infer that there are at least $(a_{i+1} - a_i - 1)$ edges $x_ax_b$ of $F$ such that $a < b$
and $a_i < b \leq a_{i+1}$. Note that this statement also trivially holds in the first case, when $e_{i+1} \in E(G)$.

Observe that the specified set of edges of $F$ are pairwise disjoint for different edges $e_{i+1}$. We infer that
$$|V(P)| - |V(R)| \leq s - r = \sum_{i=1}^{r-1} (a_{i+1}-a_i-1) \leq |F| \leq k,$$
and the lemma is proven.
\end{proof}

Lemmata~\ref{lem:pmc:P-cases} and~\ref{lem:pmc:P-short} enable us to perform the following branching strategy.
In a loop, as long as $\Cfam_{11}$ is not empty, we pick arbitrary $C \in \Cfam_{11}$,
compute a shortest path $P$ in $G[C]$ between $N_G(X_1) \cup F_1$ and $N_G(X_2) \cup F_2$,
and proceed as follows.
First, if the bound of Lemma~\ref{lem:pmc:P-short} does not hold, that is, if $|V(P)| > 3k$, then we conclude that the current guesses are incorrect and we terminate the current branch.
Second, we invoke at most $4|V(P)|$ recursive calls (branches),
in each branch assigning one of the vertices $v \in V(P)$ to one of the sets $F_1$, $F_2$, $X_1$, $X_2$ that does not contain $v$ already.
Third, we put the entire $V(P)$ into $K$ and go back to the beginning of the loop.
By Lemma~\ref{lem:pmc:missing-bound}, we may terminate the current branch if the size of the set $K$ increased
by more than $k$ since the root of the recursion.
Consequently, by the bound of Lemma~\ref{lem:pmc:P-short}, the aforementioned loop produces $\Oh(k^2)$ recursive calls, and leaves
us with a situation where $\Cfam_{11} = \emptyset$.

\subsection{Case two: components not knowing any side of the clique}

We now focus on a component $C \in \Cfam_{00}$, that is, a connected component
of $G \setminus (X_1 \cup X_2 \cup K)$
that does not contain any vertices of $N_G(X_1 \cup X_2) \cup F_1 \cup F_2$.
In particular, note that for any such component it holds that $N_G(C) \subseteq K \setminus \{v_1,v_2\}$.

We now prove a few properties of such components $C$, assuming $C \cap \Isec \neq \emptyset$.
Our goal is to prove that each such component contains a vertex of $F_1^\circ \cup X_1^\circ \cup F_2^\circ \cup X_2^\circ$, and, moreover, both the sizes and the number of candidates for such components are bounded polynomially in $k$.

\begin{lemma}\label{lem:pmc:neiK-between}
If $C \in \Cfam_{00}$ and 
$C \cap \Isec \neq \emptyset$, then $\model(\Iend{f_1}) < \Sbeg{C}{\model} < \Send{C}{\model} < \model(\Ibeg{f_2})$.
\end{lemma}
\begin{proof}
Recall that $f_1$ and $f_2$ are untouched by the solution $F$,
both belong to $X_1 \cup X_2$, and $C$ does not contain
any neighbor of $X_1 \cup X_2$.
\end{proof}

\begin{lemma}\label{lem:pmc:neiK-touched}
If $C \in \Cfam_{00}$ and 
$C \cap \Isec \neq \emptyset$, then all vertices of $C$ are
touched by the solution, and, consequently, $|C| \leq 2k$.
\end{lemma}
\begin{proof}
Let $v \in C$. If $v \in \Isec$, $v$ is touched by $F$
as $vv_1,vv_2 \in F$. If $v$ lies to the left of $\Isec$ then,
by Lemma~\ref{lem:pmc:neiK-between}, $\model(\Iend{v}) > \model(\Iend{f_1})$,
and $v$ is touched by the choice of $f_1$.
The case of $v$ lying to the right of $\Isec$ is symmetrical.
\end{proof}

\begin{lemma}\label{lem:pmc:neiK-large-v}
If $C \in \Cfam_{00}$ and 
$C \cap \Isec \neq \emptyset$, then there exists $v \in C$
such that $|K \setminus N_G(v)| \leq k$.
\end{lemma}
\begin{proof}
Observe that any vertex of $C \cap \Isec$ needs to be adjacent
to all vertices of $K$ in $G+F$, and $|F| \leq k$.
\end{proof}

\begin{lemma}\label{lem:pmc:neiK-interesting}
If $C \in \Cfam_{00}$ and 
$C \cap \Isec \neq \emptyset$, then $C$ contains a vertex of
$(F_1^\circ \setminus F_1) \cup (X_1^\circ \setminus X_1) \cup (F_2^\circ \setminus F_2) \cup (X_2^\circ \setminus X_2)$.
\end{lemma}
\begin{proof}
We first show that $C \not\subseteq \Isec$. Assume the contrary.
Let $|C|=s$ and $x_1,x_2,\ldots,x_s$ be the vertices of $C$.
Consider a model
$\model'$ created from $\model$ by taking out all events
of $\events{C}$ and inserting them, in the order
$\Ibeg{x_1},\Ibeg{x_2},\ldots,\Ibeg{x_s},\Iend{x_s},\Iend{x_{s-1}},\ldots,\Iend{x_1}$ between positions $p-1$ and $p$ (i.e., just
before the event $\Ibeg{v_2}$ at position $p$.
As $N_G(C) \subseteq K \setminus \{v_1,v_2\}$, $\model'$ is a valid
interval model of some completion $F'$ of $G$.
As $C \subseteq \Isec\setminus \{v_1,v_2\}$ and in particular $C$ is a clique in $G+F$, for any $x_i \in C$ we have
$\model(\Ibeg{x_i}) < p < \model(\Iend{x_i})$ and, consequently,
$F' \subseteq F$. Moreover, $x_iv_2 \in F \setminus F'$ for any
$x_i \in C$, contradicting the minimality of $F$.

Since $C$ is connected in $G$, we may pick $v,w \in C$ such that
$vw \in E(G)$, $v \in \Isec$ and $w \notin \Isec$;
w.l.o.g. assume that $w$ lies to the left of $\Isec$.
If $\model(\Ibeg{v}) \leq \model(\Iend{c_1})$ then
$vc_1 \in F$ and $v \in F_1^\circ \setminus F_1$.
Otherwise, we have $\model(\Iend{c_1}) < \model(\Iend{w}) < p$
and $w \in X_1^\circ \setminus X_1$.
This finishes the proof of the lemma.
\end{proof}

By Lemmata~\ref{lem:pmc:neiK-touched} and~\ref{lem:pmc:neiK-large-v},
all components $C \in \Cfam_{00}$
that may have a nonempty intersection with $\Isec$ need to (a) be
of size at most $2k$ and (b) have a vertex with at most $k$ non-neighbors
in $K$. By Lemma~\ref{lem:A-r-nei}, applied to the set $A := K$
and parameter $r := k$, 
in a YES-instance we expect $\Oh(k^2)$ components satisfying the second
requirement.
(Formally, we conclude that $(G,k)$ is a NO-instance and return $\pmcset = \emptyset$
 if the bound of Lemma~\ref{lem:A-r-nei} turns out to be violated.)
Consequently, all components satisfying
both requirements (a) and (b) have $\Oh(k^3)$ vertices in total.
This, together with Lemma~\ref{lem:pmc:neiK-interesting},
motivates the following branching step. First, we invoke
$\Oh(k^3)$ recursive calls, in each call picking a vertex
from a component satisfying both (a) and (b)
and inserting it into one of the sets $F_1$, $X_1$, $F_2$, $X_2$.
Finally,  we pass the instance to the next case,
assuming that no component of $\Cfam_{00}$
contains a vertex of $\Isec$.

\subsection{Case three: components knowing one side of the clique}

We are left with the components of $\Cfam_{01} \cup \Cfam_{10}$.
By symmetry, we may focus on $\Cfam_{10}$ only.

Consider $C \in \Cfam_{10}$.
The main obstacle we obtain in this section is that an analogue
of Lemma~\ref{lem:pmc:neiK-touched} does not hold
(in particular $C$ may contain a lot of vertices in $N_G(f_1)$)
and, consequently, $C$ may be large.
To apply arguments similar to the previous case, we need to further analyze the structure
of such component $C$.

To this end, we define $\Dfam_1 = \Ccomp{G[\bigcup \Cfam_{10} \setminus N_G(f_1)]}$.
Now, for each $D \in \Dfam_1$ we have
not only $D \cap (N_G(X_2) \cup F_2) = \emptyset$ but also $D \cap N_G(f_1) = \emptyset$,
and we can state analogues of Lemmata~\ref{lem:pmc:neiK-between} and~\ref{lem:pmc:neiK-touched}.

\begin{lemma}\label{lem:pmc:3-between}
For any $D \in \Dfam_1$ either $\Send{D}{\model} < \model(\Ibeg{f_1})$ or
$\model(\Iend{f_1}) < \Sbeg{D}{\model} < \Send{D}{\model} < \model(\Ibeg{f_2})$.
Moreover, if the second option happens, then all vertices of $D$ are touched by $F$ and $|D| \leq 2k$.
\end{lemma}
\begin{proof}
As $D$ is connected and does not contain any neighbor of the untouched vertices
$f_1$ and $f_2$, we need only to exclude the possibility $\Sbeg{D}{\model} > \model(\Iend{f_2})$.
However, this clearly follows from the fact that there exists a connected component $C \in \Cfam_{10}$
containing $D$: $N_{G+F}(C)$ contains a vertex of $X_1$ and does not contain $f_2$. This proves the first assertion of the lemma.

Assume now that 
$\model(\Iend{f_1}) < \Sbeg{D}{\model} < \Send{D}{\model} < \model(\Ibeg{f_2})$.
Pick any $v \in D$. If $v \in \Isec$, then $v$ is touched by $F$ as $vv_2 \in F$.
Otherwise $\model(\Iend{f_1}) < \model(\Iend{v}) < p$ or $\model(\Ibeg{f_2}) > \model(\Ibeg{v}) > p+1$. In both cases $v$ is touched by the choice of $f_1$ or $f_2$.
\end{proof}

The following lemma shows formally why we are interested in components of $\Dfam_1$.
\begin{lemma}\label{lem:pmc:3-D}
A component $C \in \Cfam_{10}$ contains an element of $\Isec$ if and only if
there exists $D \in \Dfam_1$, $D \subseteq C$, such that $D \cap \Isec \neq \emptyset$
or $\Sbeg{D}{\model} > p+1$.
In particular, such a component $D$ satisfies the second option of Lemma~\ref{lem:pmc:3-between}.
\end{lemma}
\begin{proof}
Assume first that such a component $D$ exists for some $C \in \Cfam_{10}$.
If $D$ contains a vertex of $\Isec$, then clearly so does $C$, so assume $\Sbeg{D}{\model} > p+1$.
Then $N_{G+F}(D) \cap X_1 = \emptyset$ but $N_{G+F}(C) \cap X_1 \neq \emptyset$.
Hence, as $G[C]$ is connected and $D$ is a connected component of $G[C]\setminus N_G(f_1)$, we infer that there exists some $z \in N_G(D) \cap N_G(f_1)$. Such a $z$ clearly belongs to $\Isec$ by Lemma~\ref{lem:pmc:K}.

In the other direction, assume that $C\cap \Isec\neq \emptyset$. Suppose first that there exists $x \in C$ with $\model(\Ibeg{x}) > p+1$.
Then $x \notin N_G(f_1)$ and $x \in D$ for some $D \in \Dfam_1$. If $D \cap \Isec \neq \emptyset$ we are done.
Otherwise, by the connectivity of $D$ we have $\Sbeg{D}{\model} > p+1$ and the claim is proven.

So we have $\model(\Ibeg{x}) < p$ for any $x \in C$, as $v_1,v_2 \notin C$. Consider an interval model
$\model'$ created from $\model$ by taking all events of $\events{C}$ that are placed at positions at least $p$,
and putting them (in the same order) just before position $p$ (i.e., between positions $p-1$ and $p$). 
As $N_G(C) \subseteq (X_1 \cup K) \setminus \{v_2\}$, this is a valid interval model of $G+F'$ for some completion $F'$.
As $\model(\Ibeg{x}) < p$ for any $x \in C$, we have $F' \subseteq F$. Moreover, $xv_2 \in F \setminus F'$ for any
$x \in C \cap \Isec$. By the minimality of $F$ we have $C \cap \Isec = \emptyset$, which contradicts our assumption about $C$ and concludes the proof.
\end{proof}

Hence, we now focus on components $D$ and try to deduce which of them may possibly satisfy one of
the conditions imposed in Lemma~\ref{lem:pmc:3-D}.
We first make use of the untouched vertex $g_1$ to filter out some clearly ``useless'' components of $\Dfam_1$.
\begin{lemma}\label{lem:pmc:3-g}
If for $D \in \Dfam_1$ we have $D \cap N_G(g_1) \neq \emptyset$ then $\Send{D}{\model} < \model(\Ibeg{f_1})$ (i.e., the first option of Lemma~\ref{lem:pmc:3-between} happens).
\end{lemma}
\begin{proof}
Follows directly from the inequality $\model(\Iend{g_1}) \leq \model(\Iend{f_1})$ (Lemma~\ref{lem:pmc:vtx-pos}).
\end{proof}

We denote $\Dfam_2 = \{D \in \Dfam_1: g_1 \notin N_G(D)\}$
and define $Z = \bigcup_{D \in \Dfam_2} N_G(D) \setminus (K \cup X_1)$.
Note that $N_G(D) \subseteq X_1 \cup K \cup N_G(f_1)$ by the definition of $\Cfam_{10}$ and $\Dfam_1$.
Consequently, $Z \subseteq N_G(f_1)\cap \bigcup \Cfam_{10}$.
The following observation is the main reason to introduce the vertex $g_1$ and ``filter out'' components
of $\Dfam_1 \setminus \Dfam_2$ in Lemma~\ref{lem:pmc:3-g}.
\begin{lemma}\label{lem:pmc:3-Z-small}
All vertices of $Z$ are touched by $F$ and, consequently, $|Z| \leq 2k$.
\end{lemma}
\begin{proof}
Let $z \in Z$ and let $D \in \Dfam_2$ such that $z \in N_G(D)$.
If $z \in \Isec$ then $zv_2 \in F$ and $z$ is touched, so assume otherwise. As $z \in N_G(f_1)$ we infer that $\model(\Iend{z}) < p$.

Consider two cases for component $D$ given by Lemma~\ref{lem:pmc:3-between}.
If $\Send{D}{\model} < \model(\Ibeg{f_1})$ then, as $D \in \Dfam_2$ and $g_1 \in N_G[f_1]$, we have
actually $\Send{D}{\model} < \model(\Ibeg{g_1})$. Hence, $\model(\Ibeg{z}) < \model(\Ibeg{g_1})$.
As $z \in N_G(f_1)$ and $z \notin \Isec$, we infer that $z$ is touched by the choice of $g_1$.
In the second case, if $\model(\Iend{f_1}) < \Sbeg{D}{\model}$ then $\model(\Iend{z}) > \model(\Iend{f_1})$.
As $\model(\Iend{z}) < p$, we infer that $z$ is touched by the choice of $f_1$.
\end{proof}

Formally, if the bound of Lemma~\ref{lem:pmc:3-Z-small} does not hold, we terminate the current branch.
Otherwise, any $D \in \Dfam_2$ satisfies $N_G(D) \subseteq K \cup X_1 \cup Z$, and $|Z| + |X_1| \leq 2k + \Oh(\sqrt{k})$.

We now focus on the possibility of $D \cap \Isec \neq \emptyset$ for some $D \in \Dfam_2$.
\begin{lemma}\label{lem:pmc:3-DIsec}
If $D \cap \Isec \neq \emptyset$ for some $D \in \Dfam_2$, then
$D \cap ((F_2^\circ \setminus F_2) \cup (X_2^\circ \setminus X_2)) \neq \emptyset$.
\end{lemma}
\begin{proof}
We first show that if $D \cap \Isec \neq \emptyset$ then there exists $w \in D$ with $\model(\Ibeg{w}) > p+1$.
Assume the contrary, and consider a model $\model'$ created from $\model$ by taking all events of $\events{D}$
that are placed by $\model$ on positions to the right of $\Isec$ (i.e., at positions with numbers at least $p$)
and move them just before position $p$ (i.e., the event $\Ibeg{v_2}$), in the same order as they appear in $\model$.
As $N_G(D) \subseteq X_1 \cup N_G(f_1)$, $\model'$ is an interval model of some completion $F'$ of $G$.
Since we supposed that no vertex of $D$ starts in $\model$ after position $p$, we have $F' \subseteq F$. Moreover, $vv_2 \in F \setminus F'$
for any $v \in D \cap \Isec$, a contradiction to the minimality of $F$.

By the connectivity of $D$, there exist $v,w \in D$ such that $vw \in E(G)$, $v \in \Isec$, and $\model(\Ibeg{w}) > p+1$.
Consider two cases. If $\model(\Iend{v}) \geq \model(\Ibeg{c_2})$ then $vc_2 \in F$ and $v \in F_2^\circ \setminus F_2$.
Otherwise we have $\model(\Ibeg{w}) < \model(\Iend{v}) < \model(\Ibeg{c_2})$, and hence, by the choice of $c_2$, $w$ is expensive. Consequently
  $w \in X_2^\circ \setminus X_2$.
\end{proof}

We now note that
if $D \cap \Isec \neq \emptyset$, then any $v \in D \cap \Isec$ needs to satisfy $|K \setminus N_G(v)| \leq k$.
Let $\Dfam_3 \subseteq \Dfam_2$ be the family of these connected components $D$ of $\Dfam_2$ that
(a) have size at most $2k$, and (b) contain a vertex $v$ that has at most $k$ non-neighbors in $K$.
By Lemma~\ref{lem:pmc:3-between}, if $D \cap \Isec \neq \emptyset$ then $D \in \Dfam_3$.
By Lemma~\ref{lem:A-r-nei} applied to the set $A := K \cup X_1 \cup Z$ and
$r = k + |Z| + |X_1| = \Oh(k)$, we infer that in a YES-instance we expect
$|\Dfam_3| = \Oh(k^2)$ (formally, we terminate the algorithm and return $\pmcset = \emptyset$
if this is not the case). Consequently, $|\bigcup \Dfam_3| = \Oh(k^3)$.
Hence, Lemma~\ref{lem:pmc:3-DIsec} allows us 
to branch into $\Oh(k^3)$ recursive calls: in each call we put one of the vertices
of $\bigcup \Dfam_3$ into one of the sets $F_2$, $X_2$.
We proceed further with the assumption that no vertex of $\bigcup \Dfam_2$ belongs to $\Isec$,
and we focus on the possibility that $\Sbeg{D}{\model} > p+1$ for some $D \in \Dfam_2$.

\begin{lemma}\label{lem:pmc:3-Dright}
If $\Sbeg{D}{\model} > p+1$ for some $D \in \Dfam_2$, then
either $Z \cap (F_2^\circ \setminus F_2) \neq \emptyset$ or there exists $w \in D \cap (X_2^\circ \setminus X_2)$
 such that $N_G(w) \cap Z = N_G(w) \cap (\Isec \setminus K) \neq \emptyset$.
\end{lemma}
\begin{proof}
First note that, as $\Sbeg{D}{\model} > p+1$, then $N_G(D) \subseteq K \cup Z$, and $D$ does not contain any vertex of $F_1^\circ$.
Moreover, as $D \subseteq C$ for some $C \in \Cfam_{10}$, we have that $N_G(D) \cap Z \neq \emptyset$.

Pick any $z \in N_G(D) \cap Z$. As $zf_1 \in E(G)$ and $\Sbeg{D}{\model} > p+1$, we have $z \in \Isec \setminus K$.
If $\model(\Iend{z}) \geq \model(\Ibeg{c_2})$, then we have $z \in F_2^\circ \setminus F_2$ and we are done.
Otherwise, any neighbor $w \in N_G(z) \cap D$ satisfies $\model(\Ibeg{w}) < \model(\Iend{z}) < \model(\Ibeg{c_2})$ and, by the choice of $c_2$, we infer that
$w \in X_2^\circ \setminus X_2$.
As $N_G(w) \subseteq D \cup K \cup Z$, such $w$ satisfies the requirements of the lemma; the fact that $N_G(w) \cap Z = N_G(w) \cap (\Isec \setminus K)$ follows easily from the assumptions about $D$ and the definition of $Z$.
\end{proof}

Lemma~\ref{lem:pmc:3-Dright}, together with the bound $|Z| \leq 2k$ of Lemma~\ref{lem:pmc:3-Z-small},
allows us to perform the following branching. In the first $|Z|$ recursive calls we pick a vertex of $Z$
and insert it into $F_2$.
Then, we invoke Lemma~\ref{lem:A-nei} on the set $A := Z$, expecting $\Oh(k^2)$ neighborhood classes
w.r.t. $Z$ in the graph $G$ (formally, if this is not the case, we conclude that $(G,k)$
is a NO-instance and return an empty set $\pmcset$). 
We branch into $\Oh(k^2)$ subcases, in each recursive call picking a neighborhood class $R$
w.r.t. $Z$ with nonempty neighborhood $N_G(R) \cap Z$ and inserting this neighborhood into $K$.

Finally, we are left with the case where the conclusion is that no component $D \in \Dfam_2$ satisfies
$\Sbeg{D}{\model} > p+1$; recall that we have already concluded before that no component $D \in \Dfam_2$ has a nonempty intersection with $\Isec$. By Lemma~\ref{lem:pmc:3-D} we infer that in fact there are no vertices of $\Isec$ at all in the components of $\Cfam_{10}$.

Therefore, we pass the instance to the symmetric case
of $\Cfam_{01}$ and we perform all the symmetric branchings. In the remaining subcase,
we can finally conclude that $K = \Isec$: We have $\Cfam_{11}=\emptyset$, and we have already concluded that there are no vertices of $\Isec$ in the components of $\Cfam_{00}$, of $\Cfam_{10}$, nor of $\Cfam_{01}$. Hence we insert the set $K$ into the constructed family $\pmcset$.

It remains to argue that we output $k^{\Oh(\sqrt{k})}$ sets for each choice
of the vertices $v_i,c_i,f_i,g_i$, $i=1,2$. Clearly, each step of the recursion invokes
$\mathrm{poly}(k)$ recursive calls. To see that the depth of the recursion can be bounded by
$\Oh(\sqrt{k})$, note that whenever we make a recursive call, we either insert a new vertex
into one of the sets $F_1$, $X_1$, $F_2$, $X_2$, or
we put into $K$ all vertices of a non-empty set $N_G(w) \cap (\Isec \setminus K)$ for some $w \in (X_1^\circ \setminus X_1) \cup (X_2^\circ \setminus X_2)$ --- hence this step can be done at most once for every $w\in X_1^\circ\cup X_2^\circ$ during the whole branching process.
As $|F_1^\circ|,|F_2^\circ| \leq \sqrt{k}$ and $|X_1^\circ| + |X_2^\circ| \leq 2\sqrt{k}+1$,
we can prune the recursion tree at depth $6\sqrt{k}+2$, obtaining the claimed bound
on the size of $\pmcset$.
This concludes the proof of Theorem~\ref{thm:pmcs}.

\section{Guessing fill-in edges with fixed endpoint}\label{sec:fill-in}
In this section we prove the following result.
\begin{theorem}\label{thm:fill-in}
Given an \icname{} instance $(G,k)$, where the Module Reduction Rule is not applicable,
and a designated vertex $v \in V(G)$, 
one can in $k^{\Oh(\sqrt{k})}n^{\Oh(1)}$ time enumerate a family
$\fiset$ of at most $k^{\Oh(\sqrt{k})} n^{70}$ subsets of $V(G)$,
each of size $\Oh(k^5)$,
satisfying the following: for any minimal solution $F$ to $(G,k)$ there exists some $B \in \fiset$
such that $w \in B$ whenever $vw \in F$.
\end{theorem}

We will mostly use Theorem~\ref{thm:fill-in} to guess the incident fill-in edges
of a cheap vertex.

\begin{corollary}\label{cor:cheap-fill-in}
Given an \icname{} instance $(G,k)$, where the Module Reduction Rule is not applicable,
and a designated vertex $v \in V(G)$, 
one can in $k^{\Oh(\sqrt{k})}n^{\Oh(1)}$ time enumerate a family $\fiset'$
of at most $k^{\Oh(\sqrt{k})} n^{70}$ subsets of $V(G)$, such that
for any minimal solution $F$ to $(G,k)$ for which $v$ is cheap w.r.t. $F$,
the set $\{w \in V(G): vw \in F\}$ belongs to $\fiset'$.
\end{corollary}
\begin{proof}
We first enumerate the family $\fiset$ of Theorem~\ref{thm:fill-in} and then define
$$\fiset' = \{A \subseteq V(G): |A| \leq \sqrt{k} \wedge \exists_{B \in \fiset} A \subseteq B\}.$$
The correctness and the size bound follows directly from Theorem~\ref{thm:fill-in}.
\end{proof}

We remark that, similarly as in the previous section, a polynomial kernel for \icname{} would save us 
a lot of effort. In fact, Theorem~\ref{thm:fill-in} becomes obvious as we could then return $\fiset = \{V(G)\}$,
(possibly worsening the polynomial bound on the size of a single element of $\fiset$).
However, the question of existence of a polynomial kernel for \icname{} remains widely open, and we need to employ
a careful analysis to obtain the promised results.

\subsection{Important vertices and sections}

\begin{figure}
\centering
\includegraphics{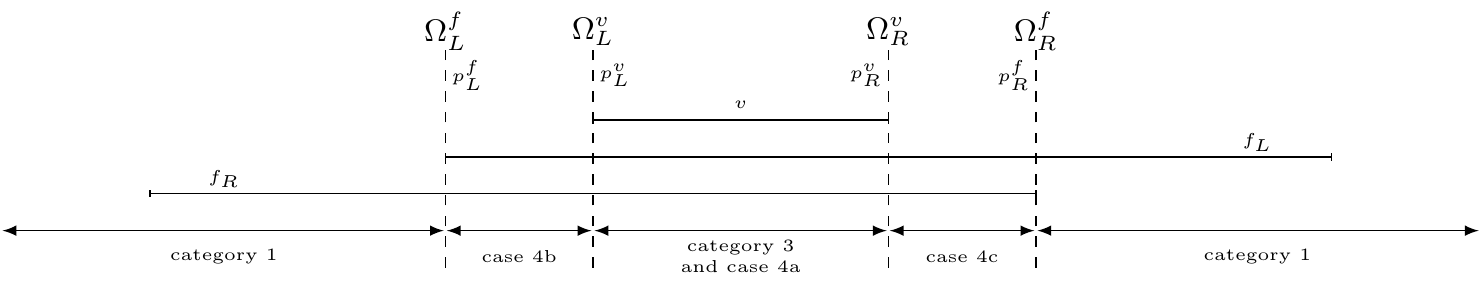}
\caption{Situation around the vertex $v$ in the proof of Theorem~\ref{thm:fill-in},
  together with categories and cases of Lemma~\ref{lem:fi:structure}.}
\label{fig:fill-in}
\end{figure}

We fix a minimal completion $F$ of the \icname{} instance $(G,k)$, and a model $\sigma$ of $G+F$.
We define the following (see also Figure~\ref{fig:fill-in}).
\begin{enumerate}
\item Denote $p_L^v = \model(\Ibeg{v})$ and $p_R^v = \model(\Iend{v})$.
\item Let $f_L$ be the untouched vertex with the rightmost starting endpoint
among untouched vertices $f$ satisfying
$\model(\Ibeg{f}) \leq p_L^v < p_R^v \leq \model(\Iend{f})$.
\item Let $f_R$ be the untouched vertex with the leftmost ending endpoint
among untouched vertices $f$ satisfying
$\model(\Ibeg{f}) \leq p_L^v < p_R^v \leq \model(\Iend{f})$.
\item Denote $p_L^f = \model(\Ibeg{f_L})$ and $p_R^f = \model(\Iend{f_R})$.
\item Denote $\Isec_L^f = \Isec_\model(p_L^f)$, $\Isec_L^v = \Isec_\model(p_L^v)$,
  $\Isec_R^v = \Isec_\model(p_R^v-1)$ and $\Isec_R^f = \Isec_\model(p_R^f-1)$.
\end{enumerate}
Note that $\Uroot$ is a good candidate for both $f_L$ and $f_R$, thus these vertices exist.
We remark also that it may happen that $v=f_L$, $v=f_R$ or $f_L = f_R$. However, we may say
the following about the order of these vertices.
\begin{lemma}\label{lem:fi:order}
$\model(\Ibeg{f_R}) \leq p_L^f \leq p_L^v < p_R^v \leq p_R^f \leq \model(\Iend{f_L})$.
\end{lemma}
\begin{proof}
The first and the last inequalities follow from the fact that $f_R$ is a good candidate for $f_L$ and vice-versa.
The remaining inequalities are straightforward from the definition.
\end{proof}

We start by enumerating all possible choices of vertices $f_L,f_R$ and sections
$\Isec_L^f$, $\Isec_L^v$, $\Isec_R^v$, $\Isec_R^f$, using the family $\sectionset$
of Theorem~\ref{thm:sections}.
By the bound of Theorem~\ref{thm:sections}, there are at most $k^{\Oh(\sqrt{k})} n^{70}$
subcases (henceforth called \emph{branches}) to consider.
In the rest of the proof we aim to output a single set $B$ of size $\Oh(k^5)$
for a single choice of the aforementioned two vertices and four sections.
That is, given $f_L,f_R$ and $\Isec_L^f$, $\Isec_L^v$, $\Isec_R^v$, $\Isec_R^f$
we show how to deduce a set $B \subseteq V(G)$ of size $\Oh(k^5)$, such that
$B$ contains $\{w: vw \in F\}$ for any minimal solution $F$ to $(G,k)$
for which the choice of $f_L, f_R$ and $\Isec_L^f$, $\Isec_L^v$, $\Isec_R^v$, $\Isec_R^f$
is correct.

Thus, henceforth we fix a choice of
$f_L,f_R$ and $\Isec_L^f$, $\Isec_L^v$, $\Isec_R^v$, $\Isec_R^f$
and we assume that the guess of these vertices and sets 
is correct for a minimal solution $F$ with model $\model$ of $G+F$.
We note that, by Lemma~\ref{lem:fi:order}, we should expect that:
\begin{align*}
v &\in \Isec_L^v \cap \Isec_R^v, \\
f_L,f_R &\in \Isec_L^f \cap \Isec_R^f, \\
\Isec_L^f \cap \Isec_R^f &\subseteq \Isec_L^f \cap \Isec_R^v \subseteq \Isec_L^v \cap \Isec_R^v, \\
\Isec_L^f \cap \Isec_R^f &\subseteq \Isec_L^v \cap \Isec_R^f \subseteq \Isec_L^v \cap \Isec_R^v.
\end{align*}
If this is not the case, we discard the branch in question.

Moreover, we maintain a set $B^\mathrm{sure}$ of vertices $w$ for which we deduce
that $vw \in F$ is implied by the choice of 
$f_L,f_R$ and $\Isec_L^f$, $\Isec_L^v$, $\Isec_R^v$, $\Isec_R^f$.
We start with $B^\mathrm{sure} = (\Isec_L^v \cup \Isec_R^v) \setminus N_G(v)$.
If at any point the size of $B^\mathrm{sure}$ exceeds $k$, we discard the current branch.

\subsection{Preliminary observations and categories of connected components}

We start with the following observation, directly implied by the assumption that $f_L$ and $f_R$
are untouched and $|F| \leq k$.
\begin{lemma}\label{lem:fi:structure}
For any connected component $C$ of $G \setminus (\Isec_L^f \cup \Isec_L^v \cup \Isec_R^v \cup \Isec_R^f)$ the following holds
\begin{enumerate}
\item\label{case:fi:outside}
If $C \cap N_G(f_L) \cap N_G(f_R) = \emptyset$, then $\Send{C}{\model} < p_L^f$
or $\Sbeg{C}{\model} > p_R^f$. In particular, $vw \notin E(G) \cup F$ for every $w \in C$.
\item If $C$ contains a vertex of  $N_G(f_L) \cap N_G(f_R)$, then
$p_L^f < \Sbeg{C}{\model} < \Send{C}{\model} < p_R^f$ and
$C \subseteq N_G(f_L) \cap N_G(f_R)$.
\item\label{case:fi:top-forced}
If, moreover, $C$ contains a neighbor of $v$ in $G$, then
$p_L^v < \Sbeg{C}{\model} < \Send{C}{\model} < p_R^v$
and $vw \in E(G) \cup F$ for every $w \in C$.
\item\label{case:fi:dont-know}
In the last case, if $C \subseteq (N_G(f_L) \cap N_G(f_R)) \setminus N_G(v)$, then
one of the following cases hold:
\begin{enumerate}
\item\label{case:fi:top}
$p_L^v < \Sbeg{C}{\model} < \Send{C}{\model} < p_R^v$
and $vw \in F$ for every $w \in C$.
Moreover, in this case $N_G(C) \subseteq \Isec_L^v \cup \Isec_R^v$.
\item \label{case:fi:left}
$p_L^f < \Sbeg{C}{\model} < \Send{C}{\model} < p_L^v$
and $vw \notin F$ for every $w \in C$.
Moreover, in this case $N_G(C) \subseteq \Isec_L^f \cup \Isec_L^v$.
\item \label{case:fi:right}
$p_R^v < \Sbeg{C}{\model} < \Send{C}{\model} < p_R^f$
and $vw \notin F$ for every $w \in C$.
Moreover, in this case $N_G(C) \subseteq \Isec_R^f \cup \Isec_R^v$.
\end{enumerate}
Moreover, if $|C| > k$, then the first option does not happen.
\end{enumerate}
\end{lemma}
By Lemma~\ref{lem:fi:structure}, we can sort the connected components of
$G \setminus (\Isec_L^f \cup \Isec_L^v \cup \Isec_R^v \cup \Isec_R^f)$
into three \emph{categories}, depending on whether they fall into
point~\ref{case:fi:outside}, \ref{case:fi:top-forced}
or~\ref{case:fi:dont-know}. Obviously, the last category is the most interesting, as we are not able to directly decide whether the vertices of the component should be inserted into $B$ or not. 
The subpoints of this category (i.e,~\ref{case:fi:top},~\ref{case:fi:left} and~\ref{case:fi:right})
are henceforth called \emph{cases}. Note that for each connected component $C$
we know its category, but we do not know its case if it falls into category~\ref{case:fi:dont-know}.

We now perform some cleaning. 
If there exists a component
$C \in \Ccomp{G \setminus (\Isec_L^f \cup \Isec_L^v \cup \Isec_R^v \cup \Isec_R^f)}$
that does not fall into any category
(e.g., we have 
$C \not\subseteq N_G(f_L) \cap N_G(f_R)$, but $C$ contains a common neighbor
of $f_L$ and $f_R$), we discard the current branch.
Moreover,
we may include into $B^\mathrm{sure}$ all non-neighbors of $v$ that lie
in a connected component $C$ that falls into category~\ref{case:fi:top-forced}
of Lemma~\ref{lem:fi:structure}, that is, that contains a neighbor of $v$.

Clearly, only at most $k$ components fall into case~\ref{case:fi:top} of Lemma~\ref{lem:fi:structure}, since each such component induces at least one fill edge incident to $v$. However, we do not know which of the components falling into category~\ref{case:fi:dont-know} are in fact those interesting ones. Hence, our main task now is to pinpoint a set of roughly $\Oh(k^4)$ potential components
falling into category~\ref{case:fi:dont-know}
for which case~\ref{case:fi:top} may possibly happen. As each such component
is of size at most $k$, this would conclude the proof of Theorem~\ref{thm:fill-in}.

Let $\Cfam$ be the family of all connected component $C$ of 
$G \setminus (\Isec_L^f \cup \Isec_L^v \cup \Isec_R^v \cup \Isec_R^f)$
that fall into category~\ref{case:fi:dont-know} of Lemma~\ref{lem:fi:structure},
that is, $C \subseteq (N_G(f_L) \cap N_G(f_R)) \setminus N_G(v)$.
We distinguish the following subfamilies that correspond to the subcases of category~\ref{case:fi:dont-know}.
\begin{align*}
\Cfam_v &= \{C \in \Cfam: N_G(C) \subseteq \Isec_L^v \cup \Isec_R^v\} \\
\Cfam_L &= \{C \in \Cfam: N_G(C) \subseteq \Isec_L^f \cup \Isec_L^v\} \\
\Cfam_R &= \{C \in \Cfam: N_G(C) \subseteq \Isec_R^f \cup \Isec_R^v\}
\end{align*}
If $\Cfam_v \cup \Cfam_L \cup \Cfam_R \neq \Cfam$, we discard the current branch.
Moreover, for any $C \in \Cfam_v \setminus (\Cfam_L \cup \Cfam_R)$ we include
all vertices of $C$ into $B^\mathrm{sure}$, as such a component will surely fall into case~\ref{case:fi:top}.

In the sequel we will consider components that belong to different combinations of sets $\Cfam_v,\Cfam_L,\Cfam_R$. The following fact, used often implicitly, follows directly from the definitions of $\Cfam_v,\Cfam_L,\Cfam_R$ and inclusion relations between $\Isec_L^f,\Isec_L^v,\Isec_R^v,\Isec_R^f$. 

\begin{lemma}\label{lem:fi:inclusions}
The following holds:
\begin{itemize}
\item If $C\in \Cfam_L \cap \Cfam_v$ then $N_G(C)\subseteq \Isec_L^v$. If moreover $C\notin \Cfam_R$, then $N_G(C)\cap (\Isec_L^v\setminus \Isec_R^v)\neq \emptyset$.
\item If $C\in \Cfam_R \cap \Cfam_v$ then $N_G(C)\subseteq \Isec_R^v$. If moreover $C\notin \Cfam_L$, then $N_G(C)\cap (\Isec_R^v\setminus \Isec_L^v)\neq \emptyset$.
\item If $C\in \Cfam_L \cap\Cfam_R$, then $N_G(C)\subseteq \Isec_L^v\cap \Isec_R^v$ and in particular $C\in \Cfam_v$.
\end{itemize}
\end{lemma}

\subsection{Troublesome components}

Our goal now is to focus on $\Cfam_L$ and pinpoint a small set
of components of $\Cfam_L \cap \Cfam_v$ that may possibly fall into case~\ref{case:fi:top}
of Lemma~\ref{lem:fi:structure}. The arguments for $\Cfam_R$ will be symmetrical.

To this end, we will construct a family $\Tfam \subseteq \Cfam_L$
of \emph{troublesome} components. Informally speaking, a component is troublesome
if it is highly unclear where or how it should live in the model $\model$.
We will argue that there is a bounded number of troublesome components (strictly speaking, $\Oh(k^2)$ of them)
and any component that falls into case~\ref{case:fi:top} of Lemma~\ref{lem:fi:structure}
is in some sense ``close'' to a troublesome component.

We start by putting into $\Tfam$ all connected components $C \in \Cfam_L$
that cannot be drawn in the model of a completion of $G$ between sections $\Isec_L^f$ and $\Isec_L^v$
without an incident edge of the solution. More formally,
we denote 
$F_L = \binom{\Isec_L^v}{2} \setminus E(G) \subseteq F$
and define the following:
\begin{definition}
A component $C \in \Cfam_L \cap \Cfam_v$ is \emph{freely drawable}
if there exists an interval model $\model_C$ of $(G+F_L)[C \cup \Isec_L^v]$
that starts with all starting events of $\events{\Isec_L^v \cap \Isec_L^f}$\
and ends with all ending events of $\events{\Isec_L^v}$.
\end{definition}
We now state the formerly informal motivation for this definition.
\begin{lemma}\label{lem:fi:not-free}
If $C \in (\Cfam_L \cap \Cfam_v) \setminus \Cfam_R$ is not freely drawable, then it is touched by $F$.
\end{lemma}
\begin{proof}
As $C \notin \Cfam_R$, it cannot fall into case~\ref{case:fi:right} of Lemma~\ref{lem:fi:structure}.
If $C$ falls into case~\ref{case:fi:top} then it is touched due to the fill-in edges incident to $v$.
Otherwise, unless $C$ is touched, the model $\model$ restricted to $C \cup \Isec_L^v$ witnesses
that $C$ is freely drawable.
\end{proof}
Finally, we remark that we may recognize freely drawable components in polynomial time.
\begin{lemma}\label{lem:fi:free-recognition}
Given $C \in \Cfam_L \cap \Cfam_v$, we can recognize if $C$ is freely drawable in polynomial time.
\end{lemma}
\begin{proof}
We simply use Lemma~\ref{lem:ic-cliques-drawing} for the graph $(G+F_L)[C \cup \Isec_L^v]$
and cliques $\Isec_L^v \cap \Isec_L^f$ and $\Isec_L^v$.
\end{proof}

Using Lemma~\ref{lem:fi:free-recognition}, we recognize all components of $(\Cfam_L \cap \Cfam_v) \setminus \Cfam_R$ that
are not freely drawable. If there are more than $2k$ of them, by Lemma~\ref{lem:fi:not-free} we may discard the current branch.
Otherwise, we put all not freely drawable components of $(\Cfam_L \cap \Cfam_v) \setminus \Cfam_R$ into $\Tfam$.

We remark that if $C$ is freely drawable, then $\Isec_L^v \cap \Isec_L^f \subseteq N_G(w)$ for any $w \in C$.

As we needed to exclude the components of $\Cfam_R$ for Lemma~\ref{lem:fi:not-free}, we now proceed to the components of $\Cfam_L \cap \Cfam_R$.
Denote $P = \Isec_L^f \cap \Isec_R^f$ and $K = (\Isec_L^v \cap \Isec_R^v) \setminus P$.
It turns out that the choice of $f_L$ and $f_R$ implies that $K$ is small.
\begin{lemma}\label{lem:fi:K-small}
All vertices of $K$ are touched by $F$ and, consequently, $|K| \leq 2k$.
\end{lemma}
\begin{proof}
Consider any $x \in K$.
As $x \in \Isec_L^v \cap \Isec_R^v$, we have $\model(\Ibeg{x}) \leq p_L^v < p_R^v \leq \model(\Iend{x})$.
As $x \notin \Isec_L^f \cap \Isec_R^f$, we have $\model(\Ibeg{x}) > p_L^f$ or
$\model(\Iend{x}) < p_R^f$. If $x$ is untouched by $F$, $x$ would be a better candidate
for $f_L$ in the first case, and a better candidate for $f_R$ in the second case.
\end{proof}
Note that by Lemma~\ref{lem:fi:inclusions} we have $N_G(C) \subseteq P \cup K$ for any $C \in \Cfam_L \cap \Cfam_R$.
Lemma~\ref{lem:fi:K-small} allows us to use the bound of Lemma~\ref{lem:A-r-nei}.
\begin{lemma}\label{lem:fi:LR-small}
$|\Cfam_L \cap \Cfam_R| = \Oh(k^2)$.
\end{lemma}
\begin{proof}
There are at most $2k$ connected components of $\Cfam_L \cap \Cfam_R$ that are touched by $F$.
Consider now untouched $C \in \Cfam_L \cap \Cfam_R$.
As $p_L^f < \Sbeg{C}{\model} < \Send{C}{\model} < p_R^f$, we have
$aw \in E(G)$ for any $w \in C$, $a \in P$.
The lemma follows from an application of Lemma~\ref{lem:A-r-nei}
to $A = P\cup K$ and $r = |K| \leq 2k$.
\end{proof}
Thus, if $|\Cfam_L \cap \Cfam_R|$ is too large, we discard the current branch.
Moreover, we can also discard the current branch if there exists $C \in \Cfam_L \cap \Cfam_R$
with $|(C \times P) \setminus E(G)| > k$: such a component $C$ would need too much fill-in edges
between itself and $P$.
If neither of the above situations happen, we insert $\Cfam_L \cap \Cfam_R$ into $\Tfam$, that is,
we treat all components of $\Cfam_L \cap \Cfam_R$ as troublesome.

We now inspect the possible order of the starting endpoints
of the vertices of $\Isec_L^v \setminus \Isec_L^f$; all these endpoints
appear between positions $p_L^f$ and $p_L^v$. We denote
$$X = \bigcup_{C \in \Cfam_L \setminus \Cfam_v} N_G(C) \cap \Isec_L^v$$
and observe the following.
\begin{lemma}\label{lem:fi:X-first}
For any $C \in (\Cfam_L \cap \Cfam_v) \setminus \Cfam_R$, if there exists 
$w \in C$ with $X \not\subseteq N_G(w)$, then $C$ is touched by $F$.
\end{lemma}
\begin{proof}
Consider such component $C$ and vertex $w\in C$.
As $C \notin \Cfam_R$, either case~\ref{case:fi:top} or case~\ref{case:fi:left}
of Lemma~\ref{lem:fi:structure} applies to $C$.
If case~\ref{case:fi:top} applies, then $wv \in F$ and we are done, so assume otherwise.

Let $D \in \Cfam_L \setminus \Cfam_v$ such that
there exists $x \in (N_G(D) \cap \Isec_L^v) \setminus N_G(w)$. Note that in particular $C\neq D$ and hence $w$ does not have any neighbor in $D$ in the graph $G$.
As $D \in \Cfam_L \setminus \Cfam_v$, there exists
some $y \in (\Isec_L^f \setminus \Isec_L^v) \cap N_G(D)$.
Since $C \in \Cfam_v$, then we have $y \notin N_G(C)$, so in particular $wy \notin E(G)$.

Let $P$ be a path in $G$ with endpoints in $x$ and $y$ and all internal vertices in $D$; such a path exists since $D$ is connected.
Note that $P$ contains no neighbor of $w$ in $G$, but connects
$y \in \Isec_L^f = \Isec_\model(p_L^f)$ with $x \in \Isec_L^v = \Isec_\model(p_L^v)$.
As $p_L^f < \model(\Ibeg{w}) < \model(\Iend{w}) < p_L^v$, $w$ neighbors some vertex
of $P$ in $G+F$, and hence $w$ is touched by $F$.
\end{proof}

By Lemma~\ref{lem:fi:X-first} we expect at most $2k$ components of
$(\Cfam_L \cap \Cfam_v) \setminus \Cfam_R$ for which
$X \not\subseteq N_G(w)$ for some $w \in C$.
If there are more such components, we discard the current branch.
Otherwise, we include all such components into $\Tfam$.

We refer to Figure~\ref{fig:fill-in2} for an illustration of some of the introduced notation.

\begin{figure}
\centering
\includegraphics{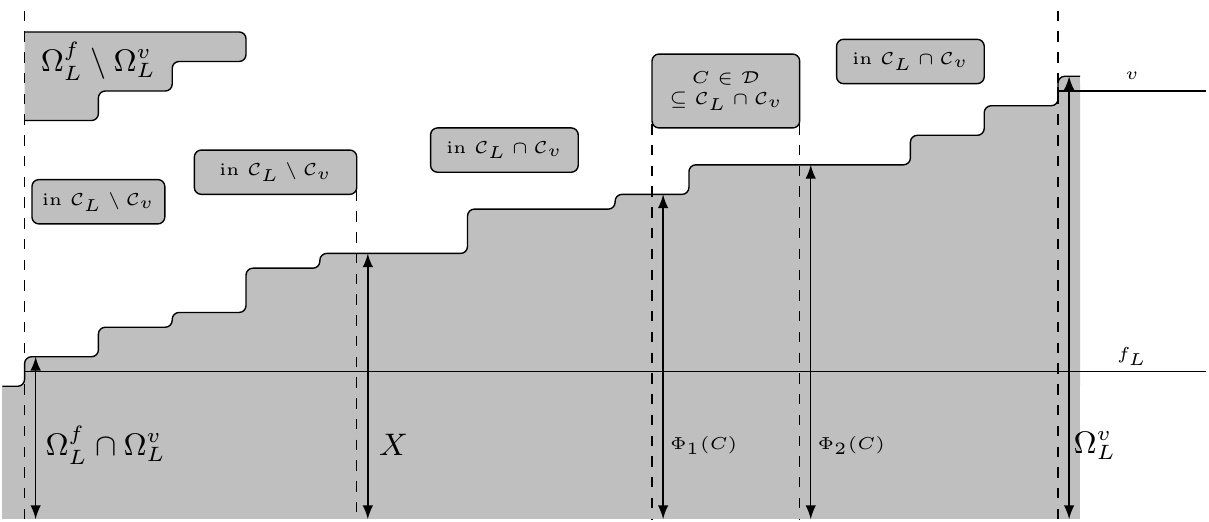}
\caption{A closer insight into the area between $\Isec_L^f$ and $\Isec_L^v$.}
\label{fig:fill-in2}
\end{figure}

We now define the following relation $\unlhd$ on the components of $(\Cfam_L \cap \Cfam_v) \setminus \Cfam_R$:
for two components $C_1,C_2 \in (\Cfam_L \cap \Cfam_v) \setminus \Cfam_R$ we have $C_1 \unlhd C_2$
iff for any $v_1 \in C_1$ and for any $v_2 \in C_2$ it holds that $N_G(v_1) \cap \Isec_L^v \subseteq N_G(v_2) \cap \Isec_L^v$.
Clearly, $\unlhd$ is a transitive and reflexive relation on $(\Cfam_L \cap \Cfam_v) \setminus \Cfam_R$.
Intuitively, $\unlhd$ should be close to a total quasi-order, and should resemble the order
in which the components of $(\Cfam_L \cap \Cfam_v) \setminus \Cfam_R$ that fall into case~\ref{case:fi:left} of Lemma~\ref{lem:fi:structure}
appear in the model $\model$, and components that are equivalent with respect to $\unlhd$ should be interchangeable modules.
This intuition is partially formalized in the following lemma.
\begin{lemma}\label{lem:fi:incomparable}
If two components $C_1,C_2 \in (\Cfam_L \cap \Cfam_v) \setminus \Cfam_R$ are incomparable with respect to $\unlhd$, then
at least one of them is touched by $F$.
\end{lemma}
\begin{proof}
If a component of $\Cfam_L$ falls into case~\ref{case:fi:top} of Lemma~\ref{lem:fi:structure}, then all
its vertices are touched. Hence, assume that both $C_1$ and $C_2$ fall into case~\ref{case:fi:left}.

If $v_1v_2 \in F$ for some $v_1 \in C_1$, $v_2 \in C_2$, then both components are touched by $F$.
Otherwise, $\Send{C_1}{\model} < \Sbeg{C_2}{\model}$
or $\Send{C_2}{\model} < \Sbeg{C_1}{\model}$; w.l.o.g. assume the first option.
However, then for any $v_1 \in C_1$ and $v_2 \in C_2$ it holds
that $N_{G+F}(v_1) \cap \Isec_L^v \subseteq N_{G+F}(v_2) \cap \Isec_L^v$. 
Hence $C_1 \unlhd C_2$ unless $C_2$ is touched.
\end{proof}

Consider now an auxiliary graph $G_\Cfam$ with vertex set $(\Cfam_L \cap \Cfam_v) \setminus \Cfam_R$ and two components
$C_1$ and $C_2$ being adjacent iff they are incomparable w.r.t. $\unlhd$. By Lemma~\ref{lem:fi:incomparable},
the family of touched components is a vertex cover of $G_\Cfam$ of size at most $2k$.
We run a $2$-approximation algorithm to find a vertex cover $\mathcal{V}$ of $G_\Cfam$.
If $|\mathcal{V}| > 4k$, we discard the current branch. Otherwise, we insert
$\mathcal{V}$ into $\Tfam$.

This concludes the construction of the family $\Tfam$ of troublesome components. Note that $|\Tfam| = \Oh(k^2)$
and $|\Tfam \setminus (\Cfam_L \cap \Cfam_R)| = \Oh(k)$.
Let $\Dfam = (\Cfam_L \cap \Cfam_v) \setminus \Tfam$ be the set of not troublesome components.
We summarize the properties of the components of $\Dfam$.
\begin{enumerate}
\item Every $C \in \Dfam$ is freely drawable.
\item $N_G(C) \subseteq \Isec_L^v$ for any $C \in \Dfam$.
\item Each component $C \in \Dfam$ does not belong to $\Cfam_R$.
That is, $N_G(C)$ contains a vertex of $\Isec_L^v \setminus \Isec_R^v$.
\item The relation $\unlhd$, restricted to $\Dfam$, is a total quasi-order.
\item For every component $C \in \Dfam$ and each $w \in C$, we have $X \subseteq N_G(w)$.
\end{enumerate}

\subsection{Being close and far from a troublesome component}

In this section we show that any component that is \emph{far} from all components of $\Tfam$, in a specific meaning defined later,
is left untouched by $F$.
This, together with a bound on the number of components \emph{close} to $\Tfam$ will conclude the proof of Theorem~\ref{thm:fill-in}.

For any component $C \in \Cfam_L$ we define the following two measures.
\begin{align*}
\minnei{C} &= \min_{w \in C} |N_G(w) \cap \Isec_L^v|, \\
\maxnei{C} &= \max_{w \in C} |N_G(w) \cap \Isec_L^v|.
\end{align*}
Note that $\maxnei{C_1} \leq \minnei{C_2}$ whenever $C_1 \unlhd C_2$. Observe moreover that $\minnei{C}\geq |X|$ for each $C\in \Dfam$. 

Consider now some $C \in \Dfam$.
We first observe that $N_G(w) \cap \Isec_L^v = N_G(w) \setminus C$ for any $w \in C$.
Second, note that, as $C$ is freely drawable,
for any $w_1,w_2 \in C$ we have $N_G(w_1) \cap \Isec_L^v \subseteq N_G(w_2) \cap \Isec_L^v$ or vice-versa.
In particular, for $C \in \Dfam$ if we define sets 
\begin{align*}
\minneiset{C} &= \bigcap_{w \in C} N_G(w) \cap \Isec_L^v, \\
\maxneiset{C} &= \bigcup_{w \in C} N_G(w) \cap \Isec_L^v,
\end{align*}
then there exists $w_1,w_2 \in C$ with $N_G(w_1) \cap \Isec_L^v = \minneiset{C}$ and
$N_G(w_2) \cap \Isec_L^v = \maxneiset{C}$. In particular,
$|\minneiset{C}| = \minnei{C}$ and $|\maxneiset{C}| = \maxnei{C}$.

Enumerate now $\Dfam = \{C^1,C^2, \ldots,C^{|\Dfam|}\}$ such that
$$C^1 \unlhd C^2 \unlhd \ldots \unlhd C^{|\Dfam|}.$$
Note that the aforementioned numeration is not unique, as $\unlhd$ is a quasi-order: they may exist
$C_1,C_2 \in \Dfam$ with $C_1 \unlhd C_2$ and $C_2 \unlhd C_1$.
However, we note that such a situation is somehow limited by inapplicability of the Module Reduction Rule.
\begin{lemma}\label{lem:fi:equal}
If $C_1 \unlhd C_2$ and $C_2 \unlhd C_1$ for some $C_1,C_2 \in \Dfam$, then
$C_1$, $C_2$ and $C_1 \cup C_2$ are modules in $G$.
Moreover, if $\Dfam' \subseteq \Dfam$ such that $C_1 \unlhd C_2$ and $C_2 \unlhd C_1$
for any $C_1,C_2 \in \Dfam'$, then $|\Dfam'| \leq 2k+2$.
\end{lemma}
\begin{proof}
By the definition of the relation $\unlhd$, we infer that 
$$N_G(v_1) \setminus C_1 = N_G(v_1) \cap \Isec_L^v = N_G(v_2) \cap \Isec_L^v = N_G(v_2) \setminus C_2$$
for any $v_1 \in C_1$, $v_2 \in C_2$. The first claim follows.
For the second claim, note that if $|\Dfam'| \geq 2k+3$, then the Module Reduction Rule would be applicable
to any $2k+3$ components of $\Dfam'$, and the set $\Isec_L^v$.
\end{proof}
\begin{corollary}\label{cor:fi:index-jump}
For any $1 \leq a \leq b \leq |\Dfam|$ we have
$$\minnei{C^b} - \maxnei{C^a} \geq \left\lceil \frac{b-a}{2k+3} \right\rceil - 1.$$
\end{corollary}
\begin{proof}
Let $a < c_1 < c_2 < \ldots < c_s < b$ be the sequence of all indices $a < c < b$
such that $\maxnei{C^{c-1}} < \maxnei{C^c}$.
By Lemma~\ref{lem:fi:equal}, $c_{i+1} - c_i \leq 2k+3$ for any $1 \leq i < s$ and $c_1 - a \leq 2k+3$, $b - c_s \leq 2k+3$.
Consequently, $(2k+3)(s+1) \geq b-a$.
The lemma follows from the observation that $s \leq \maxnei{C^{c_s}} - \maxnei{C^a} \leq \minnei{C^b} - \maxnei{C^a}$.
\end{proof}

Given the ordering $C^1, C^2, \ldots, C^{|\Dfam|}$ we can also observe the following corollary of the fact that all components of
$\Dfam$ are freely drawable.
\begin{lemma}\label{lem:fi:draw}
For any $1 \leq a \leq b \leq |\Dfam|$, if we define $F' = \binom{\maxneiset{C^b}}{2} \setminus E(G)$
then the graph
$$(G+F')\left[ \maxneiset{C^b} \cup \bigcup_{c=a}^b C^c \right]$$
is interval and admits a model that starts with the starting events of $\events{\minneiset{C^a}}$
and ends with the ending events of $\events{\maxneiset{C^b}}$.
\end{lemma}
\begin{proof}
We prove the lemma by induction on $b-a$. For the base case $a=b$, observe that the claim is equivalent to the definition of $C^a$ being freely drawable.
In the induction step, pick any $a < c \leq b$ and use the induction hypothesis for
components $C^a, C^{a+1}, \ldots, C^{c-1}$ and $C^c, C^{c+1}, \ldots, C^b$, obtaining models $\model_1$ and $\model_2$.
Create the desired model $\model_0$ by concatenating:
\begin{enumerate}
\item the model $\model_1$, with removed suffix consisting of the ending events of $\events{\maxneiset{C^{c-1}}}$,
\item the starting events of $\events{\minneiset{C^c} \setminus \maxneiset{C^{c-1}}}$, and
\item the model $\model_2$, with removed prefix consisting of the starting events of $\events{\minneiset{C^c}}$.
\end{enumerate}
It is straightforward to verify that $\model_0$ satisfies all the promised properties.
\end{proof}

We now turn our attention to the troublesome components and inspect how they interact with the family $\Dfam$.
For each $T \in \Tfam$ define the following.
\begin{align*}
\ToutL{T} &= \min \{x : \maxnei{C^x} \geq \minnei{T}\} \\
\ToutR{T} &= \max \{x : \minnei{C^x} \leq \maxnei{T}\} \\
\TinL{T} &= \min \{x: \minnei{C^x} > \minnei{T} + k\} \\
\TinR{T} &= \max \{x: \maxnei{C^x} < \maxnei{T}\}
\end{align*}
All these values can attain $+\infty$ or $-\infty$ if the corresponding set for minimization or maximization is empty.

Clearly, $\ToutL{T} \leq \TinL{T}$, $\ToutR{T} \geq \TinR{T}$ and $\ToutL{T} \leq \ToutR{T} + 1$.
We note that, by Corollary~\ref{cor:fi:index-jump}, we have $\TinL{T} - \ToutL{T} = \Oh(k^2)$ and
$\ToutR{T} - \TinR{T} = \Oh(k)$. 
We claim the following.
\begin{lemma}\label{lem:fi:wide-T}
If $\TinR{T} - \TinL{T} > 2k$, then $T$ does not fall into case~\ref{case:fi:left} of Lemma~\ref{lem:fi:structure}.
\end{lemma}
\begin{proof}
Let $x,y \in T$ such that
$|N_G(x) \cap \Isec_L^v| = \minnei{T}$ and
$|N_G(y) \cap \Isec_L^v| = \maxnei{T}$.
If $\TinR{T} - \TinL{T} > 2k$ then there exists a component $C^c$ that is untouched by $F$
for some $\TinL{T} \leq c \leq \TinR{T}$. 
Hence, for any $w \in C^c$ we have
\begin{eqnarray*}
|N_{G+F}(x) \cap \Isec_L^v| \leq |N_G(x) \cap \Isec_L^v| + k < |N_G(w) \cap \Isec_L^v| = |N_{G+F}(w) \cap \Isec_L^v| < |N_G(y) \cap \Isec_L^v| \leq |N_{G+F}(y) \cap \Isec_L^v|.
\end{eqnarray*}
Summarizing, $|N_{G+F}(x) \cap \Isec_L^v| < |N_{G+F}(w) \cap \Isec_L^v| < |N_{G+F}(y) \cap \Isec_L^v|$. As $T$ is connected in $G$ and no edge of $G+F$ connects $C^c$ with $T$, it cannot happen that both
$C^c$ and $T$ fall into case~\ref{case:fi:left} of Lemma~\ref{lem:fi:structure}.
However, since $C^c$ is untouched and does not belong to $\Cfam_R$, $C^c$ falls into case~\ref{case:fi:left} of Lemma~\ref{lem:fi:structure}.
This finishes the proof of the lemma.
\end{proof}
Let $\Tfam' = \{T \in \Tfam: \TinR{T} - \TinL{T} \leq 2k\}$ be the set of these troublesome components
for which Lemma~\ref{lem:fi:wide-T} is not applicable. Note also that for any $T \in \Tfam'$ we
have $-1 \leq \ToutR{T} - \ToutL{T} = \Oh(k^2)$.

We say that a component $C^c \in \Dfam$ is \emph{far} from a troublesome component $T$
if either $\ToutR{T} < c - \eta$ or $\ToutL{T} > c + \zeta$, 
   where
\begin{align*}
\gamma&= (2k+3)(k+2)+1, & \delta &= 2(2k+3)+1, \\
\eta&= \gamma\cdot (2k+2), & \zeta &= \delta\cdot (2k+3).
\end{align*}
A component $C$ is \emph{close} to $T$ if it is not far from $T$.
Define $\Dfam_0$ to be the set of these components $C^c \in \Dfam$ such that
$C^c$ is far from all components of $\Tfam'$ and, moreover, $\eta < c < |\Dfam| - \zeta$.

With this definition, we are now ready for the crucial argumentation of this section.
\begin{lemma}\label{lem:fi:far-untouched}
Any component $C \in \Dfam_0$ is untouched by $F$.
Consequently, such $C$ falls into case~\ref{case:fi:left} of Lemma~\ref{lem:fi:structure}.
\end{lemma}

\begin{figure}[tb]
\centering
\includegraphics[width=\linewidth]{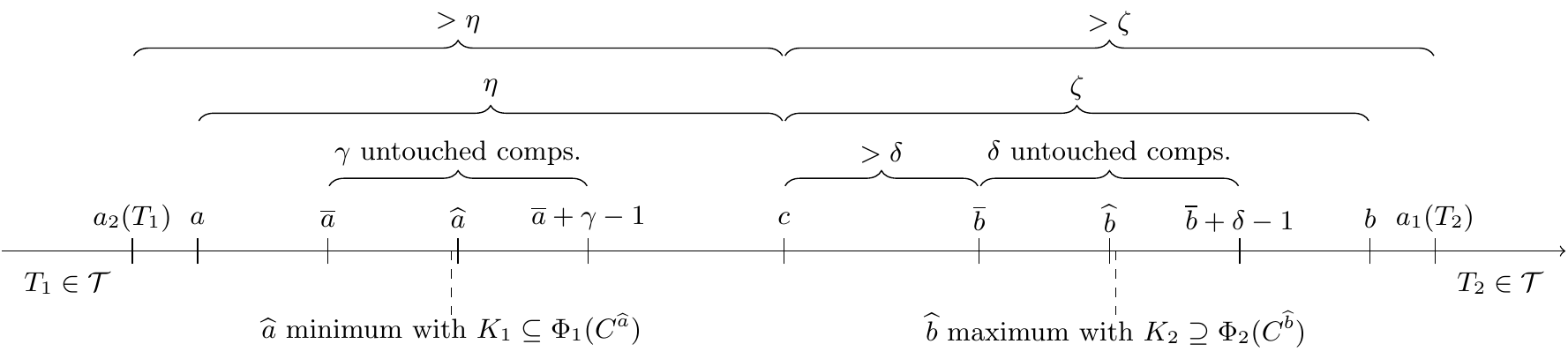}
\caption{The indices defined in the proof of Lemma~\ref{lem:fi:far-untouched}.}
\label{fig:axis}
\end{figure}

\begin{proof}
Let $C^c \in \Dfam$ be far from all components of $\Tfam'$.
Denote $a = c - \eta$ and $b = c + \zeta$.
By the assumptions of the lemma, $1 \leq a < b \leq |\Dfam|$ and, for each $T \in \Tfam'$
we have either $\maxnei{T} < \minnei{C^a}$ or $\minnei{T} > \maxnei{C^b}$.
We refer to Figure~\ref{fig:axis} for indices defined in the course of this proof.

By the Pigeonhole Principle, there exists some $\overline{a}$, $a \leq \overline{a} \leq c-\gamma$, such that
all components $C^{\overline{a}}, C^{\overline{a}+1}, \ldots, C^{\overline{a} + \gamma - 1}$ are untouched by $F$.
Symmetrically, there exists some $\overline{b}$, 
$c + \delta < \overline{b} \leq b-\delta+1$, such that all components
$C^{\overline{b}}, C^{\overline{b}+1},\ldots,C^{\overline{b}+\delta-1}$ are untouched by $F$.
By Corollary~\ref{cor:fi:index-jump}, we have
\begin{align}
k &< \minnei{C^{\overline{a}+\gamma-1}} - \maxnei{C^{\overline{a}}}, \label{eq:fi:final-a}\\
0 &< \minnei{C^{\overline{b}+\delta-1}} - \maxnei{C^{\overline{b}}}, \label{eq:fi:final-b}\\
0 &< \minnei{C^{\overline{b}}} - \maxnei{C^c} \leq \minnei{C^{\overline{b}}} - \maxnei{C^{\overline{a}+\gamma-1}}.\label{eq:fi:final-c}
\end{align}

Recall that an untouched component of $\Dfam$ needs to fall into case~\ref{case:fi:left} of Lemma~\ref{lem:fi:structure}.
Moreover, such components need to lie one after another in the model $\model$, that is,
if $C_1,C_2 \in \Dfam$ are untouched, then $\Send{C_1}{\model} < \Sbeg{C_2}{\model}$
or $\Send{C_2}{\model} < \Sbeg{C_1}{\model}$.
Note that the first case is possible only if $C_1 \unlhd C_2$, and the second one only if $C_2 \unlhd C_1$.

Let $p_1 = \Send{C^{\overline{a}+\gamma-1}}{\model}$ and $p_2 = \Sbeg{C^{\overline{b}}}{\model}$. 
From~\eqref{eq:fi:final-c} we infer that $p_1 < p_2$.
Denote $K_1 = \maxneiset{C^{\overline{a}+\gamma-1}}$ and $K_2 = \minneiset{C^{\overline{b}}}$
and observe that $\Isec_\model(p_1) = K_1 \subseteq \Isec_\model(p_2-1) \subseteq K_2$.

For any $C \in \Cfam_L$, we have either $\Send{C}{\model} \leq p_1$, $\Sbeg{C}{\model} \geq p_2$
or $p_1 < \Sbeg{C}{\model} < \Send{C}{\model} < p_2$. We claim the following.
\begin{claim}\label{cl:fi:final}
Let $C\in \Cfam_L$. If $p_1 < \Sbeg{C}{\model} < \Send{C}{\model} < p_2$, then $C \in \Dfam$ and $C = C^d$ for some $d$ with
$\maxnei{C^{\overline{a}+\gamma-1}} \leq \minnei{C^d} \leq \maxnei{C^d} \leq \minnei{C^{\overline{b}}}$ (in particular $\overline{a} < d < \overline{b} + \delta-1$, by Corollary~\ref{cor:fi:index-jump}).
\end{claim}
\begin{proof}
Observe that if $C$ satisfies $p_1 < \Sbeg{C}{\model} < \Send{C}{\model} < p_2$, then for every $w\in C$ it must hold that $K_1\subseteq N_{G+F}(w)\cap \Isec_L^v\subseteq K_2$. Since $|F|\leq k$, we infer that $|N_G(w)\cap K_1|\geq |K_1|-k$ and $N_G(w)\subseteq K_2$, for each $w\in C$. We now consider a few cases depending on the category $C$ belongs to.

If $C \notin \Cfam_v$ then $\maxnei{C} \leq |X| \leq \minnei{C^1}$ as $N_G(C) \cap \Isec_L^v \subseteq X$ by the definition of $X$.
Hence, by~\eqref{eq:fi:final-a}, $\maxnei{C} + k < |K_1|$, and the edges of $F$ cannot make $C$ adjacent to the entire $K_1$.

If $C \in \Tfam \setminus \Tfam'$, then Lemma~\ref{lem:fi:wide-T} implies that $C$ cannot lie between positions $p_1$ and $p_2$.
If $C \in \Tfam'$ then, by the choice of $C^c$, $\overline{a}$ and $\overline{b}$, we have either $\maxnei{C} < \minnei{C^{\overline{a}}}$ or
$\minnei{C} > \maxnei{C^{\overline{b}+\delta-1}}$.
In the first case, by~\eqref{eq:fi:final-a} we infer that $\maxnei{C} + k < |K_1|$.
In the second case, by~\eqref{eq:fi:final-b} we infer that $\minnei{C} > |K_2|$.
In both cases, the argumentation of the first paragraph shows that $C$ cannot lie between positions $p_1$ and $p_2$.

We are left with the case where $C \in \Dfam$ and $C = C^d$ for some $1 \leq d \leq |\Dfam|$.
By contradiction, assume first that $\minnei{C^d} < \maxnei{C^{\overline{a}+\gamma-1}}$.
If $d \geq \overline{a}$, then $C^d$ is untouched and the vertex $w \in C^d$ that has only $\minnei{C^d} < |K_1|$
neighbors in $\Isec_L^v$ cannot be placed after position $p_1$. Otherwise, 
by~\eqref{eq:fi:final-a} we have $\maxnei{C^d} + k < |K_1|$, and the edges of $F$ are not sufficient
to make $C^d$ fully adjacent to $K_1$.
In the second case, when $\maxnei{C^d} > \minnei{C^{\overline{b}}} = |K_2|$, clearly $C^d$ cannot be placed 
before position $p_2$ as there exists a vertex of $C^d$ that has more than $|K_2|$ neighbors in $\Isec_L^v$.
This finishes the proof of the claim.
\cqed\end{proof}

Define now indices $\widehat{a}$ and $\widehat{b}$ as follows: $\widehat{a}$ is minimum such that
$\minnei{C^{\widehat{a}}} \geq |K_1|$ (equivalently, $K_1 \subseteq \minneiset{C^{\widehat{a}}}$)
and $\widehat{b}$ is maximum such that 
$\maxnei{C^{\widehat{b}}} \leq |K_2|$ (equivalently, $K_2 \supseteq \maxneiset{C^{\widehat{b}}}$).
By the definition of $K_1$ and $K_2$, we have $\overline{a} < \widehat{a} \leq \overline{a}+\gamma$
and $\overline{b}-1 \leq \widehat{b} \leq \overline{b} + \delta-1$.
Denote $F_K = \binom{K_2}{2} \setminus E(G)$; note that $F_K \subseteq F$.
By Lemma~\ref{lem:fi:draw}, it is easy to see that there exists an interval model $\model_0$ of
$$(G+F_K)\left[K_2 \cup \bigcup_{d=\widehat{a}}^{\widehat{b}} C^d\right]$$
that starts with the starting events of $\events{K_1}$ and ends with the ending events of $\events{K_2}$.

Let us create a model $\model'$ from $\model$ by
\begin{enumerate}
\item removing all events of $\bigcup_{d = \widehat{a}}^{\widehat{b}} \events{C^d}$
as well as all starting events of $\events{K_2 \setminus K_1}$;
observe that, by Claim~\ref{cl:fi:final}, we have in particular removed all events that lie in $\model$ between positions $p_1$ and $p_2$, exclusive;
\item inserting all events of $\model_0$, except for the prefix consisting of the starting events of $\events{K_1}$
and the ending events of $\events{K_2}$, in the place between former positions $p_1$ and $p_2$ in $\model$, in the original order.
\end{enumerate}
Since $K_1 = \Isec_\model(p_1)$ and $K_2 = \minneiset{C^{\overline{b}}}$ we infer that $\model'$ is an interval
model of $G+F'$ for some completion $F'$. As $F_K \subseteq F$, we have $F' \subseteq F$.
Moreover, as $\widehat{a} \leq c \leq \widehat{b}$, $C^c$ is untouched by $F'$.
By the inclusion-wise minimality of $F$, $F' = F$ and the lemma is proven.
\end{proof}

We now show that almost all elements of $\Dfam$ in fact belong to $\Dfam_0$.
\begin{lemma}\label{lem:fi:few-close}
$|\Dfam \setminus \Dfam_0| = \Oh(k^4)$.
\end{lemma}
\begin{proof}
Clearly, a component $T \in \Tfam'$ is close to $\Oh(k^3)$ components of $\Dfam$.
Moreover, note that for any $T \in \Cfam_L \cap \Cfam_R$ we have
that $N_G(T) \subseteq P \cup K$, but, as $|(T \times P) \setminus E(G)| \leq k$
and $|K| \leq 2k$ (Lemma~\ref{lem:fi:K-small}),
it implies $|P|-k \leq \minnei{T} \leq \maxnei{T} \leq |P|+2k$.
Consequently, by Corollary~\ref{cor:fi:index-jump} there are $\Oh(k^2)$ components of $\Dfam$
that are close to some $T \in \Cfam_L \cap \Cfam_R$.
As $|\Tfam \setminus (\Cfam_L \cap \Cfam_R)| = \Oh(k)$, the lemma follows.
\end{proof}

\newcommand{\liten}{{\textrm{small}}}

Let $\Cfam_{\liten}$ be the family of those components $C\in \Cfam$ for which $|C|\leq k$. Note that a component $C\in \Cfam$ can fall into case \ref{case:fi:top} only if $C\in \Cfam_v\cap \Cfam_{\liten}$, since each vertex of a component falling into case \ref{case:fi:top} must have a fill-in edge to $v$, and the number of such edges is at most $k$.

Finally, denote
$$B_L = \bigcup \left((\Dfam \setminus \Dfam_0)\cap \Cfam_{\liten}\right) \cup \bigcup \left( \Tfam \cap \Cfam_v\cap \Cfam_{\liten}\right).$$
By Lemma~\ref{lem:fi:few-close} and the definition of $\Cfam_{\liten}$ we have that $\left|\bigcup \left((\Dfam \setminus \Dfam_0)\cap \Cfam_{\liten}\right)\right|=\Oh(k^5)$. Since $|\Tfam|=\Oh(k^2)$, we have $\left|\bigcup \left( \Tfam \cap \Cfam_v\cap \Cfam_{\liten}\right)\right|=\Oh(k^3)$. As a result, we obtain $|B_L| = \Oh(k^5)$. 
Symmetrically, by inspecting $\Cfam_R$ instead of $\Cfam_L$, we obtain a set $B_R$ of size $\Oh(k^5)$.

Define now $B = B^\mathrm{sure} \cup B_L \cup B_R$.
As $\Cfam_v \setminus \Tfam \subseteq \Dfam$, Lemma~\ref{lem:fi:far-untouched} ensures
that $\{w \in V(G): vw \in F\} \subseteq B$. Hence, we insert $B$ into the constructed family $\fiset$
and conclude the proof of Theorem~\ref{thm:fill-in}.

\section{Small-separation lemma}\label{sec:left-right}
In this short section we prove the following structural result.
\begin{theorem}\label{thm:left-right}
Let $(G,k)$ be a YES-instance to \icname{}, let $F$ be a minimum solution to $(G,k)$
and let $\model$ be the canonical model of $G+F$.
Let $p_L < p_R$ be two integers and denote $\Isec_L = \Isec_\model(p_L)$, $\Isec_R = \Isec_\model(p_R-1)$.
Assume $K \subseteq V(G)$ is such that $K \subseteq \Isec_L \setminus \Isec_R$ or $K \subseteq \Isec_R \setminus \Isec_L$.
Then there are at most $3\sqrt{k} + |K|$ connected components $C$ of $G \setminus (\Isec_L \cup \Isec_R)$ satisfying:
\begin{enumerate}
\item $N_G(C) \subseteq K \cup (\Isec_L \cap \Isec_R)$,
\item $p_L < \Sbeg{C}{\model} < \Send{C}{\model} < p_R$, and
\item there exists $\event \in \events{K}$ such that $\Sbeg{C}{\model} < \model(\event) < \Send{C}{\model}$.
\end{enumerate}
\end{theorem}

\subsection{A few words on motivation}

Before we proceed to the proof of Theorem~\ref{thm:left-right}, let us now shortly elaborate on
the motivation of this result.

Assume we have two vertices $x$ and $y$, and we know (have guessed) that they are cheap with respect to the minimum solution $F$ we are looking for. Moreover,
in the canonical model $\model$ of $G+F$ we have $\model(\Ibeg{x}) < \model(\Ibeg{y}) < \model(\Iend{y}) < \model(\Iend{x})$.
By Corollary~\ref{cor:cheap-fill-in}, there are only $k^{\Oh(\sqrt{k})} n^{70}$ choices
for each of the set $\incF{x}$, $\incF{y}$, so assume we know them as well.
Similarly, there is only a subexponential number of choices for the sections at the endpoints
of $x$ and $y$. Hence, assume we have guessed them and denote them by $\Isec_L^x$, $\Isec_L^y$,
   $\Isec_R^y$ and $\Isec_R^x$. Note that we may assume that standard inclusions between these sections: $\Isec_L^x\cap \Isec_R^y\subseteq \Isec_L^y$, $\Isec_L^y\cap \Isec_R^x\subseteq \Isec_R^y$, and $\Isec_L^x\cap \Isec_R^x\subseteq \Isec_L^y\cap \Isec_R^y$.

\begin{figure}[tb]
\centering
\includegraphics{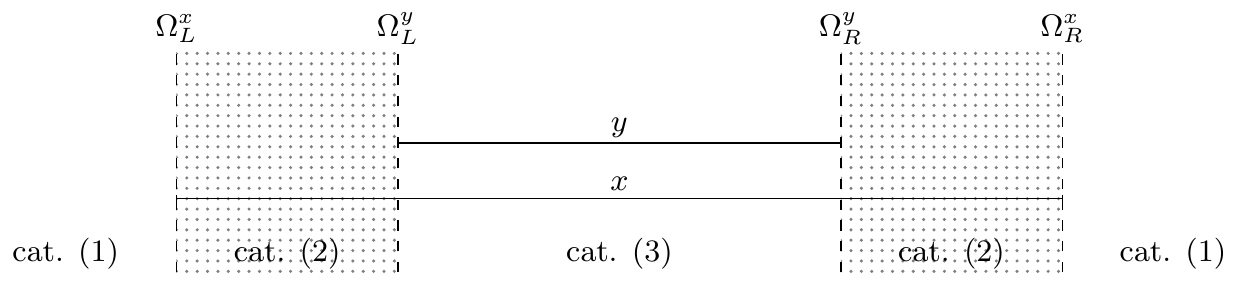}
\caption{Motivation for Theorem~\ref{thm:left-right}: we would like to reason about the alignment
 of the vertices of category (2) in the dotted areas.}
\label{fig:left-right-motiv}
\end{figure}

Consider any vertex $v \in V(G) \setminus (\Isec_L^x \cup \Isec_L^y \cup \Isec_R^y \cup \Isec_R^x)$.
Note that, by inspecting whether $vx \in E(G) \cup \incF{x}$ and whether $vy \in E(G) \cup \incF{y}$,
we may classify $v$ into one of three categories (see also Figure~\ref{fig:left-right-motiv}):
\begin{enumerate}
\item $vx \notin E(G) \cup \incF{x}$ and $vy \notin E(G) \cup \incF{y}$, hence
$\model(\Iend{v}) < \model(\Ibeg{x})$ or $\model(\Ibeg{v}) > \model(\Iend{x})$;
\item $vx \in E(G) \cup \incF{x}$ but $vy \notin E(G) \cup \incF{y}$, hence
$\model(\Ibeg{x}) < \model(\Ibeg{v}) < \model(\Iend{v}) < \model(\Ibeg{y})$ or
$\model(\Iend{y}) < \model(\Ibeg{v}) < \model(\Iend{v}) < \model(\Iend{x})$;
\item $vx \in E(G) \cup \incF{x}$ and $vy \in E(G) \cup \incF{y}$, hence
$\model(\Ibeg{y}) < \model(\Ibeg{v}) < \model(\Iend{v}) < \model(\Iend{y})$.
\end{enumerate}
Moreover, the choice of the category needs to be homogeneous among each connected component of
$G \setminus (\Isec_L^x \cup \Isec_L^y \cup \Isec_R^y \cup \Isec_R^x)$.

We will be interested mostly in the second category, and we would like to guess which components
$C$ of this category lie, in the model $\model$, to the left of the vertex $y$,
and which lie to the right of it.
Note that we may deduce this choice from the neighborhood of a component $C$
unless $N_G(C) \subseteq \Isec_L^y \cap \Isec_R^y$.

Theorem~\ref{thm:left-right} helps us if $K := (\Isec_L^y \cap \Isec_R^y) \setminus (\Isec_L^x \cap \Isec_R^x)$ is small, in particular, if it contains only expensive vertices
and thus its cardinality is bounded by $2\sqrt{k}$.
First, Lemma~\ref{lem:A-r-nei}, applied to $r = |K|$ and $A = \Isec_L^y \cap \Isec_R^y$
ensures that there are only $\mathrm{poly}(k)$ candidate components $C$.
Second, Theorem~\ref{thm:left-right} ensures that there are only $\Oh(\sqrt{k})$ such components $C$
that contain an event of $\events{K}$ between $\Sbeg{C}{\model}$ and $\Send{C}{\model}$;
we may guess them and guess on which side of $y$ they lie in the model $\model$.
Finally, we observe that the remaining components have been turned into modules in $G+F$ and,
as we shall show formally later, we may arrange them in a greedy manner.

\subsection{Proof}

By symmetry, let us assume that $K \subseteq \Isec_R \setminus \Isec_L$.
In particular, all starting events and no ending event of $\events{K}$ lie between
$p_L$ and $p_R$.
We say that a component $C$ \emph{occupies} the event $\event \in \events{K}$
if $\Sbeg{C}{\model} < \model(\event) < \Send{C}{\model}$.
Let $\Cfam$ be the family of component of $G \setminus (\Isec_L \cup \Isec_R)$ that satisfy
all conditions of Theorem~\ref{thm:left-right}, that is, we are to bound $|\Cfam|$

First, note that a much weaker bound $2k + |K|$ for Theorem~\ref{thm:left-right} is straightforward:
there are at most $2k$ components $C$ touched by $F$, and no two untouched components
may occupy the same event of $\events{K}$.
However, such a bound is useless from the point of view of the aforementioned motivation.

Second, we remark that it is quite easy to obtain a bound of order $\Oh(\sqrt{k|K|}+ |K|)$.
For each $C \in \Cfam$ pick one endpoint $\event_C\in \events{K}$ occupied by $C$.
For a starting event $\event$, denote
$n_\event = |\{C \in \Cfam: \event = \event_C\}|$. We are to bound
$|\Cfam| = \sum_\event n_\event$, where the number of non-zero values $n_\event$ is bounded
by $|K|$. Observe that $\sum_\event \binom{n_\event}{2} \leq |F| \leq k$,
as there exists at least one edge of $F$ between each pair of components that occupy the same
endpoint. The promised bound follows from the Cauchy-Schwarz inequality.

An $\Oh(\sqrt{k|K|} + |K|)$ bound is sufficient to establish a subexponential algorithm for \icname{},
but the final dependency on $k$ in the exponent would be $\Oh(k^{2/3} \log k)$.
Hence, we employ a more careful analysis of the components of $\Cfam$ to obtain the bound
promised in Theorem~\ref{thm:left-right}, and, consequently, reduce the dependency on $k$ to
exponential in $\Oh(\sqrt{k} \log k)$.

For any position $p_L \leq p < p_R$ and any component $C \in \Cfam$ we define
\begin{align*}
f(p) &= |\Isec_\model(p)|, & f_C(p) &= |\Isec_\model(p) \setminus C|.
\end{align*}
Recall that for each $C \in \Cfam$ we have $p_L < \Sbeg{C}{\model} < \Send{C}{\model} < p_R$ and $N_G(C) \subseteq K \cup (\Isec_L \cap \Isec_R) \subseteq \Isec_R$.
We refer to Figure~\ref{fig:left-right} for an overview of the notation used in this proof.

Informally speaking, the aforementioned inclusion allows us to compare the model $\model$ with its modification $\model'$, where some prefix of events of $\events{C}$ are shifted a bit to the right,
that is, $N_G(C) \subseteq \Isec_R$ ensures that $\model'$ still represents $G+F'$ for some completion $F'$. 
If $f_C$ for some $C \in \Cfam$ has a small value at some local minimum at $p \geq \Sbeg{C}{\model}$, we may shift all events of $\events{C}$ that lie before $p$ to this local minimum,
obtaining a smaller completion $F'$. We infer that $f$ is in some sense increasing, and we need to ``pay'' at least one in the value of $f$
for each component $C \in \Cfam$. Theorem~\ref{thm:left-right} will follow from an observation that the value of $f$ cannot change by much more than $|K|$.

We proceed to a formal argumentation. In the next three lemmas we establish the fact that $f$ is in some sense increasing.
\begin{lemma}\label{lem:lr:no-hole-inside}
For each $C \in \Cfam$ and each $\Sbeg{C}{\model} \leq p < \Send{C}{\model}$,
we have $f_C(p) \geq f(\Sbeg{C}{\model}-1)$.
\end{lemma}
\begin{proof}
Assume the contrary, and let $p$ be the smallest position such that $\Sbeg{C}{\model} \leq p < \Send{C}{\model}$
and $f_C(p) < f(\Sbeg{C}{\model}-1)$. Note that $f(\Sbeg{C}{\model}-1) = f_C(\Sbeg{C}{\model}-1)$.

Consider a model $\model'$ constructed from $\model$ as follows: all events of $\events{C}$ that lie before or on the position $p$ in the model $\model$
are moved (without changing their internal order) to the place just after position $p$.
As $N_G(C) \subseteq K \cup (\Isec_L \cap \Isec_R) \subseteq \Isec_R$, this is an interval model of $G+F'$ for some completion $F'$ of $G$.
We claim that $|F'| < |F|$.

Note that any $e \in F \triangle F'$ connects $C$ with $V(G) \setminus C$ ($\triangle$ denotes the symmetric difference).
Thus, it suffices to show that for each $v \in C$ we have $|\{w: vw \in F'\} \setminus C| \leq |\{w: vw \in F\} \setminus C|$, or equivalently $|\incsol{F'}{v}| \leq |\incsol{F}{v}|$, and that for at least one vertex of $C$ the inequality is sharp.

Consider any $v \in C$. If $\model(\Ibeg{v}) > p$ we have $\incsol{F'}{v} = \incsol{F}{v}$, so there is nothing to show.
If $\model(\Ibeg{v}) \leq p < \model(\Iend{v})$ then, while constructing $\model'$, we did not move $\Iend{v}$ while we moved $\Ibeg{v}$ to the right,
thus $\incsol{F'}{v} \subseteq \incsol{F}{v}$.
Moreover, as $p$ is the leftmost position with $f_C(p) < f(\Sbeg{C}{\model}-1)$, there exists $x \in V(G) \setminus C$ such that $\model(\Iend{x}) = p$.
We have $vx \in F \setminus F'$ and, consequently, $\incsol{F'}{v} \subsetneq \incsol{F}{v}$. Note that there is at least one vertex
that falls into the currently considered case by the connectivity of $C$.

We are left with the case $\model(\Iend{v}) \leq p$. However, now
$$(N_G(v) \setminus C) \uplus (\{w: vw \in F'\} \setminus C) = \Isec_\model(p) \setminus C,$$
whereas 
$$(N_G(v) \setminus C) \uplus (\{w: vw \in F \} \setminus C) \supseteq \Isec_\model(\Ibeg{v}) \setminus C;$$
here, $\uplus$ denotes a disjoint union of sets. The lemma follows from the definition of the position $p$:
$$|\Isec_\model(p) \setminus C| = f_C(p) < f_C(\model(\Ibeg{v})) = |\Isec_\model(\Ibeg{v}) \setminus C|.$$
\end{proof}
\begin{lemma}\label{lem:lr:step-inside}
For every $C \in \Cfam$ there exists an index $q$, $\Sbeg{C}{\model} \leq q < \Send{C}{\model}$,
such that $f_C(q) > f(\Sbeg{C}{\model}-1)$.
\end{lemma}
\begin{proof}
By Lemma~\ref{lem:lr:no-hole-inside} it suffices to prove that $f_C$ is not constantly to equal $f(\Sbeg{C}{\model}-1)=f_C(\Sbeg{C}{\model}-1)$ for arguments
between $\Sbeg{C}{\model}$ (inclusive) and $\Send{C}{\model}$ (exclusive). However, by the definition of $\Cfam$, there exists a starting endpoint
$\event \in \events{K}$ occupied by $C$. For such $\event$ we have $f_C(\model(\event)) \neq f_C(\model(\event)-1)$ and the lemma follows.
\end{proof}

\begin{figure}[tb]
\centering
\includegraphics{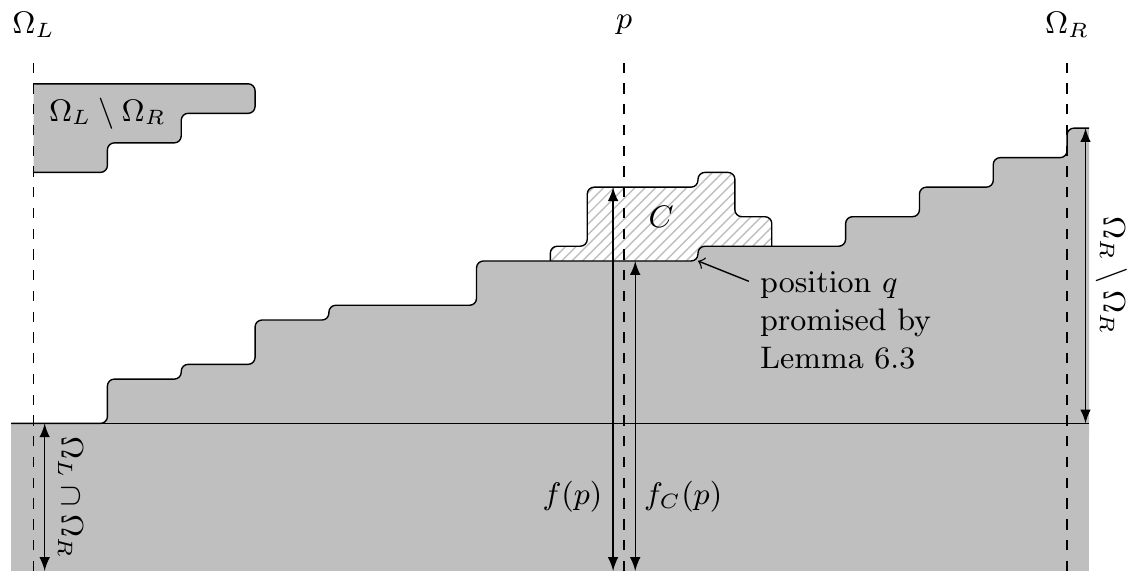}
\caption{Notation in proof of Theorem~\ref{thm:left-right}.}
\label{fig:left-right}
\end{figure}

\begin{lemma}\label{lem:lr:no-hole-after}
For every $C \in \Cfam$ and every position $p$ such that $\Send{C}{\model} \leq p < p_R$, we have $f(p) > f(\Sbeg{C}{\model}-1)$.
\end{lemma}
\begin{proof}
By contradiction, assume there exists such position $p$ with $\Send{C}{\model} \leq p < p_R$ and $f(p) \leq f(\Sbeg{C}{\model}-1)$.
Consider a model $\model'$ constructed from $\model$ by taking all events of $\events{C}$ and putting them
(without changing their internal order) between former positions $p$ and $p+1$.
As $N_G(C) \subseteq K \cup (\Isec_L \cap \Isec_R) \subseteq \Isec_R$, this is an interval model of $G+F'$ for some completion $F'$ of $G$.
Again, we claim that $|F'| < |F|$.

Note that any $e \in F \triangle F'$ connects $C$ with $V(G) \setminus C$.
Thus, it suffices to show that for any $v \in C$ we have $|\{w: vw \in F'\} \setminus C| \leq |\{w: vw \in F\} \setminus C|$ and for at least one vertex of $C$ the inequality is sharp.

Consider any $v \in C$. We have
$$(N_G(v) \setminus C) \uplus (\{w: vw \in F'\} \setminus C) = \Isec_\model(p),$$
whereas for any position $q$ such that $\model(\Ibeg{v}) \leq q < \model(\Iend{v})$ we have
$$(N_G(v) \setminus C) \uplus (\{w: vw \in F \} \setminus C) \supseteq \Isec_\model(q) \setminus C.$$
By the definition of the position $p$ and Lemma~\ref{lem:lr:no-hole-inside} we have
$$|\Isec_\model(p)| = f(p) \leq f(\Sbeg{C}{\model}-1) \leq f_C(q) = |\Isec_\model(q) \setminus C|.$$
Hence $|\incsol{F'}{v}| \leq |\incsol{F}{v}|$.

Consider now a position $q$ given by Lemma~\ref{lem:lr:step-inside}. By the connectivity of $C$, there exists $v \in C$ such that
$\model(\Ibeg{v}) \leq q < \model(\Iend{v})$. For this position we have $f(\Sbeg{C}{\model}-1) < f_C(q)$ and thus $|\incsol{F'}{v}| < |\incsol{F}{v}|$.
\end{proof}
Concluding, we obtain the following corollary.
\begin{corollary}\label{cor:lr:higher}
For any $C \in \Cfam$ and any position $\Sbeg{C}{\model} \leq p < p_R$ we have
$f(p) > f(\Sbeg{C}{\model}-1)$.
\end{corollary}
\begin{proof}
For $p < \Send{C}{\model}$ the claim follows from Lemma~\ref{lem:lr:no-hole-inside}
as $f_C(p) < f(p)$ for every $p$ with $\Sbeg{C}{\model} \leq p < \Send{C}{\model}$.
In the remaining case of $p \geq \Send{C}{\model}$, the claim follows directly from Lemma~\ref{lem:lr:no-hole-after}.
\end{proof}

We now conclude the proof of Theorem~\ref{thm:left-right} by showing that the value of $f$ cannot change too much.
A component $C \in \Cfam$ is \emph{ending expensively} if the vertex $v \in C$ with $\model(\Iend{v}) = \Send{C}{\model}$
(i.e., $\Iend{v}$ is the last event of $\events{C}$ in the model $\model$) is an expensive vertex w.r.t. $F$, and \emph{ending cheaply} otherwise.
Note that there are at most $2\sqrt{k}$ components that end expensively.
Consider a component $C \in \Cfam$ with maximum $\Send{C}{\model}$ among components that end cheaply
(if there are none, the bound of Theorem~\ref{thm:left-right} holds trivially). Let $v \in C$ satisfy $\model(\Iend{v}) = \Send{C}{\model}$.
Note that 
$$f(\Send{C}{\model}) \leq |N_G(v) \cup \incF{v}| \leq |\Isec_L \cap \Isec_R| + |K| + \sqrt{k},$$
as $v$ is cheap.
On the other hand, for any $p_L \leq p < p_R$ we have $\Isec_L \cap \Isec_R \subseteq \Isec_\model(p)$, thus
$$f(p) \geq |\Isec_L \cap \Isec_R|.$$
By Corollary~\ref{cor:lr:higher}, there are at most
$$f(\Send{C}{\model}) - \min_{p_L \leq p < p_R} f(p) \leq |K| + \sqrt{k}$$
components of $\Cfam$ that end cheaply.
Together with at most $2\sqrt{k}$ components ending expensively, we obtain the bound of Theorem~\ref{thm:left-right}.

We remark here that one can obtain a slightly better $2\sqrt{2k}+|K|$ bound by redefining a cheap vertex to be one with at most $\sqrt{2k}$ incident edges
from the solution. However, we prefer to stick with the thresholds defined in the preliminaries for the sake of clarity of the presentation.

\section{Dynamic programming}\label{sec:dp}
In this final section we describe a dynamic programming algorithm
to solve \icname{} in $\Ohstar(k^{\Oh(\sqrt{k})})$ time.
To this end, fix an \icname{} instance $(G,k)$ and,
without loss of generality, assume that the Module Reduction Rule is not applicable to $(G,k)$.

A straightforward approach, basing on the subexponential algorithm for the \textsc{Chordal Completion} problem,
would be to enumerate all possible sections via Theorem~\ref{thm:sections} and, for each section $\Isec$, try
to deduce (or guess) which components of $G \setminus \Isec$ lie to the left and which lie to the right
to the section $\Isec$. However, if $\Isec$ is large, there may be many such components with many different
neighborhoods in $\Isec$ and, consequently, such a guessing step seems expensive.
Thus, we need to employ a more involved definition of a ``separation'' to define a subproblem for the dynamic programming.

\subsection{Worlds}

We first make use of Corollary~\ref{cor:cheap-fill-in} to observe that, for a fixed vertex $v$ that is cheap in a 
given minimal solution $F$,
we can afford classifying vertices $w \in V(G) \setminus \{v\}$ depending on whether they
are included in one of the sections at endpoints of $v$, or are incident to $v$.
\begin{definition}
A \emph{world} is a tuple $\W = (v,\Isec_L,\Isec_R,p_L,p_R,F_v)$ where 
\begin{enumerate}
\item $v \in V(G)$, $\Isec_L,\Isec_R \subseteq V(G)$, $F_v \subseteq (\{v\} \times (V \setminus \{v\})) \setminus E(G)$ and $1 \leq p_L \leq p_R \leq 2n-1$;
\item $v \in \Isec_L \cap \Isec_R$;
\item $p_R - p_L = |\Isec_L \triangle \Isec_R| + 2|N_{G+F_v}(v) \setminus (\Isec_L \cup \Isec_R)|$;
\item for any $w \in \Isec_L \cup \Isec_R$ either $w=v$ or $vw \in E(G) \cup F_v$;
\item for any connected component $C$ of $G \setminus (\Isec_L \cup \Isec_R)$ either $C \subseteq N_{G+F_v}(v)$ or $C \cap N_{G+F_v}(v) = \emptyset$; and
\item $|F_v| \leq \sqrt{k}$.
\end{enumerate}
\end{definition}
For a world $\W = (v,\Isec_L,\Isec_R,p_L,p_R,F_v)$ 
we denote (see also Figure~\ref{fig:dp:world}):
\begin{align*}
\Wv{\W} &= v & \WF{\W} &= F_v \\
\WsecL{\W} &= \Isec_L & \WsecR{\W} &= \Isec_R \\
\WpL{\W} &= p_L & \WpR{\W} &= p_R \\
\WW{\W} &= N_{G+F_v}[v] & \WI{\W} &= \WW{\W} \setminus (\Isec_L \cup \Isec_R).
\end{align*}

\begin{figure}
\centering
\includegraphics{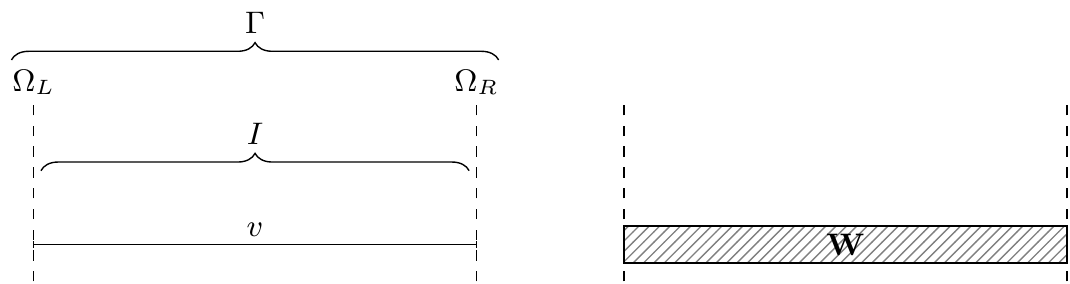}
\caption{A world with its most important elements (to the left) and its symbolic notation used in subsequent figures (to the right).}
\label{fig:dp:world}
\end{figure}

\begin{definition}
Let $F$ be a completion of $G$ and $\model$ be a model of $G+F$.
We say that the world $\W$ \emph{appears} in the model $\model$ if:
\begin{enumerate}
\item $\WF{\W} = \incF{\Wv{\W}}$,
\item $\WpL{\W} = \model(\Ibeg{\Wv{\W}})$ and $\WpR{\W} = \model(\Iend{\Wv{\W}})-1$,
\item $\WsecL{\W} = \Isec_\model(\WpL{\W})$ and $\WsecR{\W} = \Isec_\model(\WpR{\W})$.
\end{enumerate}
\end{definition}
The following observation is straightforward from the definition of a world.
\begin{lemma}\label{lem:world-exists}
For any solution $F$ to $(G,F)$ with model $\model$ of $G+F$, and any vertex $v \in V(G)$ that is cheap w.r.t. $F$,
the following tuple is in fact a world appearing in $\model$:
$$(v,\Isec_\model(\Ibeg{v}),\Isec_\model(\model(\Iend{v})-1),\model(\Ibeg{v}),\model(\Iend{v})-1,\incF{v}).$$
\end{lemma}
We denote the world defined in Lemma~\ref{lem:world-exists} by $\W(\model,v)$.

We also remark that for a world $\W$ appearing in a model $\model$, we have for every $w \notin \WsecL{\W} \cup \WsecR{\W}$
that
$$\WpL{\W} < \model(\Ibeg{w}) < \model(\Iend{w}) \leq \WpR{\W} \Leftrightarrow w\Wv{\W} \in E(G) \cup \WF{\W} \Leftrightarrow w \in \WI{\W}.$$

On the other hand, Theorem~\ref{thm:sections} and Corollary~\ref{cor:cheap-fill-in}, together with an observation
that the properties of a world can be verified in polynomial time, allow us to claim the following.
\begin{lemma}\label{lem:world-enum}
One can in $\Ohstar(k^{\Oh(\sqrt{k})})$ time enumerate a family $\worldset$ of $k^{\Oh(\sqrt{k})} n^{106}$ worlds in $G$
such that for any minimal solution $F$ to $(G,k)$, all worlds that appear in the canonical model of $G+F$
belong to $\worldset$.
\end{lemma}
We remark that the exponent $106 = 70 + 2\cdot 17 + 2$
(obtained by enumerating all possible choices $v$, $p_L$, $\Isec_L$, $\Isec_R$ and $F_v$)
is a very rough estimation. For example, one can observe that the sections $\Isec_L$ and $\Isec_R$
were already guessed in the course of guessing $F_v$ in the proof of Theorem~\ref{thm:fill-in}.
However, as the exponent in the dependency on $n$ became unholy already a few sections ago, we refrain from optimizing it.

Worlds are first basic building blocks for our states of dynamic programming: there are only relatively few interesting worlds (Lemma~\ref{lem:world-enum})
while a world $\W$ allows us to distinguish vertices that lie between the endpoints of $\Wv{\W}$ in the model we are looking for.

\subsection{Terraces}

Unfortunately, worlds are not sufficient to capture all relevant DP states.
We need a second building block, which we call a \emph{terrace}.
Intuitively, a terrace describes the behaviour either in one world (called a \emph{flat terrace})
  or in the neighborhood of a world (called a \emph{nested terrace}).

\subsubsection{Flat terraces}

\begin{definition}
A \emph{flat terrace} $\T$ consists of a single world $\W$.
\end{definition}

For a flat terrace $\T = \W$ we denote 
\begin{align*}
\TAI{\T} &= \TBI{\T} = \WI{\W} \\
\TAsecL{\T} &= \TBsecL{\T} = \WsecL{\W} \\
\TAsecR{\T} &= \TBsecR{\T} = \WsecR{\W} \\
\TApL{\T} &= \TBpL{\T} = \WpL{\W} \\
\TApR{\T} &= \TBpR{\T} = \WpR{\W}. 
\end{align*}

\subsubsection{Nested terrace}

The definition of a nested terrace is more involved. We start with a the following definition.

\begin{definition}
A \emph{nested half-terrace} $\T$ is a triple of worlds $(\Win,\Wout_1,\Wout_2)$ such that
$\Wv{\Wout_1} \neq \Wv{\Win} \neq \Wv{\Wout_2}$,
$$\WpL{\Wout_2} \leq \WpL{\Wout_1} < \WpL{\Win} \leq \WpR{\Win} < \WpR{\Wout_2} \leq \WpR{\Wout_1},$$
  and
$$|(\WsecL{\Win} \cap \WsecR{\Win}) \setminus (\WsecL{\Wout_1} \cap \WsecR{\Wout_2})| \leq 2\sqrt{k}.$$
\end{definition}

Note that we allow $\Wout_1 = \Wout_2$.
For a nested half-terrace $\T = (\Win,\Wout_1,\Wout_2)$ we denote (see also Figure~\ref{fig:dp:terrace})
\begin{align*}
\TAsecL{\T} &= \WsecL{\Wout_1} & \TBsecL{\T} &= \WsecR{\Win} \\
\TAsecR{\T} &= \WsecL{\Win} & \TBsecR{\T} &= \WsecR{\Wout_2} \\
\TApL{\T} &= \WpL{\Wout_1} & \TBpL{\T} &= \WpR{\Win} \\
\TApR{\T} &= \WpL{\Win} & \TBpR{\T} &= \WpR{\Wout_2}.
\end{align*}

\begin{figure}
\centering
\includegraphics{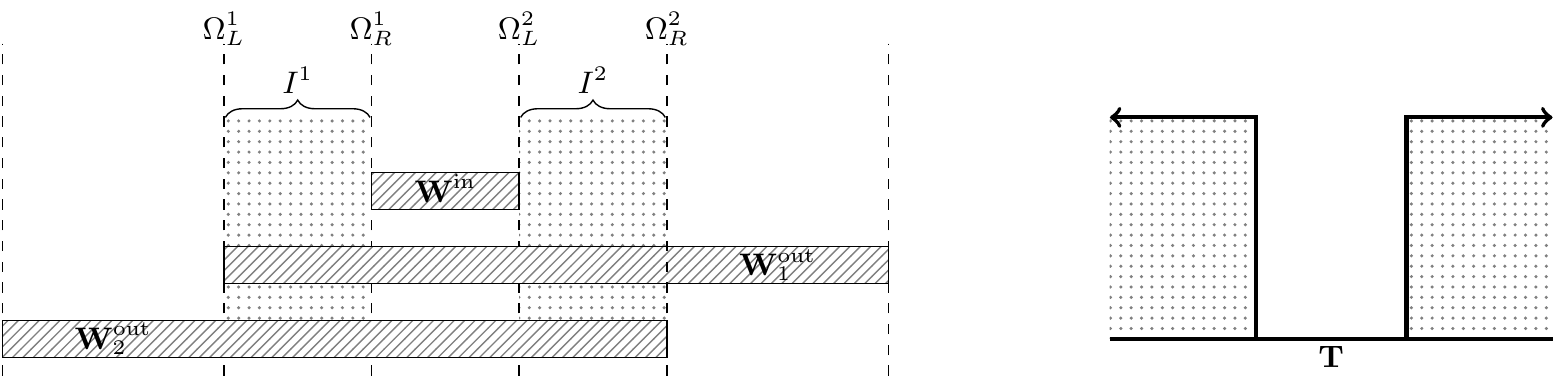}
\caption{A nested terrace with its most important notation (to the left) and its symbolic notation used in subsequent figures (to the right).
The dotted areas are the `important' areas for a terrace: the left one has borders $\Omega_L^1$, $\Omega_R^1$ and interior $I^1$,
    and the right one has borders $\Omega_L^2$, $\Omega_R^2$ and interior $I^2$.}
\label{fig:dp:terrace}
\end{figure}

However, to properly define $\TAI{\T}$ and $\TBI{\T}$ we need to enhance a nested half-terrace
$\T$ with an information, for each vertex
$v \in (\WI{\Wout_1} \cap \WI{\Wout_2}) \setminus \WW{\Win}$
whether it should lie before or after $\Wv{\Win}$ in the model $\model$ we are looking for.

\begin{definition}
A \emph{nested terrace} $\T$ is a quadruple $(\Win,\Wout_1,\Wout_2,g)$
  where $(\Win,\Wout_1,\Wout_2)$ is a nested half-terrace
and $g: (\WI{\Wout_1} \cap \WI{\Wout_2}) \setminus \WW{\Win} \to \{1,2\}$ is a function such that
whenever two vertices $x$ and $y$ in the domain of $g$ are adjacent, then $g(x) = g(y)$
(that is, $g$ is constant on each connected component in the graph induced by its domain).
\end{definition}
We may now denote for a nested terrace $(\Win,\Wout_1,\Wout_2,g)$
\begin{align*}
\TAI{\T} &= g^{-1}(1) & \TBI{\T} &= g^{-1}(2).
\end{align*}

\begin{definition}
Let $F$ be a completion of $G$ and $\model$ be a model of $G+F$.
We say that a nested terrace $\T = (\Win,\Wout_1,\Wout_2,g)$ \emph{appears} in the model $\model$ if
all $\Win,\Wout_1,\Wout_2$ appear in $\model$ and, moreover, for any
$w \in (\WI{\Wout_1} \cap \WI{\Wout_2})\setminus \WW{\Win}$
we have $\model(\Iend{w}) < \model(\Ibeg{\Wv{\Win}})$ if and only if $g(w) = 1$.
\end{definition}

A direct check from the definition shows the following.
\begin{lemma}\label{lem:terrace-exists}
Let $F$ be a completion of $G$ and $\model$ be a model of $G+F$.
Let $x \in V(G)$ be an arbitrary cheap vertex different than $\Uroot$.
Let $y_1$ be the cheap vertex with rightmost $\model(\Ibeg{y_1})$
and $y_2$ be the cheap vertex with leftmost $\model(\Iend{y_2})$
among the cheap vertices $y$ satisfying $\model(\Ibeg{y}) < \model(\Ibeg{x}) < \model(\Iend{x}) < \model(\Iend{y})$.
Then $(\W(\model,x),\W(\model,y_1),\W(\model,y_2))$ is a nested half-terrace that appears in $\model$.

Moreover, if we denote
\begin{align*}
X^1 &= \{w \in V(G): \model(\Ibeg{y_1}) < \model(\Ibeg{w}) < \model(\Iend{w}) < \model(\Ibeg{x})\} \\
X^2 &= \{w \in V(G): \model(\Iend{x}) < \model(\Ibeg{w}) < \model(\Iend{w}) < \model(\Iend{y_2})\} \\
g &= (X^1 \times \{1\}) \cup (X^2 \times \{2\})
\end{align*}
then $X^1 \cup X^2 = (\WI{\W(\model,y_1)} \cap \WI{\W(\model,y_2)}) \setminus \WW{\W(\model,x)}$ and
$(\W(\model,x),\W(\model,y_1),\W(\model,y_2),g)$ is a nested terrace that appears in $\model$.
\end{lemma}
\begin{proof}
Note that the vertices $y_1$ and $y_2$ exist, as $\Uroot$ is a candidate for both of them.
The only claim that is not straightforward is that there are at most $2\sqrt{k}$ vertices
  with $\model(\Ibeg{w}) < \model(\Ibeg{x}) < \model(\Iend{x}) < \model(\Iend{w})$
  and $\model(\Ibeg{w}) > \model(\Ibeg{y_1})$ or $\model(\Iend{w}) < \model(\Iend{y_2})$.
  However, this follows from the definition of $y_1$ and $y_2$: all such $w$ are expensive w.r.t. $F$.
\end{proof}
We denote the nested terrace defined in Lemma~\ref{lem:terrace-exists} by $\T(\model,x)$.
Note that the vertices $y_1$ and $y_2$ can be deduced from the model $\model$ and vertex $x$;
for fixed $\model$ and $x$, we denote them by $y_1(\model,x)$ and $y_2(\model,x)$.

At the end of this section we would like to include a few words about the intuition. Every terrace $\T$ has two `active' areas, $\TAI{\T}$ and $\TBI{\T}$, whose best possible completions we would like to compute. In a nested terrace these areas are in fact disjoint, and we have $\TApL{\T}\leq \TApR{\T}\leq \TBpL{\T}\leq \TBpR{\T}$. A flat terrace, however, is a degenerated case where these two areas are in fact the same. Thus, only the first and the last inequality holds, that is, we trivially have $\TApL{\T}\leq \TApR{\T}$ and $\TBpL{\T}\leq \TBpR{\T}$, but not necessarily $\TApR{\T}\leq \TBpL{\T}$ (and in fact this inequality will be most often false). Hence, when talking about an arbitrary terrace we will use {\em{only}} inequalities $\TApL{\T}\leq \TApR{\T}$ and $\TBpL{\T}\leq \TBpR{\T}$, which are true in both cases. Intuitively, in the sequel we combine pairs of terraces, and in this combination we look at only one active area of each participating terrace. Thus, we in fact have no chance of attempting using any inequality that relates the placements of two active areas of the same terrace.

\subsubsection{Enumerating terraces}

We now show that we can enumerate a relatively small family of potential terraces.
\begin{theorem}\label{thm:terrace-enum}
One can in $\Ohstar(k^{\Oh(\sqrt{k})})$ time enumerate a family $\terset$
of $k^{\Oh(\sqrt{k})} n^{318}$ terraces such that 
if $(G,k)$ is a YES-instance of \icname{}, then, for the canonical solution $F$
and the canonical model $\model$ of $G+F$, all terraces that appear in $\model$
belong to $\terset$.
\end{theorem}
\begin{proof}
Enumeration of potential flat terraces follows directly from Lemma~\ref{lem:world-enum}.
Similarly, we can enumerate
a family of $k^{\Oh(\sqrt{k})} n^{318}$ nested half-terraces
such that all nested-half terraces appearing $\model$ belong to this family.
To finish the proof we need to show that, for a fixed nested half-terrace $(\Win,\Wout_1,\Wout_2)$,
   we may enumerate
a family of $k^{\Oh(\sqrt{k})}$ potential functions $g$.
Henceforth we assume that we have a fixed nested half-terrace $(\Win,\Wout_1,\Wout_2)$ that appears in $\model$.
We describe the algorithm as a branching algorithm that generates $k^{\Oh(\sqrt{k})}$ subcases
and outputs a single function $g$ in each subcase. We argue that in the case $(\Win,\Wout_1,\Wout_2)$ indeed appears in $\model$, the correct function $g$ completing $(\Win,\Wout_1,\Wout_2)$ to a nested terrace appearing in $\model$ will be among the enumerated candidates.

Let $\Cfam_0$ be the family of these components $C \in \Ccomp{G \setminus (\WsecL{\Wout_1} \cup \WsecL{\Win} \cup \WsecR{\Win} \cup \WsecR{\Wout_2})}$
for which $C \subseteq (\WI{\Wout_1} \cap \WI{\Wout_2}) \setminus \WW{\Win}$. That is, $\bigcup \Cfam_0$ is the domain
of the function $g$ in any nested terrace $(\Win,\Wout_1,\Wout_2,g)$.

Denote $P = \WsecL{\Wout_1} \cap \WsecR{\Wout_2}$ and $K = (\WsecL{\Win} \cap \WsecR{\Win}) \setminus P$.
Note that we may assume $P \subseteq \WsecL{\Win} \cap \WsecR{\Win}$, as otherwise clearly
$(\Win,\Wout_1,\Wout_2)$
does not appear in $\model$ and we may discard such a choice of a nested half-terrace.
Moreover, by the definition of a nested half-terrace, $|K| \leq 2\sqrt{k}$.

Pick any $C \in \Cfam_0$. Note that, unless $N_G(C) \subseteq P \cup K$, we may deduce whether the vertices
of $C$ lie to the left or to the right of $\Wv{\Win}$ in the model $\model$, and, consequently, fix $g(w)$ for every $w \in C$.
Hence, in the rest of the proof we focus on the family $\Cfam \subseteq \Cfam_0$ of these components $C$
where $N_G(C) \subseteq P \cup K$.

\begin{claim}\label{cl:terrace-enum:small-C}
Providing $(\Win,\Wout_1,\Wout_2)$ appears in $\model$, it holds that $|\Cfam| = \Oh(k^2)$.
\end{claim}
\begin{proof}
If $F$ is a solution to $(G,k)$, for any $C \in \Cfam$ we have $|(C \times P) \setminus E(G)| \leq k$.
We obtain the claim by applying Lemma~\ref{lem:A-r-nei} to the set $A := P \cup K$ and threshold $r := k + |K|$.
\cqed\end{proof}
Thus, if $|\Cfam|$ exceeds the bound of Lemma~\ref{lem:A-r-nei}, we discard the choice of the nested half-terrace.
We proceed further with the assumption $|\Cfam| = \Oh(k^2)$.

Now we filter out components of $\Cfam$ that are handled by Theorem~\ref{thm:left-right}.
To this end, define $\Cfam' \subseteq \Cfam$ to be the family of components $C \in \Cfam$
such that $\Sbeg{C}{\model} < \model(\event) < \Send{C}{\model}$
for some $\event \in \events{K}$.
\begin{claim}\label{cl:terrace-enum:tiny-Cp}
$|\Cfam'| \leq 10\sqrt{k}$.
\end{claim}
\begin{proof}
The claim follows from two applications of Theorem~\ref{thm:left-right}: one to the pair of sections $\WsecL{\Wout_1}, \WsecL{\Win}$
and the set $(\WsecL{\Win} \cap \WsecR{\Win}) \setminus \WsecL{\Wout_1}$ and one
to the pair of sections $\WsecR{\Win}, \WsecR{\Wout_2}$ and the set $(\WsecL{\Win} \cap \WsecR{\Win}) \setminus \WsecR{\Wout_2}$.
\cqed\end{proof}
We guess the subfamily $\Cfam'$ and for each such $C \in \Cfam'$ we guess whether all vertices of $C$ lie to the left or to the right
of $\Wv{\Win}$ in the model $\model$. As $|\Cfam| = \Oh(k^2)$ and $|\Cfam'| \leq 10\sqrt{k}$, such a guess leads
to $k^{\Oh(\sqrt{k})}$ subcases.
We denote $\Dfam = \Cfam \setminus \Cfam'$ the family of the remaining components.

Let $\{x_L^1, x_L^2, \ldots, x_L^{s_L-1}\}$ be the set of these $x \in K$ such that $\model(\Ibeg{x}) > \WpL{\Wout_1}$, enumerated
such that $\model(\Ibeg{x_L^1}) < \model(\Ibeg{x_L^2}) < \ldots < \model(\Ibeg{x_L^{s_L-1}})$. 
Symmetrically, 
let $\{x_R^1, x_R^2, \ldots, x_R^{s_R-1}\}$ be the set of these $x \in K$ such that $\model(\Iend{x}) \leq \WpR{\Wout_2}$, enumerated
such that $\model(\Iend{x_R^1}) < \model(\Iend{x_R^2}) < \ldots < \model(\Iend{x_R^{s_R-1}})$. 
Denote $x_L^0 = \Wv{\Wout_1}$, $x_R^{s_R} = \Wv{\Wout_2}$ and $x_R^{s_L} = x_R^0 = \Wv{\Win}$.
Recall that $|K| \leq 2\sqrt{k}$; at the cost of branching into $k^{\Oh(\sqrt{k})}$ subcases, we guess the sequences $x_L^i$ and $x_R^i$.

Let us now investigate how the components of $\Dfam$ lie in the model $\model$.
\begin{claim}\label{cl:terrace-enum:consecutive-D}
For any $C \in \Dfam$, all events of $\events{C}$ are consecutive events in the model $\model$.
That is, for any $\event \notin \events{C}$ either $\model(\event) < \Sbeg{C}{\model}$ or $\model(\event) > \Send{C}{\model}$.
\end{claim}
\begin{proof}
For the sake of contradiction, assume that there exists an event $\event \notin \events{C}$ such that $\Sbeg{C}{\model} < \model(\event) < \Send{C}{\model}$.
Let $\event \in \{\Ibeg{w},\Iend{w}\}$ for some $w \notin C$.
By the definition of $\Dfam$, $w \notin K$.
Clearly, $w \notin P = \WsecL{\Wout_1} \cap \WsecR{\Wout_2}$. Hence $w \notin N_G(C)$, as $C\in \Dfam\subseteq \Cfam$.

Take now any position $p$ such that $\Sbeg{C}{\model}-1 \leq p \leq \Send{C}{\model}$ and consider a model $\model'$
created from $\model$ by taking out all events of $\events{C}$ and inserting them between former positions $p$ and $p+1$ in the original order.
As every event not in $\events{C}$ that lies between $\Sbeg{C}{\model}$ and $\Send{C}{\model}$ is an endpoint of a non-neighbor of $C$,
$\model'$ is an interval model of $G+F'$ for some completion $F'$ of $G$. Moreover, $F \triangle F'$ consists only of edges between $C$
and $V(G) \setminus C$.

Pick any $v \in C$. Clearly,
$$(N_G(v) \setminus C) \uplus (\{w: vw \in F'\} \setminus C) =  \Isec_\model(p) \setminus C.$$
On the other hand, for any position $q$ with $\model(\Ibeg{v}) \leq q < \model(\Iend{v})$ we have
$$(N_G(v) \setminus C) \uplus (\{w : vw \in F\} \setminus C) \supseteq \Isec_\model(q) \setminus C.$$
Thus, if we choose $p$ so that $|\Isec_\model(p) \setminus C|$ is minimum possible, we obtain
$|\{w: vw \in F'\} \setminus C| \leq |\{w: vw \in F\} \setminus C|$ for every $v \in C$ and, consequently, $|F'| \leq |F|$.
Consider now any $v \in C$ with $\model(\Ibeg{v}) < \model(\event) < \model(\Iend{v})$; let $(q,q')=(\model(\event)-1,\model(\event))$ if $\event$ is a closing event, and let $(q,q')=(\model(\event),\model(\event)-1)$ if $\event$ is an opening event. We infer that $|\Isec_\model(q) \setminus C|=|\Isec_\model(q') \setminus C|+1$, and hence in particular $|\Isec_\model(q) \setminus C|>|\Isec_\model(p) \setminus C|$ by the choice of $p$. We thus obtain $|\{w: vw \in F'\} \setminus C| < |\{w: vw \in F\} \setminus C|$, which implies 
$|F'| < |F|$, a contradiction with the choice of $F$.
\cqed\end{proof}
By Claim~\ref{cl:terrace-enum:consecutive-D} we infer that the components of $\Dfam$ are put into the model $\model$ in somewhat
independent and greedy manner. More precisely, define for a position $p$ a set $B(p) := \Isec_\model(p) \setminus (\bigcup \Dfam)$.
On the sets $B(p)$ we define an order as follows: $B(p) \unlhd B(q)$ if $|B(p)| < |B(q)|$ or $|B(p)| = |B(q)|$ and $B(p) \preceq B(q)$,
where $\prec$ is the order $\prec$ on $V(G)$ extended to subsets of $V(G)$ compared lexicographically.
Note that $\unlhd$ is a total order.

For any $0 \leq i < s_L$ we define $p^i_L$ to be any index $\model(\Ibeg{x_L^i}) \leq p^i_L < \model(\Ibeg{x_L^{i+1}})$
with minimum $B(p^i_L)$ according to the order $\unlhd$. Moreover, by Claim~\ref{cl:terrace-enum:consecutive-D} we can observe that for every $C\in \Dfam$, the set $B(p)$ is constant for all $p$ with $\Sbeg{C}{\model}-1 \leq p \leq \Send{C}{\model}$. Hence, we can always choose $p^i_L$ in such a way that $p^i_L<\Sbeg{C}{\model}$ or $p^i_L\geq \Send{C}{\model}$ for each $C\in \Dfam$. Consequently $\Isec_\model(p^i_L)\cap (\bigcup \Dfam)=\emptyset$ and $B(p^i_L)=\Isec_\model(p^i_L)$. Symmetrically we define $p^i_R$ for $0 \leq i < s_R$; again we can do it in such a manner that $\Isec_\model(p^i_R)\cap (\bigcup \Dfam)=\emptyset$ and $B(p^i_R)=\Isec_\model(p^i_R)$ for each $0\leq i <s_R$.

We now denote
\begin{align*}
P_L := P \cup (K \cap \Isec_\model(\WpL{\Wout_1})) = \Isec_\model(\WpL{\Wout_1})\cap \Isec_\model(\WpR{\Win}), \\
P_R := P \cup (K \cap \Isec_\model(\WpR{\Wout_2})) = \Isec_\model(\WpL{\Win})\cap \Isec_\model(\WpR{\Wout_2}).
\end{align*}
Formally, if any of the equalities above does not hold, we may discard the choice of the half-terrace. We now claim the following.
\begin{claim}\label{cl:terrace-enum:place-D}
For every $C \in \Dfam$ and for every position $p$ with $\Sbeg{C}{\model} - 1 \leq p \leq \Send{C}{\model}$,
the set $B(p)$ is the minimum (in the order $\unlhd$) set among sets $B(q)$ for $q \in P^C$, where $P^C$ is defined as:
$$P^C = \{p^i_L: N_G(C) \subseteq P_L \cup \{x^j_L: j \leq i\}\} \cup \{p^i_R: N_G(C) \subseteq P_R \cup \{x^j_R: j > i\}\}.$$
\end{claim}
\begin{proof}
As we already argued the set $B(p)$ is constant for all $p$ with $\Sbeg{C}{\model}-1 \leq p \leq \Send{C}{\model}$,
and equals $\Isec_\model(p_0) \setminus C$ for any such $p_0$, which we henceforth fix.

Assume that $C$ lies to the left of $\Wv{\Win}$ in the model $\model$. Let $0 \leq \iota < s_L$ be such that
$\model(\Ibeg{x^\iota_L}) < \Sbeg{C}{\model} < \Send{C}{\model} < \model(\Ibeg{x^{\iota+1}_L})$. Then, by the definition of
$p^\iota_L$ we have $B(p^\iota_L) \unlhd B(p_0)$. Moreover, $N_G(C) \subseteq (P \cup K) \cap B(p_0) = P_L \cup \{x^j_L: j \leq \iota\}$ and
hence $p^\iota_L \in P^C$.
The argument for $C$ lying on the right of $\Wv{\Win}$ is symmetric.
Hence, we infer that $\min_{q \in P^C} B(q) \unlhd B(p_0)$.

In the other direction, take $q_0\in P^C$ that yields the minimum set $B(q)$ with respect to $\unlhd$; note that $B(q_0)\unlhd B(p_0)$, so in particular $|B(q_0)|\leq |B(p_0)|$.
Observe that we can construct a model $\model'$ from $\model$ by taking out all events of $\events{C}$ and placing them
between position $q_0$ and $q_0+1$. By the definition of $P^C$, such a model $\model'$ is a interval model of $G+F'$ for some completion $F'$
of $G$. Observe now in $G+F$ the edges between $C$ and $V(G)\setminus C$ constitute the whole set $B(p_0)\times C$, which in particular contains all the edges between $C$ and $V(G)\setminus C$ that were present in the original graph $G$. Moreover, since $B(q_0)=\Isec_\model(q_0)$ because of $q_0\in P^C$, in $G+F'$ the edges between $C$ and $V(G)\setminus C$ constitute the whole set $B(q_0)\times C$, which again contains all the edges between $C$ and $V(G)\setminus C$ that were present $G$. Consequently $|F'|-|F|=|B(q_0)\times C|-|B(p_0)\times C|$. By the fact that $F$ is a minimum solution we infer that $|B(q_0)| \geq |B(p_0)|$, which together with the previously proven reverse inequality shows that $|B(p_0)| = |B(q_0)|$. If now it happens that $B(q_0) \prec B(p_0)$, then it is easy to observe that $F'$ is lexicographically smaller than $F$, a contradiction
to the assumption that $F$ is the canonical solution.
This concludes the proof of the claim.
\cqed\end{proof}

As the cost of $k^{\Oh(\sqrt{k})}$ additional subcases, we may guess the order $\unlhd$ restricted to the sections $B(p^i_L)$ and $B(p^i_R)$;
note that we do not want to guess neither positions $p^i_L,p^i_R$ nor sets $B(p^i_L),B(p^i_R)$ themselves, only the relative order of the sets $B(p^i_L)$ and $B(p^i_R)$
with respect to the order $\unlhd$. Observe also that some of the sets $B(p^i_L)$, $B(p^i_R)$ might be actually equal (which we also guess), but this can happen only for pairs from the opposite sides: sets $B(p^i_L)$ are pairwise different because of having different intersections with $\{x_L^i: 0\leq i\leq s_L\}$, and likewise sets $B(p^i_R)$ are pairwise different.
Once we know the order of these sets w.r.t. $\unlhd$ and the sequences $x_L^i$ and $x_R^i$,
Claim~\ref{cl:terrace-enum:place-D} allows for each component $C \in \Dfam$ to choose its place in the model $\model$ in a greedy manner.

More precisely, consider $C \in \Dfam$ and the set $P^C$ defined in Claim~\ref{cl:terrace-enum:place-D}.
Knowing the order $\unlhd$, by Claim~\ref{cl:terrace-enum:place-D} we know that $C$ is placed in the model $\model$
between $\Ibeg{x_L^i}$ and $\Ibeg{x_L^{i+1}}$ for any $0 \leq i < s_L$ such that $B(p_L^i)$ is $\unlhd$-minimum in $\{B(q): q \in P^C\}$
or
between $\Iend{x_R^i}$ and $\Iend{x_R^{i+1}}$ for any $0 \leq i < s_R$ such that $B(p_R^i)$ is $\unlhd$-minimum in $\{B(q): q \in P^C\}$.
Hence, we know whether $C$ lies to the left or to the right of $\Wv{\Win}$ in the model $\model$
unless the minimum $\{B(q) : q \in P^C\}$ is attained by some $p_L^i$ and $q_L^j$ at the same time. 

We now inspect more closely how such a situation could happen. As $B(p_L^i) = B(p_R^j)$, we have $B(p_L^i), B(p_R^j) \subseteq \WsecL{\Win} \cap \WsecR{\Win} = P \cup K$.
Hence,
$$B(p_L^i) = P_L \cup \{x_L^\ell: \ell \leq i\} = P_R \cup \{x_R^\ell: \ell > j\} = B(p_R^j).$$
In particular, for any $q \in P^C \setminus \{p_L^i, p_R^j\}$ we have
$B(p_L^i) \lhd B(q)$. Recall also that for any $0 \leq i < s_L$, we have at most
one $j = j(i)$ such that $B(p_L^i) = B(p_R^j)$.

Let $0 \leq i < s_L$ be such that $j(i)$ exists.
Let $\Dfam_i \subseteq \Dfam$ be the family of such components
$C \in \Dfam$ such that the minimum of $\{B(q): q \in P^C\}$
is attained at $X := B(p_L^i) = B(p_R^{j(i)})$.
Note that $N_{G+F}(v) \setminus C = X$ for each $v \in C$.
Hence, Lemma~\ref{lem:can:modules} applies and, as $\model$
is the canonical model of $G+F$, the components of $\Dfam_i$
are arranged according to their minimum elements in the order
$\prec$.
That is, for any $C_1, C_2 \in \Dfam_i$ such that $C_1$ lies before $\Wv{\Win}$
and $C_2$ lies after $\Wv{\Win}$ in the model $\model$, we have that the $\prec$-minimum
vertex of $C_1$ precedes the $\prec$-minimum vertex of $C_2$ in the order $\prec$.
Thus, to know which components of $\Dfam_i$
lie in the model $\model$ before $\Wv{\Win}$ it suffices
to know \emph{how many} of them lie there. As $|\Cfam| = \Oh(k^2)$
and $s_L = \Oh(\sqrt{k})$, guessing, for each $0 \leq i < s_L$ with defined $j(i)$,
how many components
of $\Dfam_i$ lie before $\Wv{\Win}$ in the model $\model$
leads to $k^{\Oh(\sqrt{k})}$ subcases.
This concludes the proof of Theorem~\ref{thm:terrace-enum}.
\end{proof}

\subsection{Dynamic programming: states and computation}

\subsubsection{DP states}

Armed with the notion of terraces, we are ready to define
the state of our dynamic programming algorithm.
\begin{definition}
A \emph{state} $\state$ is a pair of terraces $(\T_1, \T_2)$
such that $\TApL{\T_2} \leq \TBpL{\T_1} < \TApR{\T_2} \leq \TBpR{\T_1}$
and
$$\TApR{\T_2} - \TBpL{\T_1} = 2|\TBI{\T_1} \cap \TAI{\T_2}| + |\TBsecL{\T_1} \triangle \TAsecR{\T_2}|.$$
\end{definition}
We remark that each of the terraces participating in a state might be either flat or nested. Moreover, it can happen that $\T_1=\T_2$. For a state $\state = (\T_1,\T_2)$ we define (see also Figure~\ref{fig:dp:state})
\begin{align*}
\SsecL{\state} &= \TBsecL{\T_1} & \SsecR{\state} &= \TAsecR{\T_2} \\
\SpL{\state} &= \TBpL{\T_1} & \SpR{\state} &= \TApR{\T_2} \\
\SI{\state} & = \TBI{\T_1} \cap \TAI{\T_2} & \SW{\state} &= \SI{\state} \cup \SsecL{\state} \cup \SsecR{\state}
\end{align*}

\begin{figure}
\centering
\includegraphics{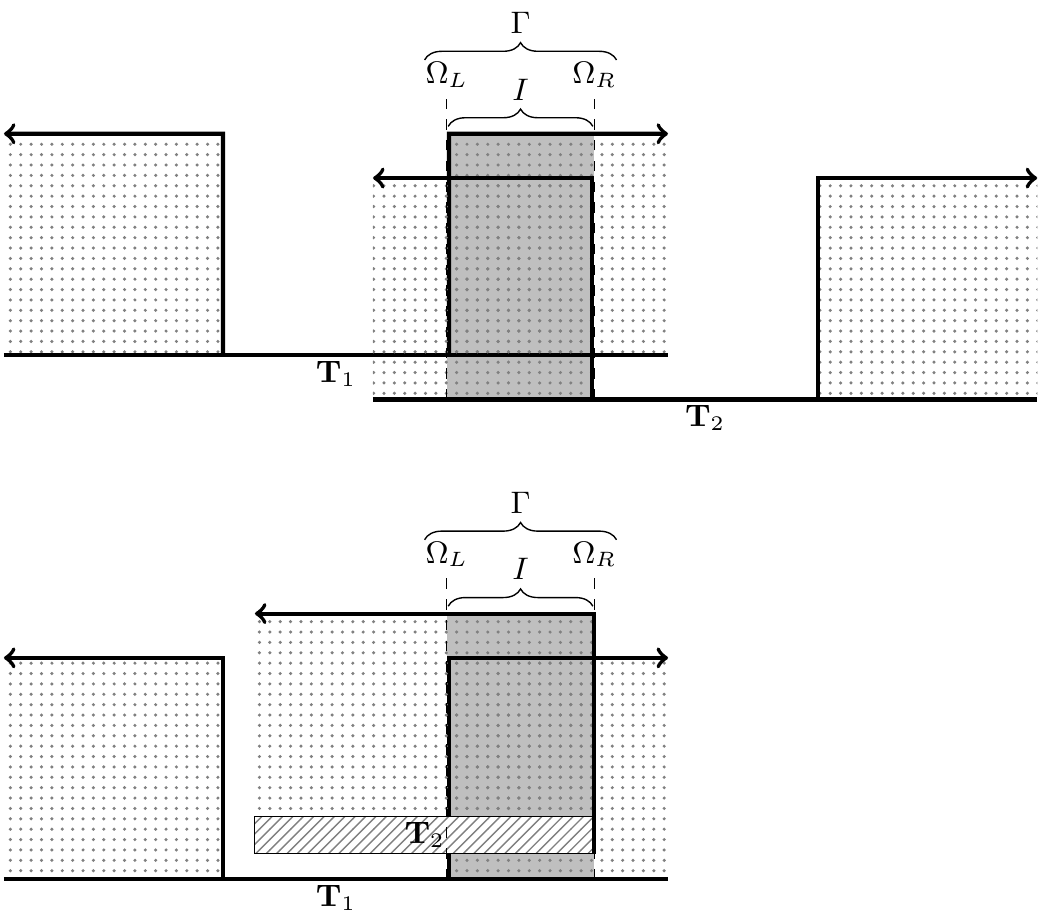}
\caption{A DP state defined by two nested terraces (above) and a nested terrace and a flat terrace (below).
 The DP state asks for the optimal way to arrange events in the gray area.
Observe that the gray area is defined as an intersection of the \emph{second} important area of the first terrace
and the \emph{first} important area of the second terrace. Furthermore, its borders are the the \emph{left} border
of the \emph{second} important area of the first terrace and the \emph{right} border of the \emph{first} important
area of the second terrace.}
\label{fig:dp:state}
\end{figure}

\begin{definition}
Let $F$ be a completion of $G$ and $\model$ be a model of $G+F$.
We say that a state $\state = (\T_1,\T_2)$ \emph{appears} in the model $\model$ if both $\T_1$ and $\T_2$ appear in $\model$.
\end{definition}

A direct check shows the following:
\begin{lemma}\label{lem:state-appear}
If $\state$ appears in a model $\model$ of a completion
$G+F$, then the events that appear
on positions $p$ satisfying
$\SpL{\state} < p \leq \SpR{\state}$ are exactly:
$$\Sevents{\state} := \events{\SI{\state}} \cup \{\Iend{v}: v \in \SsecL{\state} \setminus \SsecR{\state}\} \cup \{\Ibeg{v}: v \in \SsecR{\state} \setminus \SsecL{\state}\}.$$
\end{lemma}

Note that we have $|\Sevents{\state}|=2|\TBI{\T_1} \cap \TAI{\T_2}| + |\TBsecL{\T_1} \triangle \TAsecR{\T_2}|=\SpL{\state} - \SpR{\state}$ by the definition of a state. Observe that an immediate corollary of Theorem~\ref{thm:terrace-enum}
is an enumeration algorithm for states.
\begin{corollary}\label{cor:state-enum}
One can in $\Ohstar(k^{\Oh(\sqrt{k})})$ time enumerate a family $\stateset$
of $k^{\Oh(\sqrt{k})} n^{636}$ states such that 
if $(G,k)$ is a YES-instance of \icname{}, then, for the canonical solution $F$
and the canonical model $\model$ of $G+F$, all states that appear in $\model$
belong to $\stateset$.
\end{corollary}

\subsubsection{DP table}

Thus, a state (similarly as a world and a terrace) describes
which events of $\events{V(G)}$ lie between positions
$\SpL{\state}$ and $\SpR{\state}$. Moreover, there is only
a subexponential number of reasonable states.
However, contrary to worlds and terraces, the family of states
is rich enough to allow us to perform dynamic programming 
on a table indexed by the family $\stateset$
of Corollary~\ref{cor:state-enum}.

Formally, we say that a bijection
$\Smodel:\Sevents{\state} \to \{\SpL{\state}+1,\SpL{\state}+2,\ldots,\SpR{\state}\}$
is a \emph{completion} of state $\state$ if $\Smodel$, treated as a permutation of $\Sevents{\state}$,
preceded
with the starting events of $\SsecL{\state}$ and succeeded with the ending
events of $\SsecR{\state}$ (in any order) is an interval model of $G[\SW{\state}]+F_\Smodel$
for some completion $F_\Smodel$ of $G[\SW{\state}]$.
With a completion $\Smodel$ we associate a sequence
$\Smodel(\event_1), \Smodel(\event_2), \ldots, \Smodel(\event_{|\Sevents{\state}|})$
where $\event_1,\event_2,\ldots,\event_{|\Sevents{\state}|}$ is the ordering
of $\Sevents{\state}$ defined as follows: we first take all starting events
of $\Sevents{\state}$, sorted according to $\prec$, and then all ending
events of $\Sevents{\state}$, sorted according to reversed order $\prec$.
For two completions $\Smodel$ and $\Smodel'$ of $\state$,
we say that $\Smodel \lhd \Smodel'$ if
\begin{enumerate}
\item $|F_\Smodel| < |F_{\Smodel'}|$, or
\item $|F_\Smodel| = |F_{\Smodel'}|$ and $F_\Smodel \prec F_{\Smodel'}$, or
\item $F_\Smodel = F_{\Smodel'}$ and the sequence associated with $\Smodel$
is lexicographically smaller than the sequence associated with $\Smodel'$.
\end{enumerate}
Note that $\unlhd$ is a total order on completions of $\state$.
For a state $\state$ we define $\Smodel^\state$ to be the $\unlhd$-minimum
completion of $\state$.

In our dynamic programming algorithm we compute a value $M[\state]$ for
each $\state \in \stateset$. We aim at $M[\state] = \Smodel^\state$ at least
for each $\state$ that appears in the canonical model $\model$.
Note the following.

\begin{lemma}\label{lem:dp:model-opt}
For any $\state$ that appears in the canonical model $\model$,
we have $\Smodel^\state = \model|_{\Sevents{\state}}$.
\end{lemma}
\begin{proof}
Clearly, $\Smodel := \model|_{\Sevents{\state}}$ is a completion of $\state$
and $F_\Smodel = F \cap \binom{\SW{\state}}{2}$.
Moreover, if we consider a model $\model'$ defined as
$$\model' = \model|_{\events{V(G)} \setminus \Sevents{\state}} \cup \Smodel^\state,$$
then we obtain an interval model for $F' := (F \setminus F_\Smodel) \cup F_{\Smodel^\state}$.
Observe that:
\begin{enumerate}
\item $|F_{\Smodel^\state}| \leq |F_\Smodel|$ by the minimality of $\Smodel^\state$, whereas
if $|F_\Smodel| > |F_{\Smodel^\state}|$ then $|F'| < |F|$, contradicting the minimality of $F$;
hence $|F_\Smodel| = |F_{\Smodel^\state}|$.
\item $F_{\Smodel^\state} \preceq F_\Smodel$ by the minimality of $\Smodel^\state$, whereas
if $F_\Smodel \succ F_{\Smodel^\state}$ then $F' \prec F$, contradicting the fact that $F$ is canonical;
hence $F_\Smodel = F_{\Smodel^\state}$ and $F' = F$.
\item The sequence associated with $\Smodel^\state$ is lexicographically not larger than
the sequence associated with $\Smodel$, whereas, if it would be lexicographically strictly smaller,
then $\model'$ would be lexicographically smaller model than $\model$, contradicting the fact that $\model$
  is the canonical model of $G+F$. Hence, $\Smodel^\state = \Smodel$.
\end{enumerate}
\end{proof}

\subsubsection{DP computation}

We now proceed to the description of computation of
$M[\state]$ for $\state \in \stateset$. In the base case, if
$|\Sevents{\state}| \leq 4\sqrt{k}+4$, we find $M[\state] = \Smodel^\state$
by brute-force in $\Ohstar(k^{\Oh(\sqrt{k})})$ time by trying all possible bijections.

Consider now a state $\state$ where $|\Sevents{\state}| > 4\sqrt{k}$.
We claim that the family of sets is rich enough so that we can compute $M[\state]$
by ``gluing'' the solution of at most three substates.

More formally, to compute $M[\state]$ we iterate through all possible choices of
sequences $(\state^i)_{i=1}^s$ for $s=2,3$ where
\begin{enumerate}
\item $\SpL{\state^1} = \SpL{\state}$ and $\SsecL{\state^1} = \SsecL{\state}$,
\item $\SpR{\state^s} = \SpR{\state}$ and $\SsecR{\state^s} = \SsecR{\state}$,
\item $\SpR{\state^i} = \SpL{\state^{i+1}}$ and 
 $\SsecR{\state^i} = \SsecL{\state^{i+1}}$ for each $1 \leq i < s$,
\item $\Sevents{\state} = \biguplus_{i=1}^s \Sevents{\state^i}$,
\item $\SpR{\state^i} - \SpL{\state^i} < \SpR{\state} - \SpL{\state}$ for each $1 \leq i \leq s$.
\end{enumerate}
For each such sequence, we consider a candidate permutation $\Smodel$ defined
as a union (concatenation) of permutations $(M[\state^i])_{i=1}^s$.
As $M[\state]$ we chose the permutation $\Smodel$ which is $\unlhd$-minimum
among all considered permutations that are completions of $\state$. 
Note that, the last condition for the states $\state^i$ ensures that, 
if we compute $M[\state]$ in the order of increasing value $\SpR{\state}-\SpL{\state}$,
then in the computation we use already known values of $M[\state^i]$ for $1 \leq i \leq s$.

\begin{figure}
\centering
\includegraphics{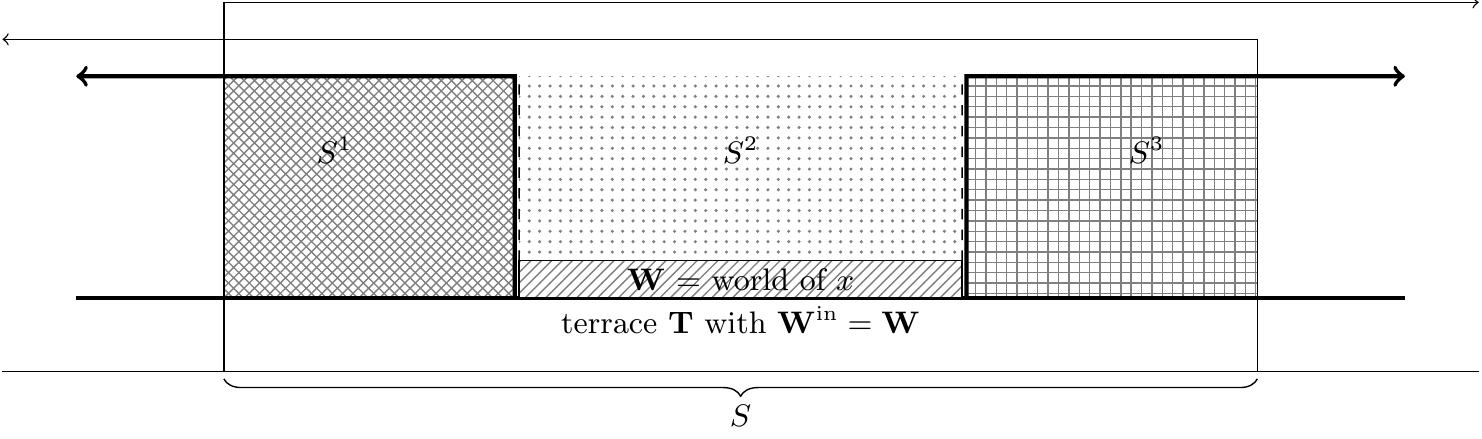}
\caption{A computation of the value for DP state $\state$ when $x \in \SI{\state}$ and we glue values from three substates.}
\label{fig:dp:states1}
\end{figure}

\begin{figure}
\centering
\includegraphics{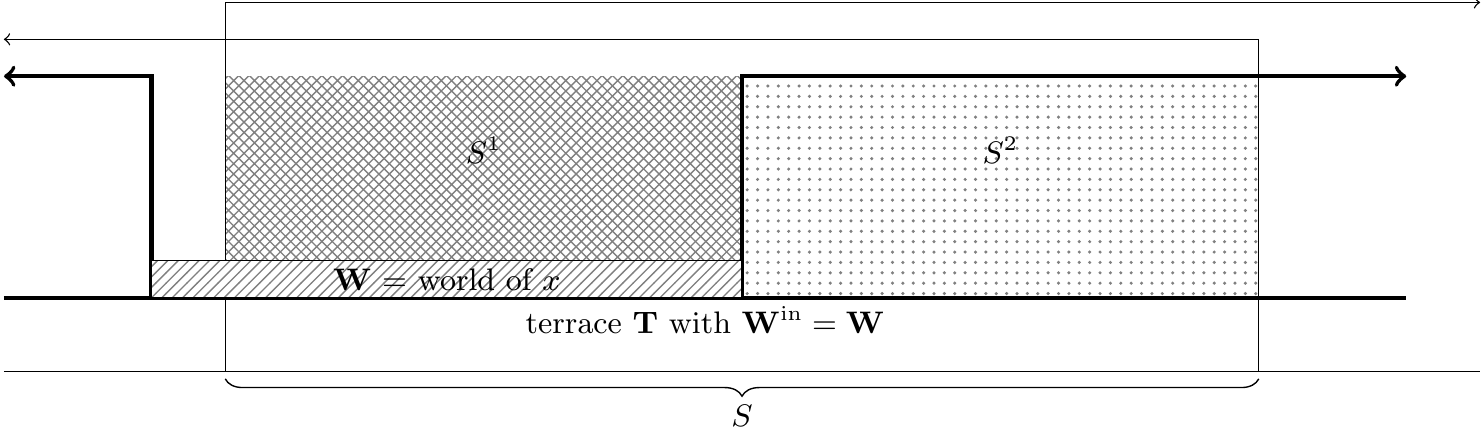}
\caption{A computation of the value for DP state $\state$ when $x \notin \SI{\state}$ and we glue values from two substates.}
\label{fig:dp:states2}
\end{figure}

If no candidate completion of $\state$ is found, we pick any permutation of $M[\state]$;
as we shall see in the next lemma, such a state $\state$ cannot appear in the canonical model $\model$.
\begin{lemma}\label{lem:dp:ok}
For any $\state$ that appears in the canonical model $\model$,
we have
$$M[\state] = \model|_{\Sevents{\state}} = \Smodel^\state.$$
\end{lemma}
\begin{proof}
The second equality is due to Lemma~\ref{lem:dp:model-opt}.
We prove that $M[\state] = \model|_{\Sevents{\state}}$ for any state $\state$ that appears
in $\model$, by induction on $|\Sevents{\state}| = \SpR{\state}-\SpL{\state}$. Note that $M[\state]$ is defined via the same minimization condition as $\Smodel^\state$ but on a smaller family of permutations, so it suffices to prove that $\model|_{\Sevents{\state}}$ is among the candidate permutations considered when computing $M[\state]$. For states where $|\Sevents{\state}| \leq 4\sqrt{k}+4$ this is clearly true, as the brute-force algorithm in fact considers all the possible candidate permutations.

Consider then $\state = (\T_1,\T_2)$ with $|\Sevents{\state}| > 4\sqrt{k}+4$. 
Observe that in this case we have at least three vertices $x \in (\SsecL{\state} \triangle \SsecR{\state}) \cup \SI{\state}$
that are cheap w.r.t. $F$.
Pick one such vertex with maximum possible value of:
\begin{equation}\label{eq:x-choice}
\min(\model(\Iend{x}),\SpR{\state}+1) - \max(\model(\Ibeg{x}),\SpL{\state}).
\end{equation}
In case of a tie, we prefer $x$ belonging to $\SI{\state}$.

We consider two cases: whether $x \in \SI{\state}$ or not.
If $x \in \SI{\state}$,
consider the flat terrace $\T_f = \W(\model,x)$ and the nested terrace $\T_n = \T(\model,x)$,
with vertices $y_1 = y_1(\model,x)$ and $y_2 = y_2(\model,x)$ (see Figure~\ref{fig:dp:states1}). 
Observe that, by the choice of $x$, we have
\begin{equation}\label{eq:x-SI}
\model(\Ibeg{y_2}) \leq \model(\Ibeg{y_1}) \leq \SpL{\state} = \TBpL{\T_1} <  \model(\Ibeg{x}) 
< \model(\Iend{x}) \leq \SpR{\state} = \TApR{\T_2} < \model(\Iend{y_2}) \leq \model(\Iend{y_1}).\end{equation}
That is, the claim that $\model(\Ibeg{y_i}) \leq \SpL{\state} < \SpR{\state} < \model(\Iend{y_i})$ for $i=1,2$
follows from \eqref{eq:x-choice} in the choice of $x$, since otherwise $y_i$ would be a better candidate for $x$.
Consider now states $\state^1 = (\T_1,\T_n)$, $\state^2 = (\T_f,\T_f)$ and $\state^3 = (\T_n,\T_2)$.
From \eqref{eq:x-SI} we infer that
$$\SpL{\state} = \SpL{\state^1} < \SpR{\state^1} = \SpL{\state^2} \leq \SpR{\state^2} = \SpL{\state^3} < \SpR{\state^3} = \SpR{\state},$$
and, consequently, the last condition for considering states $(\state^i)_{i=1}^3$ holds.
A direct check shows that these three states appear in $\model$, and the algorithm indeed
considers concatenating $M[\state^1]$, $M[\state^2]$ and $M[\state^3]$
to obtain $M[\state]$. By induction hypothesis, $M[\state^i] = \model|_{\Sevents{\state^i}}$ for $i=1,2,3$
and the inductive claim follows in this case.

In the second case, without loss of generality assume that 
$x \in \SsecL{\state} \setminus \SsecR{\state}$ (see Figure~\ref{fig:dp:states2}). Note that, by the criterion \eqref{eq:x-choice}, $x$ is such a cheap vertex with maximum $\model(\Iend{x})$.
Consider the flat terrace $\T_f = \W(\model,x)$ and the nested terrace $\T_n = \T(\model,x)$,
with vertices $y_1 = y_1(\model,x)$ and $y_2 = y_2(\model,x)$. 
Observe that, by the choice of $x$, we have
\begin{equation}\label{eq:x-left}\model(\Ibeg{y_2}) \leq \model(\Ibeg{y_1}) < \model(\Ibeg{x}) \leq \SpL{\state} = \TBpL{\T_1} 
< \model(\Iend{x}) \leq \SpR{\state} = \TApR{\T_2} < \model(\Iend{y_2}) \leq \model(\Iend{y_1}).
\end{equation}
That is, the inequality $\SpR{\state} < \model(\Iend{y_2})$ follows from the choice of rightmost possible $\model(\Iend{x})$. Consider now states $\state^1 = (\T_1,\T_f)$, $\state^2 = (\T_n,\T_2)$.
Using~\ref{eq:x-left} we observe that, unless $\model(\Iend{x}) = \SpL{\state}+1$, we have that
$$\SpL{\state} = \SpL{\state^1} < \SpR{\state^1} = \SpL{\state^2} < \SpR{\state^2} = \SpR{\state}.$$
However, if $\model(\Iend{x}) = \SpL{\state}+1$ then the value of \eqref{eq:x-choice} for the vertex $x$ equals one, and is minimum possible.
There can be at most one such $x \in \SsecL{\state} \setminus \SsecR{\state}$ and at most one such $x \in \SsecR{\state} \setminus \SsecL{\state}$.
Since there are at least three cheap vertices in $(\SsecL{\state} \triangle \SsecR{\state}) \cup \SI{\state}$, we infer that there exists one such $x' \in \SI{\state}$.
As the value of \eqref{eq:x-choice} for $x'$ is at least one, this contradicts the tie-breaking rule in the choice of $x$.

A direct check shows that both $\state^1$ and $\state^2$ appear in $\model$, and the algorithm
considers concatenating $M[\state^1]$ with $M[\state^2]$
to obtain $M[\state]$. By induction hypothesis, $M[\state^i] = \model|_{\Sevents{\state^i}}$ for $i=1,2$
and the inductive claim follows in this case as well.
This concludes the proof of Lemma~\ref{lem:dp:ok}.
\end{proof}

We now observe that the world $\W_\Uroot := \W(\model,\Uroot)$ is easy to guess:
\begin{align*}
\Wv{\W_\Uroot} &= \Uroot & \WF{\W_\Uroot} &= \emptyset \\
\WpL{\W_\Uroot} &= 1 & \WpR{\W_\Uroot} &= 2n-1 \\
\WsecL{\W_\Uroot} &= \{\Uroot\} & \WsecR{\W_\Uroot} &= \{\Uroot\}.
\end{align*}
Hence, we may proceed as follows: we compute the table $M$, read the cell $M[\state(\W_\Uroot,\W_\Uroot)]$,
and add the events $\Ibeg{\Uroot}$ and $\Iend{\Uroot}$
before and after the permutation found in this cell.
By Lemma~\ref{lem:dp:ok}, if $(G,k)$ is a YES-instance, the obtained permutation is
the canonical model for $G+F$ where $F$ is the canonical solution to $(G,k)$.
This concludes the proof of Theorem~\ref{thm:ic}.

\section{Conclusions}\label{sec:conc}
We would like to conclude our paper with two suggestions for future research.
First, in the light of our techniques the question for a polynomial kernel
for \icname{} is appealing. 
We think that the techniques developed in our work to cope with the lack of kernel,
in some sense being local kernelization arguments,
can help with obtaining an affirmative answer to this question.
The question if   \icname{}  admits a polynomial kernel  is important from  practical considerations  too.  
Although the running time of our algorithm is subexponential in $k$, so far 
   our result is mainly of theoretical importance due to  the  high degree polynomial of $n$.
 This is why  
 the most promising approach to significantly reduce the polynomial dependency on $n$
is to actually develop a polynomial kernel for \icname{}.
A polynomial kernel for \ic{} 
 would also reduce significantly the exponent in the running time by making the arguments of Section~\ref{sec:fill-in} obsolete. Needless to say, the argumentation of Sections~\ref{sec:pmc} and~\ref{sec:fill-in} could be tremendously simplified if such a polynomial kernel was at our disposal.
We remark here that it is also possible that
the very recent techniques of Cao~\cite{yixin:new}, that lead to a linear dependency
on the size of the graph in the ``forbidden subgraph'' branching algorithm,
may help decrease the dependency on the size of the graph in our algorithm.


For the second suggestion, we observe that except for the case of proper interval graphs,
the obtained subexponential parameterized algorithms for completion problems
to graph classes present in Figure~\ref{fig:diagram} run in time $k^{\Oh(\sqrt{k})} n^{\Oh(1)}$.
As an algorithm with running time bound $2^{o(\sqrt{k})} n^{\Oh(1)}$ would actually be a
$2^{o(n)}$-time algorithm, we suspect that $2^{\Oh(\sqrt{k})}$ or even $k^{\Oh(\sqrt{k})}$ may be
the best possible dependency on $k$ in the running time for these problems. Unfortunately, there is a big gap here between what we suspect and what we can prove,   even assuming the Exponential Time Hypothesis (ETH).
A natural research direction  is the quest for asymptotically tight bounds for completion problems. As  concrete open questions,  
 is there    $2^{\Omega(\sqrt{k})}n^{\Oh(1)}$ lower bound for \icname{}  under the assumption of ETH? 
 What about $2^{\Omega(\sqrt{k} \log{k})}n^{\Oh(1)}$? Or maybe  it is possible to solve   
   the completion problem to at least one of the graph classes
in Figure~\ref{fig:diagram} within  running time $2^{\Oh(\sqrt{k})} n^{\Oh(1)}$  thus shaving off the $\log k$ factor in the exponent?
%
%

\bibliographystyle{abbrv}
\bibliography{../completion}

\appendix
\section*{Appendix}\label{sec:boring}
\begin{proof}[Proof of Lemma~\ref{lem:ic-cliques-drawing}]
Without loss of generality assume that $\Omega_1$ and $\Omega_2$ are non-empty, as otherwise we may with polynomial overhead guess the first or the last event of the model.

First observe that if $G$ is disconnected, but $\Omega_1$ and $\Omega_2$ are in the same
connected component of $G$ then clearly no such interval model of $G$ exists, 
as any interval model of $G$ needs to arrange connected components of $G$ one-by-one.
Hence, assume in the rest of the proof that either $G$ is connected or $\Omega_1$ and $\Omega_2$
are contained in two different connected components of $G$.
Let $C_1$ be the connected component containing $\Omega_1$ and $C_2$ the one containing $\Omega_2$.

Consider a graph $H$ created from $G$ by adding two 3-vertex paths
$x_1,x_2,x_3$ and $y_1,y_2,y_3$ and making $x_1$ fully adjacent to $\Omega_1$
and $y_1$ fully adjacent to $\Omega_2$. We claim that there exists an interval
model of $G$ as requested in the statement of the lemma if and only if $H$
is an interval graph. Observe that such a claim would finish the proof of the lemma, as
$H$ can be constructed in linear time.

In one direction, consider the model $\model$ of $G$ as in the statement of the lemma.
Precede the ordering $\model$ with events
$\Ibeg{x_3},\Ibeg{x_2},\Iend{x_3},\Ibeg{x_1},\Iend{x_2}$
and insert the event $\Iend{x_1}$ immediately after all starting events of $\events{\Omega_1}$.
Symmetrically, succeed the ordering $\model_C$ with events
$\Ibeg{y_2},\Iend{y_1},\Ibeg{y_3},\Iend{y_2},\Iend{y_3}$
and insert the event $\Ibeg{y_1}$ immediately before all ending events of $\events{\Omega_2}$.
It is straightforward to verify that this is an interval model of the graph $H$.

In the other direction, let $\model$ be an interval model of $H$ and consider events $\Ibeg{x_2}$ and $\Iend{x_2}$. 
Observe that if $\model(\Ibeg{x_1}) < \model(\Ibeg{x_2})$ and simultaneously $\model(\Iend{x_2}) < \model(\Iend{x_1})$ (i.e., the interval of $x_1$ contains the interval of $x_2$) then there
is no place to put the endpoints of $x_3$ into the model, as $x_1x_3 \notin E(H)$ but
$x_2x_3 \in E(H)$. Consequently, either
$\model(\Ibeg{x_2}) < \model(\Ibeg{x_1}) < \model(\Iend{x_2}) < \model(\Iend{x_1})$ (case (1.i)) or
$\model(\Ibeg{x_1}) < \model(\Ibeg{x_2}) < \model(\Iend{x_1}) < \model(\Iend{x_2})$ (case (1.ii)).
Assume first that the case (1.i) happens.
As $x_1$ is adjacent to $x_2$ and to every vertex of $\Omega_1$, but no vertex of $V(G)$ is adjacent to $x_2$,
we infer that the events between $\Iend{x_2}$ and $\Iend{x_1}$ in the model $\model$
are first all starting events of $\events{\Omega_1}$ and then possibly some ending events of $\events{\Omega_1}$,
and, moreover, all other events of $\events{C_1}$
appear in $\model$ to the right of $\Iend{x_1}$. Consequently, the model $\model$,
restricted to $\events{C_1}$, starts with the starting
events of $\events{\Omega_1}$.
Observe that in the case (1.ii), i.e.,
$\model(\Ibeg{x_1}) < \model(\Ibeg{x_2}) < \model(\Iend{x_1}) < \model(\Iend{x_2})$,
we obtain the symmetric conclusion: the model $\model$, restricted to $\events{C_1}$,
ends with the ending events of $\events{\Omega_1}$.

An analogous reasoning can be made for the path $y_1,y_2,y_3$; let us denote the respective cases (2.i) and (2.ii). Consider first the case when $C_1=C_2=V(G)$ and $G$ is connected, and examine the model $\model$ restricted to $\events{C_1}=\events{C_2}=\events{V(G)}$. From our study we infer that this model starts with all the starting events of $\events{\Omega_1}$ providing that (1.i) happens, or with all the starting events of $\events{\Omega_2}$ providing that (2.i) happens. Moreover, this model ends with all the ending events of $\events{\Omega_1}$ providing that (1.ii) happens, or with all the ending events of $\events{\Omega_2}$ providing that (2.ii) happens. Observe, however, that if (1.i) and (2.i) happened simultaneously, then the first event of $\model$ restricted to $\events{V(G)}$ would be $\Ibeg{v}$ for some $v\in \Omega_1\cap\Omega_2$. In this case we would have $\Ibeg{x_1}<\Ibeg{v}<\Iend{x_1}$ and $\Ibeg{y_1}<\Ibeg{v}<\Iend{y_1}$, which means that the intervals of $x_1$ and $y_1$ would overlap, contradicting the fact that $x_1$ and $y_1$ are not adjacent in $H$. Similarly, (1.ii) and (2.ii) cannot happen simultaneously. Since either (1.i) or (1.ii) happens, and either (2.i) or (2.ii) happens, we infer that either ((1.i) and (2.ii)) happens, or ((1.ii) and (2.i)) happens. In case ((1.i) and (2.ii)) we are already done, since $\model$ restricted to $\events{V(G)}$ has exactly the desired property. In case ((1.ii) and (2.i)) it suffices to revert the model $\model$ restricted to $\events{V(G)}$.

Examine now the case when $C_1\neq C_2$. Consider model $\model'$ of $V(G)$ constructed from $\model$ by the following reshuffling of connected components of $G$: We first place the model of $C_1$, possibly reversing it if (1.ii) happened instead of (1.i). Then we arrange the models of all the connected components of $G$ other than $C_1,C_2$ in any order. Finally, we place the model of $C_2$, possibly reversing it if (2.i) happened instead of (2.ii). It is straightforward to see that this model of $G$ has the desired property.
\end{proof}

\begin{proof}[Proof of Lemma~\ref{lem:can:modules}]
Assume otherwise, and let $i$ be the smallest index such that $x_i \prec x_{i-1}$.
Denote $p = \Sbeg{C_i}{\model}$.
As $i>1$ and $N_G(v) \setminus C_j = X$ for every $1 \leq j \leq s$ and $v\in C_j$, we have that $\Isec_\model(p-1)= X$.

Consider a model $\model'$ of $G$ that is constructed as follows:
\begin{enumerate}
\item First, we take all events of $\model^{-1}(\{1,2,\ldots,p-1\}) \setminus \events{C_{i-1}}$,
  in the order as they appear in $\model$.
\item Second, we take all events of $\events{C_i}$, in the order as they appear in $\model$.
\item Third, we take all events of $\events{C_{i-1}}$, in the order as they appear in $\model$.
\item Finally, we take all events of $\model^{-1}(\{p,p+1,\ldots,2n\}) \setminus \events{C_i}$,
  in the order as they appear in $\model$.
\end{enumerate}
A direct check shows that $\model'$ is an interval model of $G$. We now claim the following: for every vertex $u \notin C_{i-1}$ we have $\model'(\Ibeg{u}) \leq \model(\Ibeg{u})$. This claim is trivial for the vertices $u\in C_i$, and for the vertices $u\notin C_{i-1}$ with $\model(\Ibeg{u})<p$. Consider then any vertex $u\notin C_{i-1}$ such that $\model(\Ibeg{v})\geq p$. Since $i>1$ and $N_G(v) \setminus C_1 = X$ for every $v\in C_1$, we infer that all the vertices of $X$ have starting events before position $p$ in $\model$, and hence $u\notin X$. Therefore $u\notin N_G(C_i)$, so in fact $\model(\Ibeg{u})>\Send{C_i}{\model}$. By the definition of $\model'$ we infer that $\model(\Ibeg{u})=\model'(\Ibeg{u})$, and the claim is proven.

Now observe that
\begin{itemize}
\item $\model'(\Ibeg{v}) \leq \model(\Ibeg{v})$ for any $v \preceq x_i$, as 
only for vertices $v \in C_{i-1}$ it is possible that $\model'(\Ibeg{v}) > \model(\Ibeg{v})$ and all vertices of
$C_{i-1}$ are at least as late as $x_{i-1} \succ x_i$ in the order $\prec$;
\item $\model'(\Ibeg{x_i}) < \model(\Ibeg{x_i})$, since $C_{i-1}$ is non-empty.
\end{itemize}
Hence, $\model$ is not the canonical model and the lemma is proven.
\end{proof}

\end{document}